\documentclass[prd,aps,amsfonts,showpacs,longbibliography,notitlepage,twocolumn,groupedaddress,nofootinbib]{revtex4-1}
\usepackage[dvipsnames]{xcolor}
\usepackage{multirow}
\usepackage{tabularx}

\usepackage{amsmath,amssymb,amsfonts,mathtools,dsfont,relsize}
\usepackage{graphicx}

\usepackage[colorlinks=true, urlcolor=violet, linkcolor=blue, citecolor=red, hyperindex=true, linktocpage=true]{hyperref}

\usepackage{amsthm}
\usepackage{qtree}
\newtheorem{thm}{Theorem}
\newtheorem{cor}[thm]{Corollary}
\newtheorem{lem}[thm]{Lemma}
\newtheorem{prop}[thm]{Proposition}
\newtheorem{prop1}{Proposition}

\usepackage{soul} 

\usepackage{enumerate}
\usepackage{stmaryrd} 

\makeatletter
\renewcommand{\p@subsection}{}
\renewcommand{\p@subsubsection}{}
\makeatother

\DeclarePairedDelimiter{\ket}{\lvert}{\rangle}

\newcommand{\ii}[0]{\mathrm{i}}

\usepackage{qcircuit}

\newcommand{\lr}[1]{\left( #1\right)}
\newcommand{\mlr}[1]{\left[ #1\right]}
\newcommand{\glr}[1]{\left\{ #1\right\}}
\newcommand{\alr}[1]{\left\langle #1\right\rangle}
\newcommand{\norm}[1]{\left\lVert#1\right\rVert}

\newcommand{\rarrow}{\quad \Rightarrow \quad}
\newcommand{\ee}{\mathrm{e}}
\newcommand{\dd}{\mathrm{d}}

\newcommand{\ddt}[1]{\frac{\mathrm{d} #1}{\mathrm{d} t}}

\newcommand{\OO}{\mathcal{O}}
\newcommand{\LL}{\mathcal{L}}
\newcommand{\cA}{\mathcal{A}}
\newcommand{\cP}{\mathcal{P}}
\newcommand{\bA}{\mathbb{A}}
\newcommand{\PP}{\mathbb{P}}
\newcommand{\QQ}{\mathbb{Q}}

\makeatletter 
    
\renewcommand\onecolumngrid{
\do@columngrid{one}{\@ne}%
\def\set@footnotewidth{\onecolumngrid}
\def\footnoterule{\kern-6pt\hrule width 1.5in\kern6pt}%
}

\renewcommand\twocolumngrid{
        \def\footnoterule{
        \dimen@\skip\footins\divide\dimen@\thr@@
        \kern-\dimen@\hrule width.5in\kern\dimen@}
        \do@columngrid{mlt}{\tw@}
}%

\makeatother

\newcommand{\upward}[2]{\multirow{1}{*}[#1 em]{#2}}
\newcolumntype{Y}{>{\centering\arraybackslash}X}
\newcolumntype{M}[1]{>{\centering\arraybackslash}m{#1}}

\newcommand{\comment}[1]{}

\begin{document}
\title{Prethermalization and the local robustness of gapped systems}

\author{Chao Yin}
\email{chao.yin@colorado.edu}
\affiliation{Department of Physics and Center for Theory of Quantum Matter, University of Colorado, Boulder, CO 80309, USA}

\author{Andrew Lucas}
\email{andrew.j.lucas@colorado.edu}
\affiliation{Department of Physics and Center for Theory of Quantum Matter, University of Colorado, Boulder, CO 80309, USA}

\begin{abstract}
We prove that prethermalization is a generic property of gapped local many-body quantum systems, subjected to small perturbations, in any spatial dimension. More precisely, let $H_0$ be a Hamiltonian, spatially local in $d$ spatial dimensions, with a gap $\Delta$ in the many-body spectrum; let $V$ be a spatially local Hamiltonian consisting of a sum of local terms, each of which is bounded by $\epsilon \ll \Delta$.  Then, the approximation that quantum dynamics is restricted to the low-energy subspace of $H_0$ is accurate, in the correlation functions of local operators, for stretched exponential time scale $\tau \sim \exp[(\Delta/\epsilon)^a]$ for any $a<1/(2d-1)$. This result does not depend on whether the perturbation closes the gap.  It significantly extends previous rigorous results on prethermalization in models where $H_0$ was frustration-free.  We infer the robustness of quantum simulation in low-energy subspaces, the existence of athermal ``scarred" correlation functions in gapped systems subject to generic perturbations, the long lifetime of false vacua in symmetry broken systems, and the robustness of quantum information in non-frustration-free gapped phases with topological order.
\end{abstract}

\date{\today}

\maketitle

\emph{Introduction.}--- Consider an exactly solved many-body quantum Hamiltonian $H_0$, assumed to be spatially local in $d$ spatial dimensions. Now, consider perturbing the Hamiltonian to $H_0+V$, where $V$ is made out of a sum of local terms, each of bounded norm $\epsilon$.  As long as we take the thermodynamic limit \emph{before} sending $\epsilon \rightarrow 0$, general lore states that a perturbation ($\epsilon>0$) has drastic qualitative  effects.  For example, the orthogonality catastrophe shows that eigenstates are extraordinarily sensitive to perturbations \cite{anderson}.  A general integrable system generally exhibits a complete rearrangement of the many-body spectrum, transitioning from Poisson ($\epsilon=0$) to Wigner-Dyson ($\epsilon \ne 0$) energy-level statistics \cite{Poilblanc_1993,rabson}. Only in special settings, such as the conjectured many-body localized phase \cite{MBL_BAA,vadim,rahul14,MBLrev_Rahul,imbrie2016many,MBLrev_RMP}, might the simple properties of many-body systems remain robust to perturbations.

With that said, it is known that in \emph{gapped} quantum many-body systems, the thermalization time scale (as measured by physical observables, i.e. local correlation functions) may be exponentially long: \begin{equation}
    t_* \sim \exp \left[ \left(\frac{\Delta}{\epsilon}\right)^a \right] \label{eq:thermalizationtime}
\end{equation}
where $\Delta$ is the gap of $H_0$, and $a>0$.  To understand why, consider the Hubbard model $H\sim \sum_{i\sim j}\epsilon  c_{\sigma i}^\dagger c_{\sigma j} + \sum_i \Delta n_{\uparrow i}n_{\downarrow i}$  \cite{doublon10,doublon12}: although two particles on the same site (called a doublon) store enormous energy and ``should" thermalize into a sea of mobile excitations by separating, there is no local perturbation that can do this!  The doublon has energy $\Delta$, but one no-doublon excitation has energy $\lesssim\epsilon$.  One must go to order $\Delta/\epsilon$ in perturbation theory to find a many-body resonance whereby a doublon can split apart while conserving energy: this implies (\ref{eq:thermalizationtime}).  Only in the last few years was this intuition put on rigorous ground \cite{deroeckprl,abanin2017rigorous}.

Existing proofs of prethermalization in the Hubbard model rely fundamentally on  peculiar aspects of the problem.  The ``unperturbed" $H_0$ consists exclusively of the repulsive potential energy -- it is a sum of local operators which: (\emph{1}) act on a single lattice site, (\emph{2}) mutually commute, and (\emph{3}) have an ``integer spectrum", such that the many-body spectrum of $H_0$ is of the form $0,\Delta,2\Delta,\ldots$.  The ``perturbation" $V$ is the kinetic (hopping) terms.  While prethermalization proofs have also been extended to Floquet and other non-Hamiltonian settings \cite{Floq_PRB,Floq_KS_16,Floq_KS_PRL,Floq_phase,Floq_power,quasiperiodic} with various experimental verifications \cite{preth_exp_Wei,preth_exp_Peng,preth_exp_BH,preth_exp_90s,preth_exp_DTC,preth_exp_Fermi}, assumptions (\emph{2}) and (\emph{3}), which lead to exact solvability, among other useful features, essentially remain.

At the same time, one may be surprised on physical grounds by this state of affairs: the intuition for prethermalization does \emph{not} rely on solvability of $H_0$, nor even a discrete spectrum in the thermodynamic limit.  In fact, it should suffice to simply say that if $\Delta$ is a many-body spectral gap of $H_0$, and any local perturbation can add energy at most $\epsilon \ll \Delta$, then one has to go to order $\Delta/\epsilon$ in perturbation theory to witness a many-body resonance wherein a system, prepared on one side of the gap of $H_0$, can ``decay" into a state on the other side. 

Indeed, this argument is consistent with a very different physical scenario: false vacuum decay.  Here, we consider a  gapped $H_0$ with degenerate ground states protected by symmetry (in the thermodynamic limit), separated from the rest of the spectrum by gap $\Delta$.  An example is an Ising ferromagnet with $\mathbb{Z}_2$ symmetry spontaneously broken in the ground state. If the perturbation $V$ explicitly breaks the symmetry, one of $H_0$'s ground states will generically have \emph{extensive energy} for $H_0+V$.  So $V$ will close the gap, and the false vacuum is one of exponentially many excited states of similar energy.  Still, path integral calculations imply the false vacuum is stable for non-perturbatively long times \cite{falseVac_Coleman}.  This is confirmed, as measured by local correlators in specific lattice models \cite{falseVac_Ising,falseVac_Hastings,falseVac_quasip,falseVac_confine,falseVac_spinchain}.  If we consider a quench at time $t=0$, since the rate per spacetime volume of nucleating a bubble of true vacuum scales as $1/t_*$, the probability a local correlator detects the true vacuum is $t^{d+1}/t_*$ in $d$ spatial dimensions, implying thermalization time $\exp((\Delta/\epsilon)^a)$.


Moreover, we expect gapped topologically-ordered phases are robust to perturbations at all times.  This could pave the way for topological quantum computing \cite{topoQC_KITAEV03,topoQC_RMP} and quantum memory \cite{topo_Hastings,QMemory_RMP} at zero temperature. However, such stability has been proven only for certain gapped Hamiltonians \cite{michalakis2013stability,frustration_free_22}. 

The gap in $H_0$ is crucial to all three stories above.  In this Letter, we prove that all three phenomena are related to a common result: when \emph{any gapped} $H_0$ is perturbed to $H_0+V$, local correlation functions are efficiently approximated by truncating to the low-energy subspace of $H_0$ for a non-perturbatively long time.  Prethermalization, captured by (\ref{eq:thermalizationtime}), is independent of the solvability of $H_0$.  This is: (\emph{1}) a substantial generalization of the theory of \cite{abanin2017rigorous}, (\emph{2}) a proof that false vacuum decay is non-perturbatively slow, and (\emph{3}) a proof of stability for gapped topological phases over non-perturbatively long times.  These diverse applications of our result are summarized in Table~\ref{Table}.

\begin{table}
	\centering
	\begin{tabularx}{0.46\textwidth}{|Y|Y|M{2cm}|}
		\hline
		scenario & assumption on $H_0$ & $t_* \ge $ ?\\
		\hline
        \upward{-0.5}{\multirow{2}{8em}{prethermalization}} & commuting; integer spectrum & \upward{-0.5}{$\exp[\Delta/\epsilon]$ \cite{abanin2017rigorous}} \\ \cline{2-3}
        & gapped & \multirow{3}{6em}{ $\exp[(\Delta/\epsilon)^a]$ (this work)} \\
        \cline{1-2}
        \upward{-0.5}{false vacuum decay} & discrete symmetry breaking (gapped) & \\
        \cline{1-2}
        \upward{-1}{\multirow{2}{8em}{stability of topological order}} & gapped &  \\ \cline{2-3}
        & frustration-free, local topological order, and gapped & \upward{-1}{$\infty$ \cite{michalakis2013stability} } \\
        \hline 
	\end{tabularx}
	\caption{Summary of \emph{rigorous} results on the robustness of gapped systems.}
	\label{Table}
\end{table}

\emph{Main Result.}--- 
Let $H_0$ and $V$ be local many-body Hamiltonians on a $d$-dimensional lattice $\Lambda$: e.g. \begin{equation}
    V = \sum_{S \subset \Lambda, S \text{ local}} V_S,
\end{equation}
where $V_S$ acts non-trivially on the degrees of freedom on sites in the geometrically local $S$, and trivially elsewhere, and $\lVert V_S\rVert \le \epsilon$.  $H_0$ has a similarly local structure, and we require the existence of a ``spectral gap" of size $\Delta$, wherein the many-body Hilbert space $\mathcal{H}$ can be decomposed into $\mathcal{H}=\mathcal{H}_<\oplus\mathcal{H}_>$, where $\mathcal{H}_<$ contains eigenvectors of eigenvalue at most $E_*$, while $\mathcal{H}_>$ contains eigenvectors of eigenvalue at least $E_*+\Delta$.
Here and below, precise definitions and proofs are contained to the Supplementary Material (SM).

For sufficiently small $\epsilon/\Delta$, there is a unitary $U$, generated by finite-time evolution with a quasi-local Hamiltonian protocol $\widetilde{H}$ with terms of strength $\epsilon$, such that \begin{equation}
  U^\dagger (H_0+V)U=  H_* + V_*, \label{eq:rotation}
\end{equation}
where $H_*$ has no matrix element connecting eigenstates of $H_0$ whose eigenvalue difference is larger than $\Delta$, while $V_*$ is a sum of local terms of strength \begin{equation}
    \lVert (V_*)_S\rVert \lesssim \epsilon \exp\left[-\left(\frac{\Delta}{\epsilon}\right)^a\right], \;\;\;\text{for any } a<\frac{1}{2d-1}.
\end{equation}
(This $a$ is likely not tight for $d>1$.)
In particular, $H_*$ is block-diagonal in $\mathcal{H}_<\oplus\mathcal{H}_>$ (i.e. protects the low/high energy subspaces).
Thus, a subspace $U\mathcal{H}_<$ of $H_0+V$ is protected for a stretched exponentially long time scale (\ref{eq:thermalizationtime}).  Since local (few-body) operators $B\approx U^\dagger B U $, there is prethermalization: dynamics in local correlation functions is efficiently truncated to the low-energy subspace of $H_0$ for non-perturbatively long times (\ref{eq:thermalizationtime}).
Moreover, $\widetilde{H}\propto -\int^\infty_{-\infty} W(t) V(t)\mathrm{d} t + {\rm O}(\epsilon^2)$ is defined order by order, where $W(t)$ is a fast-decaying function, and $V(t)=\mathrm{e}^{\mathrm{i}tH_0} V\mathrm{e}^{-\mathrm{i}tH_0}$ is dominated by terms with range $\lesssim t$ due to the Lieb-Robinson bound 
\cite{Lieb1972,ourreview}. These facts imply $\widetilde{H}$ is indeed quasi-local.

\emph{Numerical Demonstration.}--- We showcase our result with the interacting $d=1$ spin model \begin{subequations}\label{eq:H0V}\begin{align}
    H_0 &= \sum_{i=1}^{N-1}\lr{ Z_iZ_{i+1}+J_x X_i X_{i+1}}+h\sum_{i=1}^N X_i , \label{eq:H0} \\
    V &= \epsilon Z= \epsilon \sum_{i=1}^N Z_i, \label{eq:V=Z}
\end{align}\end{subequations}
where $h=0.9, J_x=0.37$. If $J_x=0$, $H_0$ is the transverse-field Ising model with two ferromagnetic ground states, separated from the excited states by a gap $2(1-h)\approx 0.2$  \cite{sachdev_book}. $J_x$ term is added to break the integrability of $H_0$, but using exact diagonalization, we find
  $H_0$ is still gapped within the ferromagnetic phase: see Fig.~\ref{fig:scar}(a).  However, this gap is extremely sensitive to $V$: the ground state $|\psi_\uparrow\rangle $ of $H_0$ with $\langle Z_i\rangle >0$ quickly merges into the excitation spectrum when $\epsilon \sim N^{-1}$. So \eqref{eq:H0V} models false vacuum decay, generalizing the literature which studies the case $J_x=0$ \cite{falseVac_Ising,falseVac_Hastings,falseVac_quasip,falseVac_confine,falseVac_spinchain}. For $\epsilon\lesssim\Delta$, we see clear non-thermal dynamical behavior in Fig.~\ref{fig:scar}(b): both if the system starts in the true false vacuum $\ket{\psi_\uparrow}$, or even the product state $\ket{\uparrow\cdots \uparrow}$. Prethermalization and slow false vacuum decay are visible in the anomalously large values of $\langle Z_i(t)\rangle$, even at $t>N/\epsilon$.  Both preparing the initial state $\ket{\uparrow\cdots \uparrow}$, and measuring $\langle Z_i(t)\rangle$, are achievable in ultracold atom experiments \cite{ion_confine21}.
  

The non-thermal behavior is also manifest when we analyze the exact eigenstates of $H_0+V$: see Fig.~\ref{fig:scar}(c).  While $V$ strongly prefers $\langle Z_i\rangle<0$, and most eigenstates near energy $E_\uparrow =\langle \psi_\uparrow|H|\psi_\uparrow\rangle $ (similar for $\ket{\uparrow\cdots \uparrow}$) obey this, there are three atypical eigenstates with $\langle Z\rangle>0$, on which $\psi_\uparrow$ has large support. Such eigenstates can be viewed as atypical ``scars'' in the finite-size spectrum. While our theorem does not say anything about the existence (or number) of such ``scars" -- we only rigorously demonstrate that $\langle Z_i(t)\rangle >0$ persists to times of at least  (\ref{eq:thermalizationtime}) -- it is intriguing that prethermalization also has clear fingerprints in the actual eigenstates of $H_0+V$. 


$H_0$ neither is commuting/frustration-free nor has integer spectrum or topological order.  Previous bounds could not prove prethermalization in this model.  Our work proves that this numerically demonstrated slow false vacuum decay persists to the thermodynamic limit, even as $V$ closes the gap of $H_0$.


\begin{figure}[t]
\centering
\includegraphics[width=.48\textwidth]{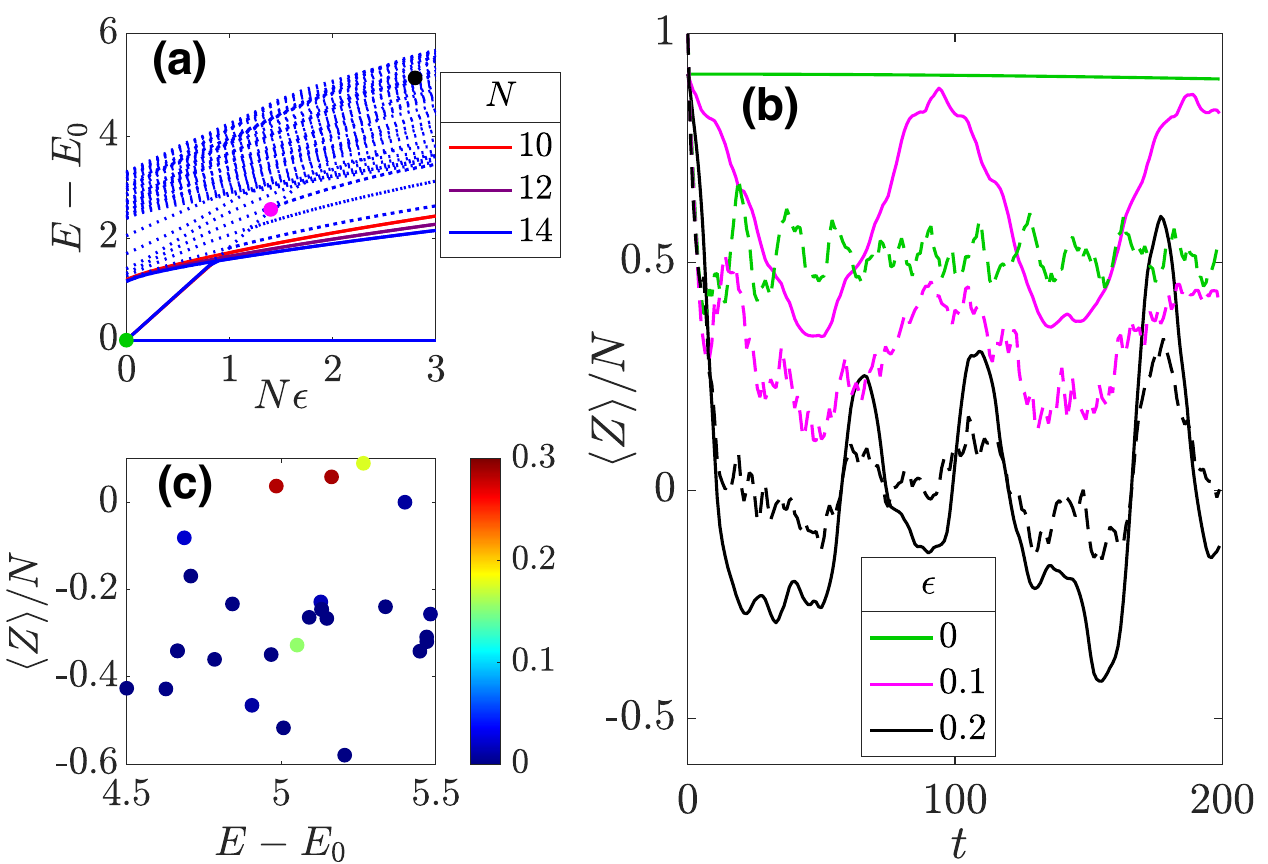}
\caption{\label{fig:scar}
(a) Spectrum of $H=H_0+V$ in (\ref{eq:H0}) and (\ref{eq:V=Z}) at $N=14$ (blue lines). The lowest $60$ eigenstates are shown. For the lowest $3$ eigenstates, data for $N=10,12$ are also shown by solid lines of different colors, indicating the gap closes at $\epsilon\sim 1/N$.  Solid dots represent $E_\uparrow = \langle \psi_\uparrow |H|\psi_\uparrow\rangle$ at $\epsilon=0,0.1,0.2$; for the latter two values, the false vacuum has been lifted above the gap. (b) Solid lines: $\langle Z\rangle/N$ with initial state $\psi_\uparrow$ for $\epsilon=0,0.1,0.2$. The green line $\epsilon=0$ has slight dynamics because $\psi_\uparrow$ is superposition of only \emph{almost} degenerate states (with finite system size).  Dashed lines: $\langle Z\rangle/N$ starting from $\ket{\uparrow\cdots \uparrow}$ instead.  Athermal behavior is observed for times $t>N/\epsilon$, even as $\Delta=0.2\sim \epsilon$. (c) Overlap of eigenstates $|E\rangle$ of $H$ with the false vacuum: $|\langle E|\psi_\uparrow\rangle|^2$,  as a function of energy around $E_\uparrow$ for $N=14,\epsilon=0.2$. The color for each eigenstate $|E\rangle$ indicates $|\langle E|\psi_\uparrow\rangle|^2$.  $|\psi_\uparrow\rangle$ is supported mainly by three atypical ``scar states" with $\langle Z\rangle>0$. }
\end{figure}



\emph{Applications of our Result.}--- An immediate consequence of our result is the generic robustness of quantum simulation of low-energy -- often \emph{constrained} -- quantum dynamics in the presence of realistic experimental perturbations.  For example, one may wish to study exotic quantum dynamics in a Hilbert space where no two adjacent spins in a 1d chain can both be up.  Yet in experiment, such a constraint can only be ``softly" implemented by penalizing adjacent up spins, e.g. via the Rydberg blockade \cite{scar_exp17}. Our result proves that for \emph{any} such model with soft constraints, the dynamics is accurately approximated by quantum dynamics in the constrained subspace of physical interest for non-perturbatively long times. This constrained dynamics often leads to quantum scars \cite{mori,scar_exp17,scar_NP18,shatter1,shatter2,scar_rev_NP,frag_confine,yoshinaga,motrunich,PRXQ_Thomas,Stephen:2022kzd,scar_rev_qp}: athermal and atypical eigenstates buried in an otherwise chaotic spectrum.  Such atypical states were found in our simulation of prethermalization.  


Our work also proves that the false vacuum has a non-perturbatively long lifetime, and that this slow decay can be accessed by experimentally accessible correlation functions and entanglement entropies, as discussed in our numerical example.
This is of some value, since the classic \cite{falseVac_Coleman} path integral calculation of false vacuum decay is quite subtle \cite{anders}, and certainly far from mathematically rigorous.
We show that thermalization (and the time scales after which eigenstate thermalization hypothesis \cite{ETH_91,ETH_94,ETH_rev16,ETH_rev18} can hold) is extraordinarily slow in all perturbations of gapped systems, starting from states in $U\mathcal{H}_<$. 

Prethermalization does not necessarily mean quasi-conservation of some global charge, as in perturbed integer-spectrum systems \cite{abanin2017rigorous}.  It is possible that this only occurs when $H_0$ has integer spectrum.  In contrast, what we describe below applies even to systems where $H_0$ contains only a single gap.  Under the assumption that the low energy spectrum of $H_0$ comes from (gapped) quasiparticle excitations, we argue in the SM that our rigorous result suggests the absence of low-energy quasiparticle proliferation \cite{falseVac_quasip} before the prethermalization time, starting from any state that has sufficiently low energy ($\ket{\psi_\uparrow}$ or $\ket{\uparrow\cdots\uparrow}$ in the numerical example). Since $H_*$ in (\ref{eq:rotation}) does not connect eigenstates of $H_0$ with energy difference larger than $\Delta$, it would not connect between states with differing numbers of low-energy quasiparticles (whose energy is at least $\Delta$).  This suggests a generalization of doublon quasi-conservation in the Hubbard model.


Most spectral gaps in many-body systems arise in gapped phases of matter, where the ground states are separated by a finite gap $\Delta$ from any excited state.  In a topological phase, there are exactly degenerate ground states \cite{QI_meet_QM}, which may serve as a logical qubit.  Our prethermalization proof implies such a qubit will remain protected in a low-dimensional subspace for extraordinarily long time scales in the presence of perturbations. This work thus provides an interesting generalization of earlier results \cite{topo_Hastings,bravyi2011short,michalakis2013stability,frustration_free_22} which proved the robustness of topological order in frustration-free Hamiltonians. 
In practice, decoherence of an experimental device may be far more dangerous than any perturbation itself to a qubit.  We cannot prove the robustness of \emph{accessible} information \cite{QMemory_RMP}: logical operators $L$ are often extensive, so even if the rotation $U$ in (\ref{eq:rotation}) is quasi-local, $\lVert U^\dagger L U - L \rVert \sim 1$ is possible.

A somewhat similar application of our result arises in SU(2)-symmetric quantum spin models, where states in the Dicke manifold (maximal $S^2$ subspace) can readily form squeezed states \cite{squeezing_rev11} of metrological value \cite{QMetro_06}.  When the Dicke manifold is protected by a spectral gap (as arises in realistic models), our work demonstrates that this protection of squeezed states is robust for exponentially long time scales in the presence of inevitable perturbations.  Of course, many practical atomic physics experiments have long-range (power-law) interactions \cite{squeezing_Rey20}, which currently lie beyond the scope of our proof.  It will be important in future work to understand whether our conclusions can be extended to this setting.  


\emph{Proof idea.}--- We now sketch the proof of our main result (details are in the SM). Although the proof structure mirrors that for Hubbard-like models \cite{abanin2017rigorous}, we need substantial technical improvements because our assumption is much weaker: we only need a single gap in $H_0$.  In what follows, $|n\rangle$ is an eigenstate of $H_0$ with eigenvalue $E_n$.

Suppose for the moment that $V$ was so small that $\lVert V\rVert \ll \Delta$, and (for convenience) suppose that $\langle m|V|n\rangle = V_{mn} \ne 0$ only if $|m\rangle$ and $|n\rangle$ are on opposite sides of the gap.   In this case we would know exactly $V$ does not close the gap, and moreover we could use first order perturbation theory to explicitly rotate the eigenstates: \begin{equation}\label{eq:n1=n+}
    |n\rangle_1 = |n\rangle + \sum_{m\ne n} \frac{V_{mn}}{E_n-E_m}|m\rangle;
\end{equation}
Moreover,
\begin{equation}
    \lVert |n\rangle_1-|n\rangle \rVert \lesssim\frac{\epsilon}{\Delta}. \label{eq:n1n}
\end{equation}
Higher orders in perturbation theory are tedious but straightforward, and (\ref{eq:n1n}) holds for the exact all-order eigenstates $|n\rangle_{H_0+V}$.  Unfortunately this series is badly behaved in the more realistic setting where each local term in $V$ is bounded by $\epsilon$ instead.  Now, $\lVert V\rVert \sim \epsilon N$ diverges with the number of lattice sites $N$. Yet this divergence should only be present in many-body states, due to the orthogonality catastrophe; local operators should be well-behaved to high order.

The operator counterpart of \eqref{eq:n1=n+} is formulated by the Schrieffer-Wolff transformations \cite{SW_66,SW_many11}, which proceed as follows.  First, we project $V$ onto terms acting within [$\mathbb{P}V$] and between [$(1-\mathbb{P})V$] the high/low-energy subspaces of $H_0$.  This can be done by defining \begin{align}\label{eq:PV=w}
    \mathbb{P}V &= \int_{-\infty}^\infty \mathrm{d}t \; w(t)\mathrm{e}^{\mathrm{i}H_0t}V\mathrm{e}^{-\mathrm{i}H_0t} \notag \\
    &= \sum_{n,m} \widehat{w}(E_n-E_m)V_{nm}|n\rangle\langle m|.
\end{align} 
Here $w(t)$ is a real-valued function with Fourier transform $\widehat{w}(\omega)$. The second line of \eqref{eq:PV=w} follows from the Heisenberg evolution $
    V(t) = \sum_{n,m} V_{nm} \mathrm{e}^{\mathrm{i}(E_n-E_m)t}|n\rangle\langle m|$.
    We don't try to calculate $|n\rangle$ or $V_{nm}$; nevertheless, the formal statement (\ref{eq:PV=w}) is valuable.
If we can find a function where $\widehat{w}(\omega)=0$ if $|\omega|\ge \Delta$, this transformation can project out the off-diagonal terms in $V$.  Such functions are known \cite{hastings2010quasi,bachmann2012automorphic}, and have asymptotic decay $w(t)\sim \mathrm{e}^{-|t|/\ln^2 |t|}$ at large $t$.  The Lieb-Robinson theorem \cite{Lieb1972,ourreview} shows that for any local operator $B_x$ supported on site $x$,  $\mathrm{e}^{\mathrm{i}H_0t}B\mathrm{e}^{-\mathrm{i}H_0t}$ is, up to exponentially small corrections, a sum of operators acting on sites within a distance $d\sim vt$ of $x$, for finite velocity $v$.   
As a result, terms in $\mathbb{P}V$ that act on sites separated by distance $r$ decay faster than $\exp[-r^{1-\delta}]$, for any $\delta >0$: this is because $w(t)$ decays a little slower than $\mathrm{e}^{-t}$, and  $B_x(t)$ has support in a ball of size $v t$, centered at $x$.


With the desired projection, we then define \begin{equation}
    D_1 = \mathbb{P}V, \;\;\;\; W_1 = (1-\mathbb{P})V,
\end{equation}
and a first order unitary rotation $U_1 = \mathrm{e}^{A_1}$ where 
\begin{equation}
[A_1,H_0] = W_1,
\end{equation}
to rotate away the off-diagonal $W_1$. $A_1$ can be found as $\mathrm{i}$ times a quasi-local Hamiltonian in a similar fashion in \eqref{eq:PV=w}.
Explicit calculation shows that the new Hamiltonian in the rotated frame \begin{equation}\label{eq:H2V2}
    U_1^\dagger (H_0+V)U_1 = H_2+V_2,
\end{equation}
is indeed block-diagonal ($H_2$ piece) for the two gapped subspaces of $H_0$ up to a $\mathrm{O}(\epsilon^2)$ piece $V_2$. Moreover, although the generator Hamiltonian $-\mathrm{i} A_1$ contains terms that decay slowly with its support, we prove $V_2$ is a sum of local terms that decay as $\exp[-r^{1-\delta}]$ with the support size $r$. To get this locality bound of $V_2$, we do require somewhat better Lieb-Robinson bounds, inspired by the equivalence class construction of \cite{chen2021operator}, than the standard ones \cite{Lieb1972}.  
\eqref{eq:H2V2} with the locality bound completes the first-order Schrieffer-Wolff transformation. In models where $H_0$ contains mutually commuting terms, this first-order process to suppress perturbations is studied in \cite{hamazaki}. Here, we not only deal with general models, but iterate this process to very high order, to obtain the non-perturbative bound (\ref{eq:thermalizationtime}).

At $k$-th order, we are given $V_k$ as the off-diagonal part in the Hamiltonian. We define $D_k = \mathbb{P}V_k$, $W_k = (1-\mathbb{P})V_k$ and $[A_k,H_0] = W_k$. Rotating the Hamiltonian by $U_k = \mathrm{e}^{A_k}$ gives the next off-diagonal $V_{k+1}$. The non-trivial aspect of this iteration is to show that $V_k$ (and $A_k,D_k,\cdots$) is not too non-local: after all, our argument for prethermalization relied on $\lVert U^\dagger B U - B\rVert \ll \lVert B\rVert$, which is only guaranteed when $U$ consists of local rotations.
As we use the same projection $\mathbb{P}$ at each step of the process, $V_k$ has increasingly large support for increasing $k$, and eventually this process becomes uncontrollable: the support of terms in $V_k$ is so large that our error $\lVert U_k^\dagger B U_k - B\rVert $ increases with $k$. 
In our proof, we can show that \begin{equation}
    \frac{\lVert V_{k+1}\rVert_{\text{local}}}{\lVert V_k\rVert_{\text{local}}} \lesssim \frac{\epsilon}{\Delta} k^{(2d-1)/(1-2\delta)}. \label{eq:Vlocal}
\end{equation}
Here $\lVert V\rVert_{\mathrm{local}}$ roughly denotes the operator norm of terms in $V$ that act non-trivially on one particular site.  From (\ref{eq:Vlocal}), we see that we must stop the Schrieffer-Wolff iterations when \begin{equation}
    k_* = \left(\frac{\Delta}{\epsilon}\right)^{a},\text{ where } a=\frac{1-2\delta}{2d-1}.
\end{equation}

Ultimately, we obtain a rotated Hamiltonian of the form (\ref{eq:rotation}), where perturbation $V_*$ is exponentially suppressed.  For any local operator $B$, we find that \begin{equation}
    \lVert U^\dagger \mathrm{e}^{\mathrm{i}Ht} B\mathrm{e}^{-\mathrm{i}Ht} U - \mathrm{e}^{\mathrm{i}H_*t} U^\dagger B U \mathrm{e}^{-\mathrm{i}H_*t} \rVert \le \epsilon t^{d+1} \mathrm{e}^{-k_*}.
\end{equation}
Namely, there exists a mild quasi-local rotation of (sums of) local operators such that the genuine dynamics of operators (and correlation functions, etc.) appear to be restricted to the low/high-energy subspaces of $H_0$ for the prethermal time scale (\ref{eq:thermalizationtime}).  This completes (the sketch of) our proof that prethermalization is a generic feature of any perturbed gapped model.


\emph{Outlook.}---In this Letter, we have proved that the prethermalization of doublons in the Hubbard model is but one manifestation of a universal phenomenon, whereby distinct sectors of a gapped Hamiltonian $H_0$ remain protected for (stretched) exponentially long times in the presence of local perturbations $V$. Prethermalization, in all measurable local correlation functions, is generic to any perturbation of a gapped system.  We thus immediately provide a rigorous proof that the false vacuum decays non-perturbatively slowly, placing less rigorous field-theoretic calculations \cite{falseVac_Coleman} on firmer footing.

Our result shows that is always reasonable to simulate quantum dynamics generated by $V$ in constrained models, so long as one studies $H_0+V$, where $H_0$'s ground state manifold is the constrained subspace of interest, and $H_0$ has a large spectral gap $\Delta$.  Even if $H_0+V$ is gapless and chaotic, the (locally rotated) ground states of $H_0$ serve as effective ``scar states" which will exhibit athermal dynamics for extraordinarily long times.  We anticipate that this observation will have practical implications for the preparation of interesting entangled states on the Dicke manifold in future atomic physics experiments, and for the ease of recovering qubits under imperfect local encoding.


\emph{Acknowledgements.}--- We thank Thomas Iadecola, Alessio Lerose and Haoqing Zhang for valuable comments. This work was supported by a Research Fellowship from the Alfred P. Sloan Foundation under Grant FG-2020-13795 (AL) and by the U.S. Air Force Office of Scientific Research under Grant FA9550-21-1-0195 (CY, AL).

\onecolumngrid

\newpage

\renewcommand{\theequation}{S\arabic{equation}}
\renewcommand{\thefigure}{S\arabic{figure}}
\setcounter{equation}{0}
\setcounter{figure}{0}

\begin{center}
{\large \textbf{Supplementary Material}}
\end{center}

\section{Preliminaries}
In this section we review a few mathematical facts, and precisely state our assumptions about the models we study.

\subsection{Models of interest}
We consider many-body quantum systems defined on a (finite) $d$-dimensional ``lattice", with vertex set $\Lambda$.  Let $\mathsf{d}: \Lambda \times \Lambda \rightarrow \mathbb{Z}^+$ denote the Manhattan distance between two vertices in $\Lambda$.  Note that $\mathsf{d}(i,j)=0$ if and only if $i=j$, while two vertices are defined to be neighbors if $\mathsf{d}(i,j)=1$.   The diameter of a subset $S\subseteq\Lambda$, denoted $\mathrm{diam}(S)$, is defined as \begin{equation}
    \mathrm{diam}(S) = \max_{i,j\in S} \mathsf{d}(i,j).
\end{equation}
Similarly, the boundary of a set $S$ is defined precisely as \begin{equation}
    \partial S = \lbrace i\in S : \text{ there exists $j\notin S$ with } \mathsf{d}(i,j)=1\rbrace.
\end{equation}
Although we will typically refer to $\Lambda$ as a lattice, we do not require it to have an translation symmetry (automorphism subgroup isomorphic to $\mathbb{Z}^d$). Instead, we require that there exists a finite constant $c_d$ such that for any $S\subseteq \Lambda$ \begin{equation}\label{eq:cd}
    |\partial S| \le c_d \cdot (1+\mathrm{diam}S)^{d-1}, \quad \mathrm{and} \quad |S| \le c_d \cdot (1+\mathrm{diam}S)^{d}.
\end{equation}
We will implicitly be interested in the regime where $|\Lambda|\rightarrow\infty$.

We associate to each vertex in $\Lambda$ a $q$-dimensional ``qudit", such that the global Hilbert space is (on a finite lattice) $\mathcal{H}=(\mathbb{C}^q)^\Lambda$.  We consider Hamiltonian 
\begin{align}\label{eq:H0+V}
    H = H_0 + V_1,
\end{align}
where $H_0$ and $V_1$ are both spatially local operators on $\Lambda$, in the sense that there exists constants $B,\kappa_0,\epsilon_0>0$ such that we may write \begin{equation}\label{eq:H0V1}
    H_0 = B \sum_{S\subseteq \Lambda} \mathrm{e}^{-\kappa_0 \mathrm{diam}(S)} H_{0,S}, \;\;\;\;\;\;\; V_1 = \epsilon_0 \sum_{S\subseteq \Lambda} \mathrm{e}^{-\kappa_0 \mathrm{diam}(S)} V_{1,S},
\end{equation}
where we assume that $\lVert H_{0,S}\rVert, \lVert V_{1,S}\rVert \le 1$, with the operator norm here the standard infninity norm (maximal singular value), and $H_{0,S}$ and $V_S$ operators that act non-trivially only on sites in $S$.  We do not require that $H_{0,S}$ acts non-trivially on all sites contained within $S$.  We assume that the spectrum of $H_0$ has a non-trivial gap $\Delta$, so that the many-body Hilbert space $\mathcal{H}$ can be decomposed into $\mathcal{H}=\mathcal{H}_<\oplus\mathcal{H}_>$, where $\mathcal{H}_<$ contains eigenvectors of eigenvalue at most $E_*$, while $\mathcal{H}_>$ contains eigenvectors of eigenvalue at least $E_*+\Delta$. 

The perturbation is weak in the sense that $\epsilon_0/\Delta$ will be small -- we postpone precise definition of how small to (\ref{eq:v1<}).  In fact, we can even slightly relax the requirements on $H_0$ and $V_1$ from above, though for practical models the above should suffice.  (Models of interest not captured by the above assumptions, such as those with power-law interactions, are not within the scope of our proof.) 

\subsection{Superimposing a simplicial lattice}
We have not specified the lattice $\Lambda$ beyond requiring it being $d$-dimensional in \eqref{eq:cd}. However, to prove our main results, more information about $\Lambda$ is needed to conveniently organize the support of operators. As a result, we fix the specific lattice by assuming that $\Lambda$ is the $d$-dimensional simplicial lattice defined as follows (see e.g. \cite{simplicial_lattice}).\footnote{The strategy of our proof also works for other lattices, e.g., the hypercubic lattice. However, there will be more terms to keep track of when decomposing an evolved operator, so we stick with the simplicial lattice with as few terms as possible. } 
Starting from an auxiliary $(d+1)$-dimensional hypercubic lattice with orthogonal basis $\mathbf{e}_1,\cdots, \mathbf{e}_{d+1}$, define a redundant basis \begin{equation}
    \mathcal{E}_p = \mathbf{e}_p - \frac{\mathbf{e}_1+\cdots+ \mathbf{e}_{d+1}}{d+1}, \quad p=1,\cdots, d+1,
\end{equation}
that satisfies \begin{equation}
    \mathcal{E}_1+ \cdots + \mathcal{E}_{d+1} = 0.
\end{equation}
All lattice points of the form $\sum_p n_p \mathbf{e}_p = \sum_p n_p \mathcal{E}_p$ with constraint $n_1+\cdots+n_{p+1}=0$, then lie on the $d$-dimensional hyperplane $\mathbf{x}_1+\cdots+\mathbf{x}_{d+1}=0$, and form the $d$-dimensional simplicial lattice. In a nutshell, each group of $d+1$ nearest sites in the simplicial lattice, serve as the vertices of the $d$-dimensional regular simplex that they form. As examples, the $2$d simplicial lattice is the triangular lattice, while the $3$d simplicial lattice is the fcc lattice made of regular tetrahedrons.

From now on, we focus on the simplicial lattice that automatically satisfies \eqref{eq:cd}, with $c_d$ determined by $d$. This is not a big restriction, since a model on an arbitrary $d$-dimensional lattice $\Lambda_0$ can be transformed into one on the simplicial lattice $\Lambda$ as follows. One can superimpose a simplicial lattice on top of the original lattice, and move all qudits to their nearest simplical lattice site (as measured by Euclidean distance in $\mathbb{R}^d$). See Fig.~\ref{fig:simplicial} for a sketch. A site in $\Lambda$ will contain at most $\mathrm{O}(1)$ qudits. If a site contains $m>1$ qudits, combine them to form an ``$mq$-dit": a single degree of freedom with $mq$-dimensional Hilbert space. Furthermore, the original Hamiltonian satisfying \eqref{eq:H0V1}, remains at least as local in the new simplicial lattice, since grouping sites together cannot increase Manhattan distance between (possibly now grouped) sites. Finally, all results that we prove for the new simplicial lattice, can be transformed back to the original $\Lambda_0$.  Any book-keeping factors that arise during this process will be $\mathrm{O}(1)$ and not affect any main results.

\begin{figure}[t]
\centering
\includegraphics[width=.6\textwidth]{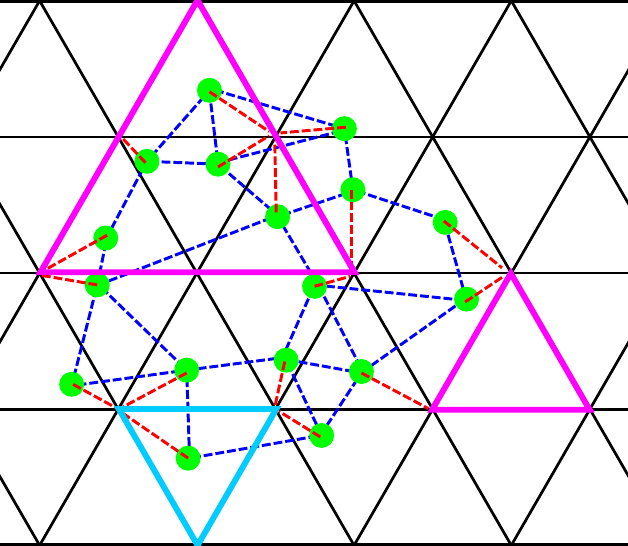}
\caption{\label{fig:simplicial}
A sketch of superimposing a simplicial lattice $\Lambda$ (black solid lines) on the original lattice $\Lambda_0$ of qudits (green dots connected by blue dashed lines) in $2$d. All qudits are moved to their nearest lattice site of the triangular lattice, as shown by red dashed lines. Some sites have $m>1$ qudits, where we combine them to a ``$mq$-dit". While some sites may have $m=0$ qudits, this is not a problem since it is equivalent to having one qudit on those sites that does not interact with the rest of the system. Given the locality of the original Hamiltonian, the ``$mq$-dits" in $\Lambda$ also only interact with their neighbors. We will consider simplices $S\subseteq \Lambda$ of fixed orientation. Namely, we say $S$ is a simplex if it is like the magenta triangles; it can not be the cyan triangle. }
\end{figure}

We say a subset $S\subseteq \Lambda$ is a simplex, if there are $d+1$ sites $i_1,\cdots,i_{d+1}\in S$, such that they are the vertices forming a $d$-dimensional regular simplex, and that $S$ is exactly all sites in $\Lambda$ contained in that regular simplex. Moreover, we only consider simplices $S$ of fixed orientation, namely there are $d+1$ fixed vectors $\mathsf{E}_1,\cdots,\mathsf{E}_{d+1}$, such that any simplex $S$ have them as the normal vectors (pointing outwards) of its $d+1$ faces. In Fig.~\ref{fig:simplicial} for example, the magenta triangles are simplices we consider, while the cyan one is not. We will use the following geometric fact: (see Fig.~\ref{fig:Y}(b) as an illustration)
\begin{prop}\label{prop:simp}
Let $S_0, S \subseteq \Lambda$ be two simplices with fixed direction, such that $S_0 \subseteq S$. Let $f_1,\cdots,f_{d+1}$ be the faces of $S$. Then \begin{equation}\label{eq:S0fp}
    \sum_{p=1}^{d+1} \mathsf{d}(S_0, f_p) = \mathrm{diam}S-\mathrm{diam}S_0.
\end{equation}
\end{prop}
\begin{proof}
Consider the process that grows the faces $f'_1,\cdots,f'_{d+1}$ of $S_0$ one by one to coincide with $S$. First, suppose the opposite vertex of $f'_1$ in $S_0$ is $i_1\in S_0$. We first grow $f'_1$ to $f_1$ in the sense that fixing $i_1$, while enlarging all the edges of $S_0$ connecting $i_1$ to reach $f_1$. Then after this first step, we get a new simplex $S_1\subseteq S$ that has $i_1$ as a vertex, and its opposite face overlapping with $f_1$. During this step, the edge enlarges by a length exactly $\mathsf{d}(S_0,f_1)$: \begin{equation}
    \mathsf{d}(S_0, f_1) = \mathrm{diam}S_1-\mathrm{diam}S_0,
\end{equation}
because the distance is measured by the Manhattan distance on the underlying simplicial lattice, not a Euclidean metric in $\mathbb{R}^d$. At the second step, we grow $S_1$ to $S_2$ in a similar way to reach $f_2$, with relation \begin{equation}
    \mathsf{d}(S_0, f_2) = \mathsf{d}(S_1, f_2) = \mathrm{diam}S_2-\mathrm{diam}S_1.
\end{equation}
Iterating this, we get an equation like above at each step up to the final $(d+1)$-th step, so that their summation produces \eqref{eq:S0fp}, because \begin{equation}
    \mathrm{diam}S-\mathrm{diam}S_0 = \mathrm{diam}S_{d+1}-\mathrm{diam}S_{d}+\mathrm{diam}S_d-\mathrm{diam}S_{d-1}+ \cdots + \mathrm{diam}S_1-\mathrm{diam}S_0,
\end{equation}
with $S=S_{d+1}$.
\end{proof}

\subsection{The $\kappa$-norm of an operator}
For an extensive operator $\OO$, there exist (many) local decompositions \begin{equation}\label{eq:O=OS}
    \OO = \sum_S \OO_S,
\end{equation}
where $S\subseteq \Lambda$ is always a simplex, and $\OO_S$ is supported inside $S$. We do not require $\OO_S$ acts nontrivially on the boundary of $S$, and the decomposition is not unique. However, there always exists an ``optimal decomposition" where we assign terms in $\mathcal{O}$ to the smallest possible simplex $S$.  We quantify this by defining the $(\alpha,\kappa)$-norm of $\OO$ as \begin{align}\label{eq:norm}
    \|\OO\|_{\alpha, \kappa}:= \inf_{\{\OO_S\}} \max_{i \in \Lambda} \sum_{S \ni i} \mathrm{e}^{\kappa (\mathrm{diam}S)^\alpha} \left\|\OO_{S}\right\|,
\end{align}
where $\inf_{\{\OO_S\}}$ is the infimum over all local decompositions. The parameters $\alpha,\kappa$ are both non-negative. 
    We will choose $\alpha$ as a fixed parameter that is close to $1$ from below, and we will just call ``$\kappa$-norm'' and use notation $\norm{\cdot}_\kappa$, with $\alpha$ being implicit.  Note that the prethermalization proof for commuting models \cite{abanin2017rigorous} uses a similar norm but with weight function $\ee^{\kappa |S|} \sim \ee^{\kappa \mathrm{diam}(S)^{d}}$. Here we will see (around Proposition \ref{prop:wt}) that we are forced to use $\alpha<1$ for general non-commuting $H_0$, which leads to the stretched exponential in the final bound. 
We will always assume that we can choose a ``best'' decomposition $\OO=\sum_S \OO_S$ that realizes its $\kappa$-norm: \begin{equation}
    \sum_{S \ni i} \mathrm{e}^{\kappa (\mathrm{diam}S)^\alpha} \left\|\OO_{S}\right\| = \norm{\OO}_\kappa.
\end{equation}
Strictly speaking, this should be viewed as choosing a decomposition that is $\delta$-close to $\norm{\OO}_\kappa$ (which is provably possible), and taking $\delta\rightarrow 0$ in the end, a mathematical annoyance that does not affect the structure of the proofs that follow.

The perturbation $V_1$ is weak in the sense that \begin{equation}\label{eq:v1<}
    \epsilon := \norm{V_1}_{\alpha,\kappa_1} \ll \Delta,
\end{equation}
where $\alpha<1$, $\kappa_1$ is some order $1$ constant, while $\Delta$ is a spectral gap of $H_0$.  While formally $\Delta$ can be any energy scale, our prethermalization bound seems most profound when it corresponds to a gap, as we will then prove that there is a notion of prethermalization -- dynamics is (for long times) approximately governed by a gapped Hamiltonian, even if the true Hamiltonian is no longer gapped.

\subsection{The Lieb-Robinson bound}
Define the Liouvillian superoperator $\LL_0$ by \begin{equation}\label{eq:L0}
    \LL_0 \OO := \ii [H_0, \OO],\quad \forall \OO.
\end{equation} 
We assume $H_0$ contains local interactions, as defined by (\ref{eq:H0V1}). Then, the following Lieb-Robinson bound holds:

\begin{prop}
There exists constants $\mu', \mu, u$ such that
 \begin{equation}\label{eq:LRB}
    \norm{\mlr{\ee^{t\LL_0}\OO_S, \OO'_{S'}}} \le 2\norm{\OO_S}\norm{\OO'_{S'}}  \mu'\min\lr{|\partial S|, |\partial S'|} \ee^{\mu (u|t|-\mathsf{d}(S,S') )},
\end{equation}
for any pair of local operators $\OO_S, \OO'_{S'}$ that do not overlap: $S\cap S'=\emptyset$.
\end{prop}

This bound is slightly stronger than many commonly stated Lieb-Robinson bounds in the literature.  We do not present an explicit proof here, as it can be shown following analogous methods to those we employ later in the proof of Proposition \ref{prop:reproducing}.

\section{Main results}
Now we summarize our main results and describe formally a few applications sketched in the main text.
\subsection{Main theorem}
For a given $H_0$, we say an operator $\OO$ is $\Delta$-diagonal, if \begin{equation}\label{eq:PO=0}
    \alr{E|\OO|E'}=0,
\end{equation}
for any pair of eigenstates $|E\rangle,|E'\rangle$ of $H_0$ that has energy difference $|E-E'|\ge \Delta$.  We do not need to know what these eigenstates are explicitly to gain value from this definition:  if we know $\OO$ is $\Delta$-diagonal and $H_0$ has a many-body gap of size $\Delta$ in the spectrum, then $\OO$ is block-diagonal between $\mathcal{H}_<$ and $\mathcal{H}_>$.  With this in mind, we present our main theorem:

\begin{thm}\label{thm1}
For the Hamiltonian \eqref{eq:H0+V} defined on the $d$-dimensional simplicial lattice, suppose $H_0$ has Lieb-Robinson bound \eqref{eq:LRB} (which defines parameter $\mu$). For any $\alpha \in (0,1)$ and \begin{equation}\label{eq:kap<mu}
    \kappa_1 \le \frac{\mu}{5},
\end{equation}
there exist constants $c_*,c_V,c_A,c_D,c'_V$ determined by $\alpha,d,\mu,\mu',\kappa_1$ and the ratio $u/\Delta$, that achieve the following:

Define \begin{equation}\label{eq:k*}
    k_*= \left\lfloor c_* \lr{\Delta/\epsilon}^{\frac{2\alpha-1}{2d-1}} \right\rfloor,
\end{equation} 
and \begin{equation}\label{eq:kapk}
    \kappa_k = \frac{\kappa_1}{1+\ln k}.
\end{equation}
For any small perturbation $V_1$ with \begin{equation}\label{eq:V<Delta}
    \norm{V_1}_{\kappa_1}=\epsilon \le c_V \Delta,
\end{equation}
there exists a quasi-local unitary $U=\ee^{A_1}\cdots \ee^{A_{k_*-1}}$ with \begin{equation}\label{eq:A<2}
    \norm{A_k}_{\kappa_{k}} \le c_A \epsilon\, 2^{-k}, \quad k=1,\cdots, k_*-1,
\end{equation}
that approximately block-diagonalizes $H=H_0+V_1$: \begin{equation}\label{eq:UHU=}
    U^\dagger HU = H_0 + D_* + V_*,
\end{equation}
where $D_*$ is $\Delta$-diagonal with respect to $H_0$. (Thus, $H_0+D_*$ is also block-diagonal with respect to $\mathcal{H}_<$ and $\mathcal{H}_>$.) Furthermore, the local norms are bounded by \begin{subequations}\begin{align}\label{eq:D*}
    \norm{D_*}_{\kappa_*} &\le c_D \epsilon, \\ \norm{V_*}_{\kappa_*} &\le c'_V \epsilon\, 2^{-k_*},
\end{align}\end{subequations}
where $\kappa_*=\kappa_{k_*}$.
\end{thm}

Because of \eqref{eq:A<2}, the unitary transformation $U$ (called Schrieffer-Wolff transformation) rotates the Hilbert space slightly, in the sense that locality in the rotated frame is, at zeroth order of $\epsilon$, the same as that in the original frame.
Moreover, the above Theorem implies that in this locally rotated frame, local dynamics is to a high accuracy generated by a dressed Hamiltonian $H_0+D_*$, until the prethermal time \begin{equation}
    t_* \sim 1/\norm{V_*}_{\kappa_*}\sim \exp\mlr{(\Delta/\epsilon)^{\frac{2\alpha-1}{2d-1}}}, 
\end{equation} 
when $V_*$ starts to play a role. The optimal choice of $\alpha$ is then $\alpha\rightarrow 1$. Before $t_*$, $H_0+D_*$ preserves any gapped subspace of $H_0$ that has gap larger than $\Delta$ away from the complement spectrum.
These heuristic arguments are formalized in the following Corollary:\footnote{This is related to Theorem 3.1-3.3 in \cite{abanin2017rigorous}.}

\begin{cor}\label{cor1}
Following Theorem \ref{thm1}, there exist constants $c_U,c_{\mathrm{oper}}$ determined by $\alpha,d,\mu,\mu',\kappa_1$ and the ratio $u/\Delta$, such that following statements hold.
\begin{enumerate}
    \item \textbf{ Locality of $U$:} For any local operator $\mathcal{O}=\OO_S$ supported in a connected set $S$, $U^\dagger \OO U$ is quasi-local and close to $\OO$ in the sense that \begin{equation}\label{eq:cU}
        U^\dagger \OO U = \OO+ \sum_{r=0}^\infty \OO_r,
    \end{equation}
    where $\OO_r$ is supported in $B(S,r)=\{i\in\Lambda: \mathsf{d}(i,S)\le r \}$ [we demand each term in $\OO_r$ is \emph{not} supported in $B(S,r-1)$], and decays rapidly with $r$: \begin{equation}\label{eq:Or<S}
        \norm{\OO_r}\le c_U\epsilon |S|\norm{\OO} \ee^{-\kappa_* r^\alpha}.
    \end{equation}
    
    \item \textbf{Local operator dynamics is approximately generated by $H_0+D_*$ up to an exponentially long time $t_*$:} For any local operator $\mathcal{O}=\OO_S$,
\begin{equation}\label{eq:coper}
    \norm{\ee^{\ii t H}\mathcal{O} \ee^{-\ii t H} - U\ee^{\ii t (H_0+D_*)}U^\dagger \mathcal{O} U\ee^{-\ii t (H_0+D_*)}U^\dagger } \le c_{\mathrm{oper}} \epsilon |S|(|t|+1)^{d+1}2^{-k_*} \norm{\OO}.
\end{equation}

    \item \textbf{Gapped subspaces of $H_0$ are locally preserved up to $t_*$:} Suppose the initial density matrix $\rho_0$ is of the form \begin{equation}
        \rho_0=U \tilde{\rho}_0 U^\dagger,
    \end{equation}
    where $\tilde{\rho}_0$ is supported inside the gapped subspace $\mathcal{H}_<$ of $H_0$ that has gap $\Delta$ to the complement spectrum $\mathcal{H}_>$. Define the reduced density matrix on set $S$ after time evolution \begin{equation}\label{eq:rhoS}
        \rho_S(t) := \mathrm{Tr}_{S^c}\lr{ \ee^{-\ii t H}\rho_0 \ee^{\ii t H}},
    \end{equation} 
    where partial trace is taken on $S^c$, the complement of $S$. Further define another reduced density matrix \begin{equation}\label{eq:rhoS'}
        \rho'_S(t) := \mathrm{Tr}_{S^c} \lr{ U\ee^{-\ii t (H_0+D_*)}U^\dagger \rho_0 U\ee^{-\ii t (H_0+D_*)}U^\dagger } = \mathrm{Tr}_{S^c} \lr{ U\ee^{-\ii t (H_0+D_*)} \tilde{\rho}_0\ee^{-\ii t (H_0+D_*)}U^\dagger },
    \end{equation}
    as reference, which stays in the gapped subspace $\mathcal{H}_<$ up to rotation by $U$. Then $\rho_S(t)$ is close to $\rho'_S(t)$ in trace norm \begin{equation}\label{eq:rho-rho}
        \norm{\rho_S(t)-\rho'_S(t)}_1 \le c_{\mathrm{oper}} \epsilon |S|(|t|+1)^{d+1}2^{-k_*}.
    \end{equation}
\end{enumerate}
\end{cor}

Although Theorem \ref{thm1} applies to any $H_0$ that have (at least) a single gap in its spectrum, we do not know of examples where $H_0$ is not built out of commuting operators, and yet such a gap appears in the middle of the spectrum.  So the most physically relevant case is therefore a gap separating the ground states to excited states, i.e. $H_0$ is in a gapped phase. Note that for commuting $H_0$ like the interaction in the Hubbard model, there are an extensive number of gaps, and one would rather use the prethermalization result in \cite{abanin2017rigorous} to get a true exponential prethermalization time.

\subsection{Proof of Theorem \ref{thm1}}

\begin{proof}[Proof of Theorem \ref{thm1}]
We prove by iteration. Suppose at the $k$-th step, we have rotated $H$ by $U_{k-1} = \ee^{A_1}\cdots \ee^{A_{k-1}}$.  We write \begin{equation}
    U_{k-1}^\dagger H U_{k-1} = H_0 + D_k + V_k,
\end{equation}
where $D_k$ is $\Delta$-diagonal.
For example at $k=1$, we have $U_0=1$ and $D_1 = 0$. To go to step $(k+1)$, we further rotate \begin{align}\label{eq:Uk=k+1}
    U_{k}^\dagger H U_{k}=\ee^{-\cA_k}(U_{k-1}^\dagger H U_{k-1}) =\ee^{-\cA_k}(H_0 + D_k + V_k) = H_0 + D_{k+1} + V_{k+1},
\end{align}
where the superoperator $\cA_k$ is defined similar to \eqref{eq:L0}: \begin{equation}
    \cA_k=[A_k,\cdot],
\end{equation}
so that $U_{k}=U_{k-1} \ee^{A_k}$.
We choose $A_k$ using the following Proposition, which is proved in Section \ref{sec:filter}.

\begin{prop}\label{prop:PO<}
For a fixed $H_0$ satisfying the Lieb-Robinson bound \eqref{eq:LRB}, there exist superoperators $\PP, \bA$, defined by
\begin{align} \label{eq:PO}
    \PP \OO &= \int\limits^\infty_{-\infty} \mathrm{d}t\; w(t)\ee^{t\LL_0} \OO , \\ \label{eq:AO}
    \bA \OO &= \ii\int\limits^\infty_{-\infty} \mathrm{d}t\; W(t)\ee^{t\LL_0} \OO ,
\end{align}
with function $w(t),W(t)$ determined by $\Delta$, such that for any operator $\OO$, $\PP\OO$ is $\Delta$-diagonal, 
and \begin{equation}\label{eq:HA=O}
    [H_0, \bA\OO] + \OO = \PP \OO.
\end{equation}
Moreover, let $\kappa^\prime<\kappa\le \kappa_1$ with $\kappa_1$ satisfying \eqref{eq:kap<mu}, and define $\delta\kappa = \kappa-\kappa^\prime$.  Then
\begin{subequations}\begin{align}\label{eq:PO<}
    \norm{\PP \OO}_{\kappa'} &\le c_w \mlr{\max(1,-\ln \delta\kappa)}^{d-1} \norm{\OO}_\kappa, \\
    \norm{\bA \OO}_{\kappa'} &\le \frac{c_W}{\Delta} \mlr{\max(1,-\ln \delta\kappa)}^{d-1} \norm{\OO}_\kappa, \label{eq:AO<}
\end{align}\end{subequations}
where $c_w$ and $c_W$ only depend on $\alpha,d,\kappa_1,\mu,\mu'$ and the ratio $u/\Delta$.
\end{prop} 

We then choose \begin{equation}\label{eq:A=V}
    A_k=\bA V_k,
\end{equation}
to satisfy \begin{equation}\label{eq:HA=V}
    [H_0, A_k] + V_k = \PP V_k.
\end{equation}
As a result, we assign \begin{equation}\label{eq:Dk+1}
    D_{k+1} = D_k+ \PP V_k,
\end{equation} 
that is still $\Delta$-diagonal, and determine $V_{k+1}$ from \eqref{eq:Uk=k+1}: \begin{align}\label{eq:Vk+1}
    V_{k+1} &= \ee^{-\cA_k}(H_0 + D_k + V_k) - H_0 - (D_k+ \PP V_k) \nonumber\\ 
    &= (\ee^{-\cA_k} - 1)H_0 + (1-\PP)V_k + (\ee^{-\cA_k}-1)(D_k+V_k)  \nonumber\\
    &= \int\limits_0^1 \dd s \lr{ \ee^{-s\cA_k}-1}(\PP-1)V_k +(\ee^{-\cA_k}-1)(D_k+V_k),
\end{align}
where we have plugged \eqref{eq:Dk+1} in the first line. In the third line, we have used \begin{equation}\label{eq:ecA-1}
    \ee^{-\cA_k} - 1 = -\int\limits_0^1 \dd s \ee^{-s\cA_k} \cA_k,
\end{equation}
and \eqref{eq:HA=V}.

Now we calculate the local norms using the following Proposition, which generalizes Lemma 4.1 of \cite{abanin2017rigorous}.  The proof of these results is quite tedious and is postponed to Section \ref{sec:LO-O<}.

\begin{prop}\label{prop:AO-O}
If \begin{equation}\label{eq:cA}
    \cA=[A,\cdot],
\end{equation}
with operator $A$ satisfying \begin{equation}\label{eq:A<C}
    \norm{A}_\kappa \le \lr{\delta\kappa}^{\frac{3d+1}{\alpha}},
\end{equation} 
with $\delta\kappa<\kappa\le \kappa_1$,
then
\begin{equation}\label{eq:AO-O}
    \norm{\ee^{-\mathcal{A}} \OO -\OO}_{\kappa'} \le c_- (\delta \kappa)^{-\frac{2d-1}{\alpha}} \norm{A}_{\kappa} \norm{\OO}_\kappa,
\end{equation}
where $\kappa'=\kappa-\delta\kappa$, and $c_-$ depends on $\alpha,d$ and $\kappa_1$.
\end{prop} 

Define \begin{equation}
    a_k = \norm{A_k}_{\kappa_k}, \quad d_k = \norm{D_k}_{\kappa_k},\quad v_k = \norm{V_k}_{\kappa_k},\quad \tilde{v}_k = c_w\mlr{\max(1,-\ln \delta\kappa_{k+1})}^{d-1} v_k,
\end{equation}
where the decreasing sequence $\kappa_k$ is defined in \eqref{eq:kapk}, and \begin{align}
    &\delta\kappa_k := \kappa_{k-1}-\kappa_{k} = \frac{\kappa_1\ln[1+1/(k-1)]}{(1+\ln k)[1+\ln(k-1)]} \ge \frac{\kappa'_1}{k\ln^2k}
    \end{align}implies that \begin{align}
    &\max(1,-\ln \delta\kappa_k) \le c_\delta \ln k, \label{eq:kapk>}
\end{align}
for some constants $\kappa'_1$ and $c_\delta$ determined by $\kappa_1$. This choice \eqref{eq:kapk} of $\kappa_k$  makes $\delta\kappa_k\sim 1/k\ln^2 k$ decay about as slowly as possible while keeping $\kappa_k>0$. 
Using \eqref{eq:PO<},\eqref{eq:AO<} and \eqref{eq:AO-O}, the iterative definitions \eqref{eq:A=V},\eqref{eq:Dk+1} and \eqref{eq:Vk+1} lead to iterations of the local bounds \begin{subequations}\begin{align}\label{eq:a<v}
    a_{k-1} &\le \frac{c_W}{c_w} \frac{\tilde{v}_{k-1}}{\Delta}, \\
    d_{k} &\le d_{k-1} + \tilde{v}_{k-1}, \label{eq:dk} \\
    v_{k} &\le c_-(\delta\kappa_k)^{-\frac{2d-1}{\alpha}} a_{k-1} \mlr{\frac{1}{2}\lr{v_{k-1}+\tilde{v}_{k-1}}+ d_{k-1}+v_{k-1}} \le c_-(\delta\kappa_k)^{-\frac{2d-1}{\alpha}} a_{k-1}(d_{k-1}+ \tilde{c}_w \tilde{v}_{k-1}), \label{eq:vk}
\end{align}\end{subequations}
where we have shifted $k\rightarrow k-1$, and used $\tilde{v}_{k-1} \ge c_w v_{k-1}$ and $\tilde{c}_w=(c_w+3)/(2c_w)$.
Plugging \eqref{eq:a<v} into \eqref{eq:vk} and combining constants using \eqref{eq:kapk>}, we get the iteration for $\tilde{v}_k$ \begin{equation}\label{eq:tvk}
    \tilde{v}_k \le \frac{\tilde{c}_v}{\Delta} (k\ln^2 k)^{\frac{2d-1}{\alpha}} (\ln k)^{d-1} \tilde{v}_{k-1}(d_{k-1}+ \tilde{c}_w\tilde{v}_{k-1}) \le \frac{c_v}{\Delta} k^{\frac{2d-1}{2\alpha-1}} \tilde{v}_{k-1}(d_{k-1}+ \tilde{c}_w\tilde{v}_{k-1}),
\end{equation}
where we have replaced a power function of $\ln k$ by a power function of $k$: $k^{\frac{2d-1}{2\alpha-1}-\frac{2d-1}{\alpha}}$, with the price of adjusting the prefactor. \eqref{eq:dk} and \eqref{eq:tvk} comprises the closed iteration for $d_k$ and $\tilde{v}_k$, assuming the condition \eqref{eq:A<C} which transforms to \begin{equation}\label{eq:v<k}
    \tilde{v}_{k-1}\le c_a(k\ln^2 k)^{\frac{3d+1}{\alpha}} \Delta,
\end{equation}
with constant $c_a$.  We will later verify this condition can be achieved.

For sufficiently small $\epsilon/\Delta$, $d_k$ keeps at the order of $\epsilon=v_1$, and \eqref{eq:tvk} leads to $\tilde{v}_{k}\sim d_k \tilde{v}_{k-1}/\Delta \sim v_1^{k}/\Delta^{k-1}$. The iteration continues up to $k_{*}\sim (\Delta/v_1)^{\frac{2\alpha-1}{2d-1}}$ when the power of $k$ in \eqref{eq:tvk} dominates, and the iteration terminates. In this process \eqref{eq:v<k} is guaranteed to hold, since $\tilde{v}_k$ decays exponentially while the right hand side of \eqref{eq:v<k} decays as a power law. To make the above arguments rigorous, we assume \eqref{eq:v<k} and \begin{equation} \label{eq:v<2}
    \tilde{v}_{k} \le \tilde{v}_1/2^{k-1},
\end{equation}
holds for all steps $k=1,\cdots,k'$ before $k'+1$. Then \eqref{eq:dk} yields \begin{equation}
    d_{k'} \le c_w \sum_{k=1}^{k'-1} \tilde{v}_{k} \le c_w \tilde{v}_1 \frac{1}{1-1/2} = 2c_w \tilde{v}_1.
\end{equation}
If \eqref{eq:v<k} also holds at $k=k'+1$, then \eqref{eq:tvk} yields \begin{equation}
    \tilde{v}_{k'+1} \le \frac{c_v}{\Delta} (k'+1)^{\frac{2d-1}{2\alpha-1}} \tilde{v}_{k'}(2c_w + \tilde{c}_w) \tilde{v}_1.
\end{equation}
The right hand side is bounded by $\tilde{v}_{k'}/2$, as long as \begin{equation}
    \frac{c_v}{\Delta} (k'+1)^{\frac{2d-1}{2\alpha-1}}(2c_w + \tilde{c}_w) \tilde{v}_1 \le \frac{1}{2}
    \end{equation}
    which holds so long as \begin{equation}
    k'+1\le k_* := \left\lfloor \lr{\frac{\Delta}{c_v(2c_w + \tilde{c}_w) \tilde{v}_1}}^{\frac{2\alpha-1}{2d-1}} \right\rfloor \equiv \left\lfloor c_*\lr{\frac{\Delta}{v_1} }^{\frac{2\alpha-1}{2d-1}}  \right\rfloor,
\end{equation}
where $c_*$ is determined by $c_v, c_w$ and $\kappa_1$. Thus \eqref{eq:v<2} also holds for the next step $k'+1$, as long as $k'+1\le k_*$ and \eqref{eq:v<k} holds for $k=k'+1$. 

Finally, we verify \eqref{eq:v<k} using our assumption \eqref{eq:V<Delta} with a sufficiently small constant $c_V$: From \eqref{eq:v<2}, it suffices to prove \begin{equation}
    \tilde{v}_1 2^{2-k'} \le c_a (k'\ln^2 k')^{\frac{3d+1}{\alpha}} \Delta,
\end{equation}
for all $k'\ge 2$,
which indeed holds given \eqref{eq:V<Delta} with \begin{equation}
    c_V= \frac{c_a}{c_w \mlr{\max(1,-\ln \delta\kappa_2)}^{d-1}} \min_{k'\ge 2} 2^{k'-2}(k'\ln^2 k')^{\frac{3d+1}{\alpha}} >0. \label{eq:cVbound}
\end{equation}

To summarize, if \eqref{eq:V<Delta} and \eqref{eq:cVbound} hold, \eqref{eq:v<2} holds iteratively up to $k=k_*$, which further leads to \eqref{eq:A<2} by \eqref{eq:a<v}. At the final step, define $D_*=D_{k_*}$ that is $\Delta$-diagonal and $V_*=V_{k_*}$, which then satisfy \eqref{eq:UHU=} and \eqref{eq:D*}.
\end{proof}

\subsection{Proof of Corollary \ref{cor1}}
\begin{proof}[Proof of Corollary \ref{cor1}]
We prove the three statements one by one, where latter proof relies on previous results.

\textbf{1.} Rewrite \begin{equation}
    U = \mathcal{T} \mathrm{exp}\mlr{ \int\limits^1_0 \dd s A(s)},
\end{equation}
where $\mathcal{T}$ is time-ordering, and \begin{equation}
    A(s) = \left\{ \begin{aligned}
        &0, & &s< 2^{1-k_*} \\
        &2^k A_k, & 2^{-k} \le \, & s < 2^{1-k}, \quad (k=1,\cdots, k_*-1)
    \end{aligned} \right.
\end{equation}
Then \eqref{eq:A<2} leads to \begin{equation}
    \norm{A(s)}_{\kappa_*} \le c_A \epsilon, \quad \forall s.
\end{equation}
$U^\dagger \OO U$ is the Heisenberg evolution under the time-dependent Hamiltonian $\ii A(s)$. However, we will use notation $U=\ee^{A}$ for simplicity, with the time-ordering being implicit.
To determine the decomposition \eqref{eq:cU}, we first define \begin{equation}
    \OO_0 := \ee^{-A|_S} \OO \ee^{A|_S} -\OO = -\int^1_0 \dd s \ee^{-sA|_S} [A|_S,\OO] \ee^{sA|_S},
\end{equation}
which is similar to \eqref{eq:ecA-1}. Here $A|_S := \sum_{S'\subseteq S} A_{S'}$, with the optimal decomposition $A_{S'}$ that realizes $\norm{A}_{\kappa_*}$. $\OO_0$ is indeed bounded by \eqref{eq:Or<S}: \begin{align}\label{eq:AOA-O<}
    \norm{\OO_0} &\le \int\limits^1_0 \dd s \norm{ \ee^{-sA|_S} [A|_S,\OO] \ee^{sA|_S}}= \int\limits^1_0 \dd s  \norm{[A|_S,\OO] } \le 2\norm{\OO}\norm{A|_S} \nonumber\\
    &\le 2 \norm{\OO} \sum_{j\in S} \sum_{S'\ni j} \norm{A_{S'}} \le 2\norm{\OO}|S| \norm{A}_{\kappa_*} \le 2 c_A \epsilon |S|\norm{\OO} ,
\end{align}
where we used $A$ is anti-Hermitian in the second line. In the second line of \eqref{eq:AOA-O<}, we have used the fact that each term contained in $A|_S$ must have one site $j\in S$ as its support, so that we bound by first summing over $j\in S$, and then over $S'$ that contains $j$, along with $1 \le \mathrm{e}^{\kappa_* \mathrm{diam}(S)^\alpha}$ to invoke the $\kappa_*$-norm for the optimal decomposition of $A$.
Although many individual factors $S'$ could be badly overestimated in this step for finite $\kappa_*$, the $|S|$ factor in \eqref{eq:AOA-O<} is parametrically optimal, since it is the number of small-region factors $S'$ with $|S'|=\mathrm{O}(1)$ that are contained in $S$.

It remains to bound $\OO_r$ in \eqref{eq:cU} with $r\ge 1$, using that interaction strength decays as $\norm{A}_{S'} \lesssim \epsilon\, \ee^{-\kappa_* (\mathrm{diam}S')^\alpha}$. Although such decay is too slow to have a Lieb-Robinson bound like \eqref{eq:LRB}, in Proposition \ref{prop:reproducing} we prove a bound $\sim \epsilon \ee^{-\kappa_* r^\alpha}$ in \eqref{eq:C1+} for the time-evolved commutator of two operators separated by distance $r$. Choosing $Z=S$ and $Z'=\Lambda \setminus B(S,r)$ as defined in Proposition \ref{prop:reproducing}, we may write (by the triangle inequality) \begin{align}
   \left\lVert \mathcal{O}_r \right\rVert \le  \left\lVert \sum_{k=r}^\infty \mathcal{O}_k\right\rVert + \left\lVert \sum_{k=r+1}^\infty \mathcal{O}_k\right\rVert  
\end{align}
and bound, using Eq. (12) of \cite{hastings2010locality}:
\begin{align}
    \left\lVert \sum_{k=r}^\infty \mathcal{O}_k\right\rVert \le \int\limits_{C_r} \mathrm{d}U \; \lVert [U,\mathcal{O}] \rVert \le \sup_U \lVert [U,\mathcal{O}] \rVert
\end{align}
where the set $C_r$ is over unitaries acting non-trivially on $Z^\prime$ and $\mathrm{d}U$ is the Haar measure.  Proposition \ref{prop:reproducing} bounds the right hand side and leads to (\ref{eq:coper}). 


\textbf{2.} Define superoperators \begin{equation}
    \LL=\ii[H,\cdot], \quad \tilde{\LL}=\ii[U^\dagger HU,\cdot],\quad \tilde{\LL}'=\ii[H_0+D_*,\cdot], \quad \mathcal{U} = U^\dagger \cdot U.
\end{equation}
The left hand side of \eqref{eq:coper} is then \begin{equation}\label{eq:L-LUO}
    \norm{\lr{\ee^{t\LL}-\mathcal{U}^\dagger \ee^{t\tilde{\LL}'}\mathcal{U}}\OO } = \norm{\lr{\ee^{t\tilde{\LL}}-\ee^{t\tilde{\LL}'}}\mathcal{U}\OO },
\end{equation}
where we have used the fact that $\mathcal{U}^\dagger$ does not change the operator norm. Using the Duhamel identity \begin{equation}
    \ee^{t\tilde{\LL}}-\ee^{t\tilde{\LL}'} = \int^t_0 \dd t' \ee^{(t-t') \tilde{\LL}'}(\tilde{\LL}-\tilde{\LL}') \ee^{t' \tilde{\LL}},
\end{equation}
\eqref{eq:L-LUO} is further bounded by \begin{equation}\label{eq:L-LUO<}
    \norm{\lr{\ee^{t\tilde{\LL}}-\ee^{t\tilde{\LL}'}}\mathcal{U}\OO } \le \int^t_0 \dd t' \norm{(\tilde{\LL}-\tilde{\LL}') \ee^{t' \tilde{\LL}} \mathcal{U}\OO} \le |t| \max_{0\le t'\le t} \norm{[U V_* U^\dagger, \OO(t')]},
\end{equation}
where $\OO(t')=\ee^{t'\LL}\OO$, and we have used \eqref{eq:UHU=}. In the commutator, the first operator $U V_* U^\dagger$ is extensive, yet should be close to $V_*$ according to statement \textbf{1} of this Corollary. Thus the local norm of $U V_* U^\dagger$ is still exponentially small: \begin{equation}\label{eq:UVU<}
    \norm{U V_* U^\dagger}_{\kappa'_*} \lesssim \epsilon 2^{-k_*},
\end{equation}
for some $0<\kappa'_*<\kappa_*$. See Proposition \ref{prop:AO-O} for details. The second operator $\OO(t')$ in the commutator in \eqref{eq:L-LUO<} is evolved by the Hamiltonian $H=H_0+V_1$, where interactions $H_{S'}$ decay at least sub-exponentially with $\mathrm{diam}S'$ according to \eqref{eq:LRB} and \eqref{eq:V<Delta}. Although we will prove an algebraic light cone $t\sim r^\alpha$ for such Hamiltonians in Proposition \ref{prop:reproducing}, the light cone is actually linear. After all, sufficiently fast decaying power-law interactions is sufficient to yield linear light cones \cite{Kuwahara:2019rlw}. Thus $\OO(t')$ is mostly supported in a region $\tilde{S}$ of linear size $\mathrm{diam}S + v_{\delta,t'} |t'|$, except for a small part of operator norm $\delta\ll 1$, according to Eq.(5) in \cite{Kuwahara:2019rlw}. Their last equation of section II also ensures that one can safely ignore the $\delta$-tail of $\OO(t')$ acting outside of $\tilde{S}$, because the velocity $v_{\delta,t'}$ grows very mildly when decreasing $\delta\rightarrow0$, if the power-law interaction decays sufficiently fast. When taking commutator with $U V_* U^\dagger$, only terms in $U V_* U^\dagger$ that are within this effective support $\tilde{S}$ will contribute. This effect is bounded by a volume factor $|\tilde{S}| \lesssim |S| (|t'|+1)^d$. Combining the factor $|t|$ in \eqref{eq:L-LUO<}, the local norm \eqref{eq:UVU<}, and $\norm{\OO(t')}=\norm{\OO}$, there must exist some constant $c_{\mathrm{oper}}$ such that \eqref{eq:coper} holds. 

\textbf{3.} The trace norm is related to the operator norm by \begin{align}
    \norm{\rho_S(t)-\rho'_S(t)}_1 &= \max_{\OO=\OO_S:\norm{\OO}\le 1} \mathrm{Tr}\mlr{\OO\lr{\rho_S(t)-\rho'_S(t)}} \nonumber\\ 
    &= \max_{\OO=\OO_S:\norm{\OO}\le 1} \mathrm{Tr}\mlr{\rho_0 \lr{\ee^{\ii t H}\mathcal{O} \ee^{-\ii t H} - U\ee^{\ii t (H_0+D_*)}U^\dagger \mathcal{O} U\ee^{-\ii t (H_0+D_*)}U^\dagger }} \nonumber\\
    &\le \max_{\OO=\OO_S:\norm{\OO}\le 1} \norm{\ee^{\ii t H}\mathcal{O} \ee^{-\ii t H} - U\ee^{\ii t (H_0+D_*)}U^\dagger \mathcal{O} U\ee^{-\ii t (H_0+D_*)}U^\dagger } \le c_{\mathrm{oper}} \epsilon |S|(|t|+1)^{d+1}2^{-k_*},
\end{align}
where $\OO$ is an arbitrary operator supported in $S$. The second line comes from the definitions \eqref{eq:rhoS},\eqref{eq:rhoS'} and rearranging orders in the trace. The last step follows from the second result of this corollary, \eqref{eq:coper}.
\end{proof}

\section{Filter function and its locality when acting on operators}\label{sec:filter}
This section contains the proof of Proposition \ref{prop:PO<}, which requires the existence of a function $w(t)$ (as sketched in the main text) that can be used to build a projector $\mathbb{P}$ onto $\Delta$-diagonal operators.

\subsection{Defining the filter function}
In the proof above we frequently want to project out, for some operator, ``off-resonant'' matrix elements that connects pairs of eigenstates of $H_0$ that have energy difference $|E-E'|>\Delta$. This can be achieved as follows. Define the $\Delta$-dependent filter function $w\in L^1(\mathbb{R})$ \cite{hastings2010quasi,bachmann2012automorphic}
\begin{align}
    w(t) := c_\Delta \Delta \prod_{n=1}^\infty \lr{\frac{\sin a_n t}{a_n t}}^2,\quad \mathrm{where}\quad a_1=c_1\Delta, \quad a_n=\frac{a_1}{n\ln^2n}, \quad\forall n\ge 2.
\end{align}
Here $c_1\approx 0.161$ is chosen such that \begin{equation}
    \sum_{n=1}^\infty a_n=\frac{\Delta}{2},
\end{equation} and $c_\Delta\in (1/(2\pi), 1/\pi)$ is a pure number chosen so that the function is normalized:\begin{equation}
    \int_{-\infty}^\infty \mathrm{d}t \; w(t) = 1. \label{eq:intwt1}
\end{equation}
We define a similarly related odd function $W(t)$ by \begin{equation}
    W(t)=-W(-t)=\int^\infty_t \dd s \; w(s), \quad (t>0).
\end{equation} These two functions have useful properties summarized in the following Proposition, which is proved in \cite{hastings2010quasi,bachmann2012automorphic}. 

\begin{prop}
$w(t)$ satisfies the following:
\begin{enumerate}
    \item It is even in $t$, with bound $0\le w(t) \le c_\Delta \Delta$.
    \item The Fourier transform is compact \begin{equation}
        \hat{w}(\omega)= 0, \quad\forall |\omega|\ge \Delta,
    \end{equation}
    and bounded $|\hat{w}(\omega)|\le \int w(t) \dd t=1=\hat{w}(0)$.
    \item subexponential decay: \begin{align}\label{eq:w<}
        w(t) \le 2(\ee \Delta)^2t\ee^{-\frac{2}{7}\frac{\Delta t}{\ln^2\Delta t}}, \quad \text{ if } t\ge \ee^{1/\sqrt{2}}\Delta^{-1}.
    \end{align}
\end{enumerate}
$W(t)$ satisfies 
\begin{enumerate}
    \item $W(t)$ is bounded as \begin{equation}
    |W(t)| \le W(0+) = -W(0-) = 1/2.
\end{equation}
\item For any bounded function $f(t)$, \begin{equation}\label{eq:Wf}
    \int^\infty_{-\infty} W(t) \ddt{f(t)} \dd t = -f(0)+ \int^\infty_{-\infty} w(t) f(t)\dd t.
\end{equation}
\item $W(t)$ has weakly subexponential decay: \begin{align}\label{eq:W<}
        |W(t)| \le 2\ee^2 \int^\infty_{\Delta|t|} s\ee^{-\frac{2}{7}\frac{s}{\ln^2 s}}\dd s, \quad \text{ if } |t|\ge \ee^{1/\sqrt{2}}\Delta^{-1}.
    \end{align}
\end{enumerate} 
\end{prop}

As a remark, one may wonder if the nearly exponential decay of $w(t)\sim \exp(-t/\ln^2 t)$ can be improved, for example to true exponential decay, while preserving the compact Fourier transform property. This is forbidden by a well-known math result:  
\begin{prop}\label{prop:wt}
   If $w(t)$ satisfies $|w(t)|\le C \mathrm{e}^{-c |t|}$ for all $t\in\mathbb{R}$, then its Fourier transform $\hat{w}(\omega)$ cannot have compact support unless $w(t)=0$.
\end{prop}
\begin{proof}
    The $n$-th derivative of its Fourier transform is bounded by \begin{align}
    |\hat{w}^{(n)}(\omega)| = \left| \int^{\infty}_{-\infty} \mathrm{d} t \, w(t) t^n \mathrm{e}^{\mathrm{i} \omega t} \right| \le C \int^{\infty}_{-\infty} \mathrm{d} t \, t^n \mathrm{e}^{-c|t|} = 2C c^{-n} n!.
\end{align}
This implies the Taylor series of $\hat{w}(\omega)$ at any $\omega\in \mathbb{R}$ has radius of convergence at least $c$, so it is real analytic over all $\mathbb{R}$ and cannot have compact support unless it is identically $0$.
\end{proof}

Since $w(t)$ is in the integral \eqref{eq:PO}, the operator $\mathbb{P}\mathcal{O}$ decays slower than exponential in its support diameter. This makes our choice of $\kappa$ norm \eqref{eq:norm} with $\alpha<1$ almost optimal. 

\subsection{Lieb-Robinson bound for the $\kappa$-norm}\label{sec:L0<}
In this section, we establish a technical lemma for proving Proposition \ref{prop:PO<}, which can be skipped in a first reading.  In a nutshell, we wish to bound the growth in the $\kappa$-norm with the Lieb-Robinson bounds.  This approach uses established, albeit tedious, methods.

Since the $\kappa$-norm takes infimum over all possible local decompositions of the form \eqref{eq:O=OS}, it suffices to prove bound on a particular local decomposition. Moreover, we will frequently use the $\kappa$-norm probed at vertex $i$: \begin{align}\label{eq:kappai}
    \|\OO\|_{\kappa, i}:= \sum_{S \ni i} \mathrm{e}^{\kappa (\mathrm{diam}S)^\alpha} \left\|\OO_{S}\right\|, 
\end{align}
where a local decomposition is implicitly chosen. Then at the final step of the proof, we will use $\|\OO\|_{\kappa} = \inf_{\{\OO_S\}} \max_{i \in \Lambda} \|\OO\|_{\kappa, i}$.

Any evolved local operator $\ee^{t\LL_0}\OO_{S_0}$, can be decomposed by \begin{equation}\label{eq:Qr}
    \ee^{t\LL_0}\OO_{S_0} = \sum_{r=0}^\infty \QQ_r \lr{\ee^{t\LL_0}\OO_{S_0}},
\end{equation}
where the projector $\mathbb{Q}_r := \QQ_{r}^\prime - \QQ_{r-1}^\prime$ and \begin{equation} \label{eq:QQprimedef}
\QQ_r^\prime := \int\limits_{\text{Haar outside }S({S_0}, r)} \mathrm{d}U \; U^\dagger \mathcal{O} U, 
\end{equation}
where $S(S_0,r)$ is the simplex $S\supset S_0$ whose faces are all of distance $r$ to the parallel faces of $S_0$. The measure ``Haar outside $S({S_0}, r)$" denotes the Haar measure on all unitary operators supported outside the set $S(S_0,r)$.  Put simply, $\QQ_r$ is a projection onto all operators whose farthest support from ${S_0}$ is a ``distance'' $r$ away from ${S_0}$.  Note that \begin{equation}
 \lVert \QQ^\prime_r \OO\rVert \le \lVert \OO\rVert, \;\;\;\;\;\;\;  \lVert \QQ_r \OO\rVert \le  2\lVert \OO\rVert,  \label{eq:haaravg}
\end{equation} where the second inequality comes from (\ref{eq:QQprimedef}) and the triangle inequality applied to the definition of $\QQ_r$.

Using \eqref{eq:Qr}, we define a local decomposition for $\ee^{t\LL_0}\OO$ by simple extension: \begin{equation}\label{eq:L0O}
    \ee^{t\LL_0}\OO = \sum_{S_0} \sum_{r=0}^\infty \QQ_{r,{S_0}} \lr{\ee^{t\LL_0}\OO_{S_0}},
\end{equation}
where we have shown the explicit dependence of $\QQ_r$ on ${S_0}$, and each $\QQ_{r,{S_0}}$ term is viewed as supported in simplex $S(S_0,r)$. The decomposition $\OO=\sum_{S_0}\OO_{S_0}$ is optimal and is chosen to be (arbitrarily close to) minimizing the $\kappa$-norm $\norm{\OO}_\kappa$. 
The advantage of this decomposition is to invoke the Lieb-Robinson bound \eqref{eq:LRB}, which implies for any $r$, \begin{equation}\label{eq:Qr<}
    \norm{\QQ_r \lr{\ee^{t\LL_0}\OO_{S_0}}} \le 2\norm{\OO_{S_0}} \min\lr{1, c_d \mu'(1+\mathrm{diam} {S_0})^{d-1} \ee^{\mu(u|t|-r)} },
\end{equation}
where we have used \eqref{eq:cd}.
Here the first argument $1$ in $\min$ is from (\ref{eq:haaravg}).

\begin{lem}\label{lem:L0<}
If $\kappa'<\kappa\le \kappa_1$ and $H_0$ satisfies the Lieb-Robinson bound \eqref{eq:LRB} with \begin{equation}\label{eq:mu>kap}
    \mu\ge 5\kappa_1,
\end{equation}
then
\begin{align}\label{eq:L0O<}
\norm{\ee^{t\LL_0} \OO}_{\kappa'} \le c_0 \mlr{\max(1,-\ln \delta\kappa)}^{d-1}\norm{\OO}_\kappa (4u |t|+c_\mu)^d \ee^{\kappa' (4u |t|)^\alpha},
\end{align}
where $\delta\kappa = \kappa-\kappa'$, $c_0$ and $c_\mu$ only depends on $d,\kappa_1,\mu,\mu'$, and $\alpha$.
\end{lem}

\begin{proof}
Assume $t\ge 0$ without loss of generality (alternatively set $|t|\rightarrow t$). For the decomposition \eqref{eq:L0O}, suppose that a given vertex $i$ in \eqref{eq:kappai} is the one responsible for the $\kappa'$-norm of $\ee^{t\LL_0}\OO$: in what follows, we will assume this $i$ is fixed, as in \eqref{eq:kappai}, but the result holds for any $i$ and thus eventually for \eqref{eq:L0O<} as well.  Further fix a set ${S_0}$.  Then the initial operator $\OO_{S_0}$ can contribute an amount
\begin{equation} \label{eq:KOO{S_0}}
    K_{\kappa',i}(\OO_{S_0}) \le \norm{\QQ_* \lr{\ee^{t\LL_0}\OO_{S_0}}} \ee^{\kappa'(s-1 + 2r_0)^\alpha}\mathbb{I}(r_0 \ge \mathsf{d}(i,{S_0})) + \sum_{r=r_0+1, r\ge \mathsf{d}(i,{S_0})}^\infty \norm{\QQ_r \lr{\ee^{t\LL_0}\OO_{S_0}}} \ee^{\kappa'(s-1 + 2r)^\alpha},
\end{equation}
according to decomposition \eqref{eq:Qr}, where we write $s=1+\mathrm{diam}{S_0}$ for convenience, and $\mathbb{I}(\cdot)$ is the indicator function that returns $1$ for input True and $0$ for False. We have combined the leading terms \begin{equation}
    \QQ_0+\cdots+\QQ_{r_0}=\QQ_*
\end{equation} into a single operator, which has support in $S({S_0},r_0)$, with $r_0$ a constant chosen shortly to distinguish between pieces of the operator with ``large" and ``small" support. 

Intuitively, the contribution $K_{\kappa',i}(\OO_{S_0})$ is small if $i$ is far from the initial support ${S_0}$, compared to the distance $ut$ that an operator can expand during time $t$. To formalize, observe that \eqref{eq:Qr<} decays exponentially with $r$ at sufficiently large $r$.  This exponential decay is assured to kick in once  $r>r_0^\prime$, where we define $r_0^\prime$ as the solution to \begin{equation}
    1=c_d\mu's^{d-1} \ee^{\mu(ut-r_0^\prime)}, \rarrow r_0^\prime= ut + \frac{1}{\mu} \ln\lr{c_d\mu' s^{d-1}}. 
\end{equation}
We now choose $r_0\approx 2r_0^\prime$, noting that $r_0$ depends on $t$ and $s$:   \begin{equation}\label{eq:r0}
    r_0 = 2\left\lfloor ut + \frac{1}{\mu} \ln\lr{c_d\mu' s^{d-1}} \right\rfloor +1.
\end{equation} 
Here $\lfloor a\rfloor$ is the largest integer below $a$. Note that we are free to choose a large enough $\mu'$ so that $r_0$ is always positive. As a result, \eqref{eq:Qr<} transforms to \begin{subequations}\begin{align}
    \norm{\QQ_{*} \lr{\ee^{t\LL_0}\OO_{S_0}}} &\le 2 \norm{\OO_{S_0}}, \label{eq:Qr0<} \\
    \norm{\QQ_r \lr{\ee^{t\LL_0}\OO_{S_0}}} &\le 2c_d\mu'\norm{\OO_{S_0}} s^{d-1} \ee^{\mu(ut-r)} \le 2c_d\mu'\norm{\OO_{S_0}} s^{d-1} \ee^{\mu\lr{ut-\frac{r_0+1}{2}-\frac{r}{2}}} \le 2 \norm{\OO_{S_0}}\ee^{-\mu r/2},\quad  r>r_0. \label{eq:Qr<e-r}
\end{align}\end{subequations}

We can now bound $K_{\kappa',i}(\OO_{S_0})$.   We need to consider whether $r_0$ is larger or smaller than $\mathsf{d}(i,{S_0})$.  Let us start with the possibility that $r_0<\mathsf{d}(i,{S_0})$.   Then \eqref{eq:Qr<e-r} yields \begin{align}\label{eq:KO<}
    K_{\kappa',i}(\OO_{S_0}) &= 2\norm{\OO_{S_0}} \sum_{r=\mathsf{d}(i,{S_0})}^\infty  \ee^{-\mu r/2}\ee^{\kappa'(s-1 + 2r)^\alpha} \le 2\norm{\OO_{S_0}} \sum_{r=\mathsf{d}(i,{S_0})}^\infty  \ee^{-\mu r/2}\ee^{\kappa' [(s-1)^\alpha+ 2r]} \nonumber\\ &\le  2\norm{\OO_{S_0}} \ee^{-(\mu/2 - 2\kappa') \mathsf{d}(i,{S_0})}\ee^{\kappa' (s-1)^\alpha} \lr{1-\ee^{-(\mu/2 - 2\kappa')}}^{-1} \nonumber \\
    &\le \frac{2}{1-\ee^{-\mu/10}}\norm{\OO_{S_0}} \ee^{-\mu \mathsf{d}(i,{S_0})/10}\ee^{\kappa' (s-1)^\alpha}, \quad r_0<\mathsf{d}(i,{S_0}).
\end{align}
In the first line, we have used $\alpha<1$ and \begin{equation}\label{eq:a+b<}
    (a+b)^\alpha \le a^\alpha + b^\alpha, \quad \forall a,b\ge 0.
\end{equation}
In the second line we have summed the geometric series and used \eqref{eq:mu>kap}. For the second case, $r_0\ge \mathsf{d}(i,{S_0})$, the $r>r_0$ part of summation is done exactly as \eqref{eq:KO<}, with $\mathsf{d}(i,{S_0})$ replaced by $r_0+1$.  Thus combining with \eqref{eq:Qr0<} yields \begin{equation}\label{eq:KO<1}
    K_{\kappa',i}(\OO_{S_0}) \le 2\norm{\OO_{S_0}} \mlr{ \ee^{\kappa'(s-1+2r_0)^\alpha} +  \frac{1}{\ee^{\mu/10}-1} \ee^{-\mu r_0/10}\ee^{\kappa' (s-1)^\alpha}} \le \lr{2+\frac{20}{\mu}} \norm{\OO_{S_0}}\ee^{\kappa'(s-1+2r_0)^\alpha}, \quad r_0\ge \mathsf{d}(i,{S_0}),
\end{equation}
 where we have simplified the prefactor using \begin{equation}
    \ee^{a}-1 \ge a.
\end{equation}

The $\kappa'$-norm at $i$, denoted by $\norm{\ee^{t\LL_0} \OO}_{\kappa',i}$, is the sum $K_{\kappa',i}(\OO_{S_0})$ over all $\OO_{S_0}$. For each $\OO_{S_0}$, let $x=\mathsf{d}(i,{S_0})$: by definition, there is at least one site $j\in {S_0}$ such that $\mathsf{d}(i,j)=\mathsf{d}(i,{S_0})\equiv x$. We can sum over ${S_0}$ by grouping the sums according to the $x$ and $j$:  the outermost sum will be over $x$, then we will sum over $j$ at a fixed distance $x$, and then sum over sets ${S_0}$ with $j\in {S_0}$.  Note that there can be multiple valid $j$ for each ${S_0}$, so this sum will overestimate the bound: 
\begin{align}\label{eq:<K+K}
    \norm{\ee^{t\LL_0} \OO}_{\kappa',i} &\le  \sum_{x=0}^\infty \sum_{j: \mathsf{d}(i,j)=x} \sum_{\OO_{S_0}}^{j,x} K_{\kappa',i}(\OO_{S_0}) = \sum_{x=0}^\infty \sum_{j: \mathsf{d}(i,j)=x} \mlr{\sum_{\OO_{S_0}:s<f_0(x)}^{j,x} K_{\kappa',i}(\OO_{S_0}) +\sum_{\OO_{S_0}:s\ge f_0(x)}^{j,x} K_{\kappa',i}(\OO_{S_0})},
\end{align}
where $\sum^{j,x}_{\OO_{S_0}}$ means the restriction that ${S_0}\ni j$ and $\mathsf{d}(i,{S_0})=x$ (as described above in words). In the latter equality, we separated the sum according to whether $r_0(s)< \mathsf{d}(i,{S_0})=x$, which is equivalent to whether $1+\mathrm{diam}{S_0}=s<f_0(x)$, where \begin{equation}
    x = 2\left\lfloor ut + \frac{1}{\mu}\ln \left(c_d\mu' f_0(x)^{d-1}\right) \right\rfloor + 1. \label{eq:f0x}
\end{equation} 

Note that for $d=1$ there is no solution for $f_0(x)$ since $r_0$ is a fixed number independent of $s$, and hence we do not need to sum over $s\ge f_0(x)$ in $d=1$: \eqref{eq:<K+K} simply vanishes for all $x>r_0$. Thus for $d=1$ we can simply set $f_0(x)=1$ for $x\le r_0$ and $f_0(x)=+\infty$ otherwise, so that \eqref{eq:<K+K} and the following equations still make sense.

 We first bound the first term in \eqref{eq:<K+K} using \eqref{eq:KO<}: \begin{align}\label{eq:c_<}
    \sum_{x=0}^\infty \sum_{j: \mathsf{d}(i,j)=x} \sum_{\OO_{S_0}:s<f_0(x)}^{j,x} K_{\kappa',i}(\OO_{S_0}) &\le \frac{2}{1-\ee^{-\mu/10}} \sum_{x=0}^\infty \sum_{j: \mathsf{d}(i,j)=x} \sum_{\OO_{S_0}:{S_0}\ni j} \norm{\OO_{S_0}} \ee^{-\mu x/10}\ee^{\kappa' (s-1)^\alpha} \nonumber\\
    &\le \frac{2}{1-\ee^{-\mu/10}} \sum_{x=0}^\infty \sum_{j: \mathsf{d}(i,j)=x} \ee^{-\mu x/10} \norm{\OO}_{\kappa'} \nonumber\\
    &\le \frac{2 c_d}{1-\ee^{-\mu/10}} \norm{\OO}_{\kappa'}  \sum_{x=0}^\infty (1+2x)^{d-1} \ee^{-\mu x/10} = c_< \norm{\OO}_\kappa,
\end{align}
which is well upper bounded by the right hand side of \eqref{eq:L0O<}. Here we have used \eqref{eq:cd} with $S=\{j:\mathsf{d}(i,j)\le x\}$ to get \begin{equation}\label{eq:cLam}
    \sum_{j:\mathsf{d}(i,j)= x} \le c_d (1+2x)^{d-1},
\end{equation}
because $\partial S=\{j:\mathsf{d}(i,j)= x\}$ and $\mathrm{diam}S\le 2x$ from triangle inequality.
The constant $c_<$ only depends on $d$ and $\mu$. We also used $\{\OO_{S_0}\}$ is the optimal decomposition of $\OO$ that realizes its $\kappa'$-norm.

We now evaluate the second term in \eqref{eq:<K+K}.  First, we use that if $s\ge f_0(x)$, we may as well use \eqref{eq:KO<1} to bound \begin{align}
    K_{\kappa',i}(\mathcal{O}_{S_0}) &\le \left(2+\frac{20}{\mu}\right)\norm{\mathcal{O}_{S_0}} \mathrm{e}^{\kappa^\prime (s-1+2r_0)^\alpha} \notag \\
    &\le \left(2+\frac{20}{\mu}\right)\norm{\mathcal{O}_{S_0}}\exp\glr{\kappa'\mlr{(4ut)^\alpha+ \lr{s+1+\frac{4}{\mu} \ln\lr{c_d\mu' s^{d-1}}}^\alpha } } \notag \\
    &\le \left(2+\frac{20}{\mu}\right)\mathrm{e}^{\kappa^\prime (4ut)^\alpha}\norm{\mathcal{O}_{S_0}}\exp\glr{\kappa' \mlr{(s-1)^\alpha+c_{\ln}} }. \label{eq:KOSfirst}
\end{align}
In the second line we have used \eqref{eq:a+b<}. In the third line we have used the fact that the function $(s+a\ln s+b)^\alpha - s^\alpha$
is upper bounded for any $a,b$ (given $\alpha<1$), and the resulting constant $c_{\ln}$ in \eqref{eq:KOSfirst} is determined by $\alpha,d,\mu,\mu'$.  Now, combining (\ref{eq:<K+K}), (\ref{eq:c_<}), and (\ref{eq:KOSfirst}), we find
\begin{align}\label{eq:L0O<x}
    \norm{\ee^{t\LL_0} \OO}_{\kappa',i}  &\le c_< \norm{\OO}_\kappa +\left(2+\frac{20}{\mu}\right)\ee^{\kappa'(4ut)^\alpha} \sum_{x=0}^\infty \sum_{j: \mathsf{d}(i,j)=x} \sum_{\OO_{S_0}:s\ge f_0(x)}^{j,x}\norm{\OO_{S_0}}\exp\glr{\kappa' \mlr{(s-1)^\alpha+c_{\ln}} } \nonumber\\
    &\le c_< \norm{\OO}_\kappa + \left(2+\frac{20}{\mu}\right)\ee^{\kappa'\mlr{c_{\ln}+(4ut)^\alpha}} \sum_{x=0}^\infty \sum_{j: \mathsf{d}(i,j)=x} \ee^{-\delta\kappa (f_0(x)-1)^\alpha } \norm{\OO}_\kappa \nonumber\\
    &\le c_< \norm{\OO}_\kappa + \left(2+\frac{20}{\mu}\right)\ee^{\kappa'\mlr{c_{\ln}+(4ut)^\alpha}}\norm{\OO}_\kappa \sum_{x=0}^\infty c_d (2x+1)^{d-1}\ee^{-\delta\kappa (f_0(x)-1)^\alpha }.
\end{align}
 In the second line of (\ref{eq:L0O<x}) we have used the Markov inequality \begin{align}\label{eq:markov}
    \sum_{\OO_{S_0}:s\ge f_0(x)}^{j,x}\norm{\OO_{S_0}}\ee^{\kappa' (s-1)^\alpha } &= \sum_{\OO_{S_0}:s\ge f_0(x)}^{j,x}\norm{\OO_{S_0}} \ee^{(\kappa-\delta\kappa) (s-1)^\alpha } \nonumber\\ 
    &\le \sum_{\OO_{S_0}:s\ge f_0(x)}^{j,x}\norm{\OO_{S_0}} \ee^{\kappa (s-1)^\alpha } \ee^{-\delta\kappa (f_0(x)-1)^\alpha } \le \norm{\OO}_\kappa \ee^{-\delta\kappa (f_0(x)-1)^\alpha }.
\end{align}
In the last line of (\ref{eq:markov}) we have overestimated the final sum over ${S_0}$. 
Returning to the last line of (\ref{eq:L0O<x}), we have used \eqref{eq:cLam}. 
 
 If $d=1$, the sum over $x$ in \eqref{eq:L0O<x} is finite: $x\le r_0 \approx 2ut$, and \eqref{eq:L0O<} follows easily. 
 
For $d>1$, note that there exists a constant $c'_\mu$ that depends on $d,\mu,\mu'$, such that when $x\ge 4ut+c'_\mu$, \eqref{eq:f0x} yields \begin{equation}
    f_0(x) \ge \lr{c_d\mu'}^{-\frac{1}{d-1}}\ee^{\frac{\mu}{d-1}\lr{\frac{x-1}{2} -ut-1}} \ge \ee^{\frac{\mu x}{4(d-1)}}\equiv y, \;\;\; (x\ge 4ut+c'_\mu).
\end{equation}
Changing variable from $x$ to $y$ defined above, \eqref{eq:L0O<x} becomes \begin{align}\label{eq:}
    \norm{\ee^{t\LL_0} \OO}_{\kappa',i}  &\le c_< \norm{\OO}_\kappa+ \frac{(2\mu+20)c_d}{\mu}\ee^{\kappa'\mlr{c_{\ln}+(4ut)^\alpha}}\norm{\OO}_\kappa \mlr{(4ut+c'_\mu+1)^d + \lr{\frac{4(d-1)}{\mu}}^d c_{\mathrm{int}}\int^\infty_1 \frac{\dd y}{y} (\ln y)^{d-1}\ee^{-\delta\kappa (y-1)^\alpha }  } \nonumber\\
    &\le c_< \norm{\OO}_\kappa+\frac{(2\mu+20)c_d}{\mu}\ee^{\kappa'\glr{c_{\ln}+(4ut)^\alpha}}\norm{\OO}_\kappa \mlr{(4ut+c'_\mu+1)^d + \lr{\frac{4(d-1)}{\mu}}^d c'_{\ln} \mlr{\max(1,-\ln \delta\kappa)}^{d-1} } \nonumber\\
    &\le c_< \norm{\OO}_\kappa+ \frac{(2\mu+20)c_d c'_{\ln}}{\mu}\ee^{\kappa'\mlr{c_{\ln}+(4ut)^\alpha}}\norm{\OO}_\kappa \mlr{\max(1,-\ln \delta\kappa)}^{d-1} \lr{4ut+c'_\mu + \frac{4(d-1)}{\mu}}^d.
\end{align}
Here in the first line we have summed over $x<4ut+c'_\mu$, and $c_{\mathrm{int}}$ is a constant due to replacing the sum in \eqref{eq:L0O<x} by an integral.  In the second line we have used the fast decay of the exponential function: if $\delta\kappa = \Omega(1)$, then the integral is $\mathrm{O}(1)$; otherwise if $\delta\kappa\ll 1$, then rescaling $\delta\kappa y^\alpha = \tilde{y}^\alpha$ will pull out an overall factor $\lr{\ln \frac{1}{\delta\kappa}}^{d-1}$. In the last line, we combined the two terms using $a^d+b^d \le (a+b)^d$. Finally, as our result does not depend at all on $i$, we are free to re-label our final O(1) constants sitting out in front; in doing so, we arrive at \eqref{eq:L0O<}.
\end{proof}

\subsection{Proof of Proposition \ref{prop:PO<}}

\begin{proof}[Proof of Proposition \ref{prop:PO<}]
\eqref{eq:HA=O} comes from \begin{align}\label{eq:HAO=P-1}
    [H_0, \bA\OO] = \int^\infty_{-\infty} W(t)\ddt{} \lr{\ee^{t\LL_0}\OO} \dd t = (-1+\PP) \OO,
\end{align} 
following \eqref{eq:Wf}.
To prove \eqref{eq:PO<}, note that the condition \eqref{eq:mu>kap} of Lemma \ref{lem:L0<} is just \eqref{eq:kap<mu}. Thus we use \eqref{eq:L0O<} to get
\begin{align}\label{eq:PO<<}
    &\norm{\PP \OO}_{\kappa'} \le \int^\infty_{-\infty} \dd t\; w(t) \norm{\ee^{t\LL_0} \OO}_{\kappa'}  \le 2c_0 \mlr{\max(1,-\ln \delta\kappa)}^{d-1} \norm{\OO}_\kappa \int^\infty_0 \dd t\; w(t)  (4u t+c_\mu)^d \ee^{\kappa' (4u t)^\alpha} \nonumber\\ 
    &\le 2c_0\mlr{\max(1,-\ln \delta\kappa)}^{d-1} \norm{\OO}_\kappa \mlr{ \lr{c_\mu +\ee^{\frac{1}{\sqrt{2}}}\frac{4u}{\Delta} }^d \ee^{\kappa\lr{\ee^{\frac{1}{\sqrt{2}}}\frac{4u}{\Delta} }^\alpha}+ 2\ee^2 \int^\infty_{\ee^{\frac{1}{\sqrt{2}}}} \dd \tilde{t}\ee^{\kappa\lr{\ee^{\frac{1}{\sqrt{2}}}\tilde{t}\frac{4u}{\Delta} }^\alpha-\frac{2\tilde{t}}{7\ln^2 \tilde{t}}}\lr{c_\mu+\ee^{\frac{1}{\sqrt{2}}}\tilde{t}\frac{4u}{\Delta} }^d  }.
\end{align}
In the second line, we have first used \eqref{eq:w<} to bound the large $t$ tails of the integral; for the small $t$ limit we have simply bounded the integral by using the maximum of each term in the integrand separately in the domain $t<\mathrm{e}^{1/\sqrt{2}}\Delta^{-1}$, together with (\ref{eq:intwt1}).
which reduces to the form of \eqref{eq:PO<} since the integral converges for any $\alpha<1$. \eqref{eq:AO<} comes from \eqref{eq:W<} using almost identical manipulations.
\end{proof}

\section{Proof of Proposition \ref{prop:AO-O}}\label{sec:LO-O<}

\subsection{Motivation}

It remains to prove Proposition \ref{prop:AO-O}, which is the main difficulty of generalizing \cite{abanin2017rigorous}. Since this Proposition is a self-contained bound on local dynamics beyond conventional Lieb-Robinson bound, we restate it below replacing $A$ with $\ii tH$. This notation will be used for this whole section.  Therefore, in this section, the Hamiltonian $H$ is \emph{not} the same one as \eqref{eq:H0+V}: we only require it to be local in the sense of \eqref{eq:tH<C}.

\setcounter{prop1}{5}
\begin{prop1}[restatement]\label{prop:LO-O}
Suppose $\kappa'<\kappa\le \kappa_1$ with \begin{equation}
    \delta\kappa=\kappa-\kappa'.
\end{equation}
If \begin{equation}\label{eq:L}
    \LL=\ii [H,\cdot],
\end{equation}
with Hamiltonian $H$ on a $d$-dimensional simplicial lattice, satisfying \begin{equation}\label{eq:tH<C}
    |t|\norm{H}_\kappa \le (\delta\kappa)^{\frac{3d+1}{\alpha}},
\end{equation}
then
\begin{equation}\label{eq:LO-O}
    \norm{\ee^{t\LL} \OO -\OO}_{\kappa'} \le c_- (\delta \kappa)^{-\frac{2d-1}{\alpha} } |t|\norm{H}_{\kappa} \norm{\OO}_\kappa,
\end{equation}
where $0<c_-<\infty$ depends on $d,\alpha$ and $\kappa_1$.\footnote{We believe the exponent of $\delta\kappa$ can be improved, for example, to $(\delta \kappa)^{-\frac{1}{\alpha}\max(d,2d-2) }$, using complicated geometrical facts about simplices. Such improvement will lead to a larger $a>\frac{1}{2d-1}$ in the prethermal time scale $t_*\sim \exp \lr{\frac{\Delta}{\epsilon}}^a$. }
\end{prop1}

The difference between Proposition \ref{prop:AO-O} and its counterpart, Lemma 4.1 in \cite{abanin2017rigorous}, is that here $H$ is less local: \begin{equation}\label{eq:AS<S}
    \norm{H_S}\lesssim \ee^{-\kappa (\mathrm{diam}S)^\alpha},\quad \alpha<1.
\end{equation} 
In contrast,  $\norm{H_S}\lesssim\ee^{-\kappa |S|}$ in \cite{abanin2017rigorous}. As a result, a simple combinatorial expansion of $\ee^{t\LL}\OO$ over $t$, which was done in \cite{abanin2017rigorous}, no longer converges for $d>1$.  Essentially, the issue is that while $H$ is sub-exponentially localized\footnote{As Hamiltonians with algebraic tails have rigorous Lieb-Robinson bounds \cite{fossfeig,Chen:2019hou,Kuwahara:2019rlw,Tran:2020xpc,Tran:2021ogo}, this remains an extremely strong condition.  Still, it requires some care to re-sum and find a bound on $\lVert \mathrm{e}^{\mathcal{L}t} \mathcal{O}\rVert_\kappa$, which is not the object usually bounded by a Lieb-Robinson Theorem!}, the growth in the $\kappa$-norm can be dominated by tiny terms in $H$ with unusually large polygons $S$ in which they are supported.  In $d>1$, there are an increasingly large number of ways for such large polygons to intersect with $\mathcal{O}$, so they must be summed up with some care to not overcount.

To sketch the proof that follows, we first observe that evolution generated by a Hamiltonian $H$ of form \eqref{eq:AS<S} still has a Lieb-Robinson bound that we will prove in Proposition \ref{prop:reproducing} using established methods (see, e.g., \cite{hastings2010locality}), which resums the divergence of the simple Taylor expansion mentioned above. The idea is illustrated in Fig.~\ref{fig:d1} for $1$d, where we focus on a single local operator $\OO_{S_0}$, since the conclusion for the extensive $\OO$ follows simply by superposition. The rectangle at each layer represents commuting the operator with a local term $H_X$ in the Hamiltonian $H=\sum_X H_X$. Suppose at some step, the operator has support on $S$. Then there are $\sim |S|$ terms of $H_X$ that act nontrivially on the operator. However, only $\sim |\partial S|$ of them (red rectangles in Fig.~\ref{fig:d1}) grow the operator to a strictly larger support, while the ``bulk'' ones (yellow rectangles) yield unitary rotations inside the support. Using Lieb-Robinson techniques, one can essentially ignore these internal rotations, and bound the operator growth only by the ``boundary'' ones, leading to a convergent series.  

Unfortunately there is an important technical difference between a standard Lieb-Robinson bound, which bounds $\lVert[ A_S(t),B_R]\rVert$ for fixed sets $S$ and $R$, and a bound on $\lVert \mathcal{O}(t)\rVert_\kappa$.  In the former, we need only keep track of (loosely speaking) terms that grow set $S$ towards set $R$.  In the latter, we need to keep track of all terms which grow the operator in any direction on the lattice -- for $d>1$, increasingly large operators have a large perimeter with many possible ways to grow. In particular, suppose the Hamiltonian terms are all of size $\mathrm{diam}S\sim r$. Then the typical size of a grown operator is $\ge 2r$, because each end of the operator can be attached by a $H_S$ of size $r$ that just touches the end.  Even in $d=1$, a Lieb-Robinson bound can affirm that it took time $ut$ for (terms of high weight) in the operator to expand a distance $r$ to the right; during the same time it also will likely expand a distance $r$ to the left.  When $r\gg |S|$, this implies that $|S|$ grows twice as fast as the Lieb-Robinson velocity.  So in every iteration, the typical size of an operator in $H_k$ will at least double. Returning to the overarching sketch of our proof, this would make $k_*\sim \log \frac{\Delta}{v_1}$. So we need to use the extra fact that, by attaching more $H_S$ to the initial operator, the amplitude is suppressed by more powers of $\norm{H}_\kappa$. In other words, we need to differentiate cases where $1$ or $2$ ends are attached by $H_S$, which is not considered in conventional Lieb-Robinson bounds. 



Our strategy can be intuitively described first in $d=1$, and we will do so now. 
Let operator $\mathcal{O}_S$ have support on connected subset $S\subseteq \Lambda$: namely $S$ is an interval. $\mathcal{O}_S$ can only grow at the two ends (left and right) of the support interval. For example in Fig.~\ref{fig:d1}, we are studying a term in the expansion of the time-evolved operator of the form \begin{equation*}
    \mathrm{e}^{\mathcal{L}t}\mathcal{O}_{S_0} \subseteq \frac{t^5}{5!} \mathcal{L}_{X_5}\mathcal{L}_{X_4}\mathcal{L}_{X_3}\mathcal{L}_{X_2}\mathcal{L}_{X_1} \mathcal{O}_{S_0}.
\end{equation*} Intuitively the operator grows as follows: the left end first moves from site $4$ to $3$ by $X_1$, and then from $3$ to $2$ by $X_3$; similarly on the right end, $X_4$ moves us from 6 to 8.  So, in order to grow the initial domain $S_0=\lbrace 4,5,6\rbrace$ by two sites on each end, we must traverse:
\begin{equation*}
\label{eq:leftright}
    \left\{\begin{aligned}
    \mathrm{left} &: 4\xrightarrow{X_1} 3 \xrightarrow{X_3} 2\\
    \mathrm{right} &: 6\xrightarrow{X_3} 7 \xrightarrow{X_4} 8
    \end{aligned}\right..
\end{equation*}
This pattern only depends on the boundary terms $H_X$ (red rectangles in Fig.~\ref{fig:d1}). 
However, observe that in general the Taylor expansion of $\mathrm{e}^{\mathcal{L}t}\mathcal{O}_{S_0}$ will contain many additional terms which act entirely inside of $S_0$.  We do not want to count these terms, since they cannot grow the support of the operator at all. The key observation, first made in \cite{chen2021operator}, is that one can elegantly classify all of the possible orderings for the ``red" terms in $H$ (those that grow the operator), in such a way that all possible intermediate sequences of yellow terms can be re-exponentiated to form a unitary operation (which leaves operator norms invariant)!  The practical consequence of this observation is that we only need to bound the contributions of red terms (and the number of possible patterns of red terms) when building a Lieb-Robinson bound for the $\kappa$-norm.  We emphasize that it is crucial that we track both the left and right end: the main technical issue addressed in this section is how to find such a ``direction-resolved'' Lieb-Robinson bound in $d>1$.  

\begin{figure}[t]
\centering
\includegraphics[width=.85\textwidth]{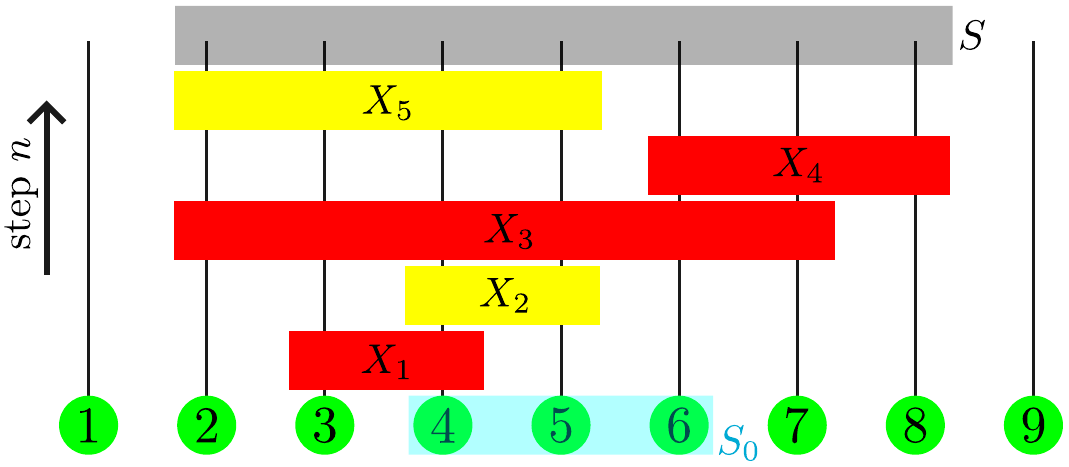}
\caption{\label{fig:d1}
A sketch of the Heisenberg evolution for an operator, which is equivalent to step-by-step taking its commutator with local terms $H_X$ of the Hamiltonian $H$. Starting from the initial operator supported on $S_0=\{4,5,6\}$, some $H_{X_n}$s truly grow the support (red rectangles) to finally reach $S=\{2,3,\cdots,8\}$, while others do not grow the support and serve as internal rotations (yellow rectangles). Only the former, which intersect with the boundary of the operator at each step, contributes to enlarging the $\kappa$-norm of the final operator. }
\end{figure}

\subsection{Irreducible skeleton representation}

We first focus on a single initial operator $\OO=\OO_{S_0}$, locally supported in a fixed $d$-dimensional simplex $S_0$. Recall that we work with the simplicial lattice, and the Hamiltonian $H=\sum_X H_X$ is also expanded in the simplices $X$ on which operators are supported. All simplices are regular and have the same orientation, as shown by the magenta triangles in Fig.~\ref{fig:simplicial}. 

To describe how the Hamiltonian couples different lattice sites, we
use the factor graph theory approach developed in \cite{chen2021operator}.   A factor graph $G=(\Lambda, F, E)$ is defined as follows: $\Lambda$ contains all of the lattice sites (as before), while $F$ is the set of all factors $X$ that appear in $H=\sum_X H_X$. Their union $\Lambda \cup F$ serves as the vertices of the factor graph $G$, and we connect each factor $X\subseteq F$ to all its contained sites $i\in X$.  Hence, the edges connecting $F$ and $\Lambda$ form the edge set $E$ of the factor graph $G$. As an example, in Fig.~\ref{fig:skeleton}(a) the green circles are lattice sites and the red rectangles are factors. Rectangles representing nearest-neighbor three-site interaction are connected to their supporting sites by black lines, while the connection for the six-site interactions are explicitly shown only for one of its supporting sites. 

We expand $\ee^{t\LL}\OO$ in powers of $t$: \begin{equation}\label{eq:Ln}
    \ee^{t\LL}\OO = \sum_{n=0}^\infty \frac{t^n}{n!} \LL^n\OO.
\end{equation}
Further expanding $\LL$ into factors using $\LL=\sum_X \LL_X$, each term is of the form $\LL_{X_n}\cdots\LL_{X_1} \OO$, where $\mathcal{L}_X := \mathrm{i}\mathrm{Adj}_{H_X}$. Crucially, observe that such a term is nonzero only if the sequence $T=(S_0,X_1,\cdots,X_n)$ obeys a \emph{casual structure}: each $X_m$ intersects with the previous support $S_0\cup X_1\cup\cdots\cup X_{m-1}$.\footnote{If this does not happen, then the expression must vanish as it contains within it a commutator of operators supported on disjoint sets.} If this condition holds, we say $T$ is a causal tree.  The tree structure, embedded in the factor graph $G$, is thoroughly defined in \cite{chen2021operator}; in what follows we will borrow closely from their formalism. We see that \eqref{eq:Ln} only contains contributions from the sequences $\mathcal{L}_{X_n}\cdots \mathcal{L}_{X_1}$ arising from causal trees $T$ whose root is $S_0$.  Defining $\mathcal{T}_{S_0}$ to be the set of all causal trees that have a root in $S_0$, we write:
\begin{equation}\label{eq:sumT}
    \ee^{t\LL}\OO = \sum_{T\in \mathcal{T}_{S_0}} \frac{t^n}{n!} \LL_{X_n}\cdots\LL_{X_1} \OO.
\end{equation} 
For a given causal tree $T$, let us denote with $S(T)$ the unique (smallest) circumscribed simplex of $T$. For example, the circumscribed simplex for the subset $\{2,3,4,5,9\}\subseteq \Lambda$ in Fig.~\ref{fig:skeleton}(a) is the triangle $S=\{1,2,\cdots,10\}$, since the subset touches all three faces of $S$. Then \eqref{eq:sumT} gives rise to a decomposition \begin{equation}\label{eq:sumS}
    \ee^{t\LL}\OO = \sum_{S\supset S_0} \sum_{T\in\mathcal{T}_{S_0} : S(T)=S} \frac{t^n}{n!} \LL^n\OO  := \sum_{S\supset S_0} \lr{\ee^{t\LL}\OO}_{S},
\end{equation}
of the evolved operator to simplices $S$.  
Note that some local terms of the evolved operator are assigned to a looser support $S$, if the operator accidentally becomes the identity operator at sites in $\partial S$. For example, if $\OO=Z_1Z_2$ and $X_1=Z_1X_2$, then $\LL_{X_1}\OO\propto I_1 Y_2$ acts trivially on site $1$, which we however assign a support of $S=\{1,2\}$ in the decomposition \eqref{eq:Ln}.

\begin{figure}[t]
\centering
\includegraphics[width=.85\textwidth]{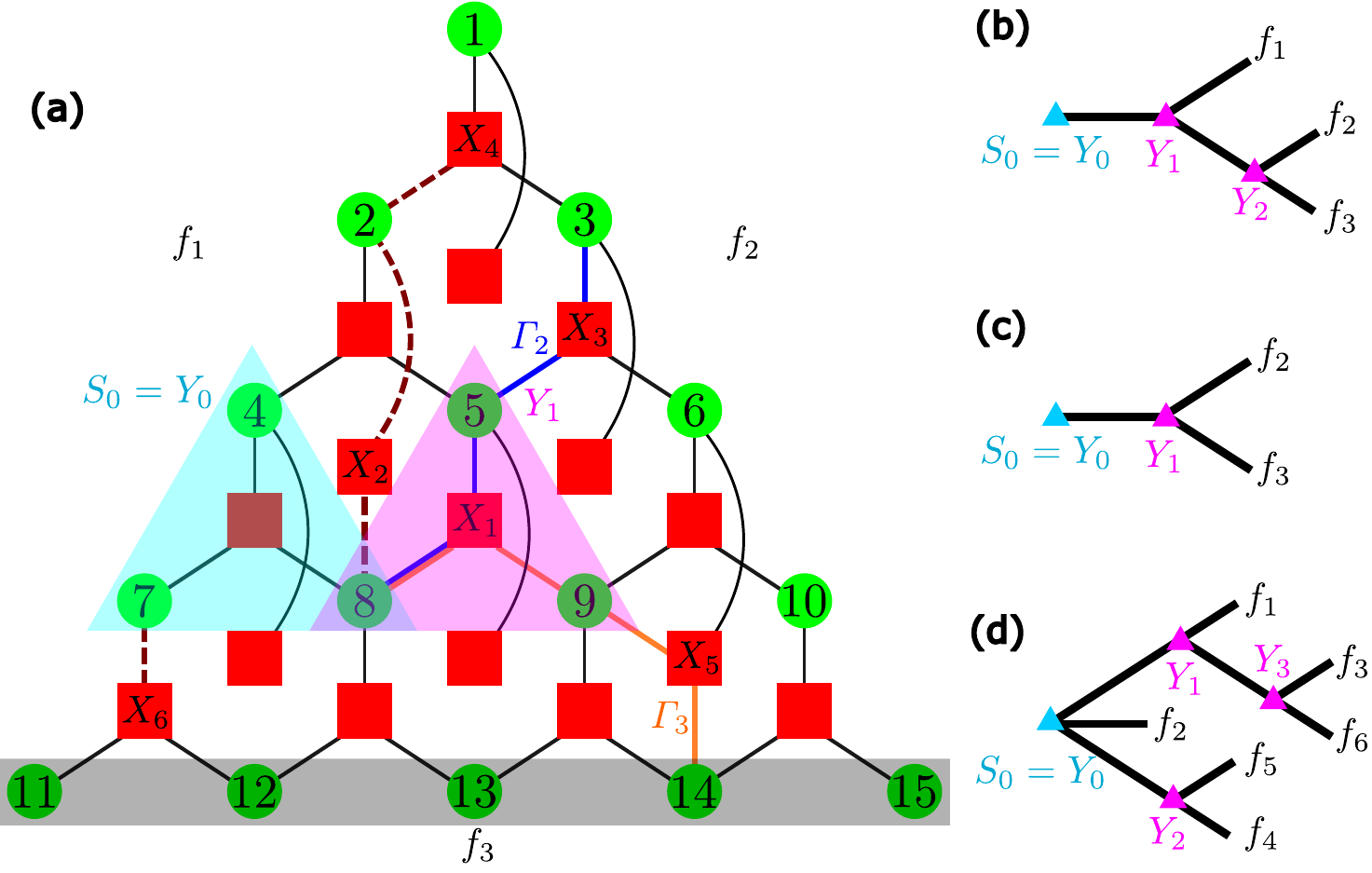}
\caption{\label{fig:skeleton} (a) Example of the factor graph and the irreducible skeleton construction at $d=2$. The lattice sites are numbered green circles composing a triangular lattice. We consider two kinds of factors here represented by red squares: three-site interactions in $H$ supported on the smallest triangles, and six-site interactions on the second smallest triangles. In the factor graph $G$, each factor is connected to all of the lattice sites that support it. In the figure, however, three-site square factors are connected to all of their lattices sites by black lines, while only one connection edge is explicitly shown for each six-site square factor. 
Having constructed the factor graph, we consider operator growth from $S_0=\{4,7,8\}$ (the cyan triangle) to the whole triangle $S=\{1,2,\cdots,15\}$. In particular, we consider a causal tree $T=(S_0,X_1,\cdots, X_6)$ that starts from $S_0$ and reaches the three faces of $S$, i.e., $S(T)=S$. Since $S_0$ already intersects with the face $f_1=\{1,2,4,7,11\}$, one only needs to figure out the unique irreducible paths from $S_0$ to $f_2=\{1,3,6,10,15\}$ and $f_3=\{11,12,\cdots,15\}$. For example, the irreducible path from $S_0$ to $f_2$ is $\varGamma_2=8\rightarrow X_1 \rightarrow 5\rightarrow X_3 \rightarrow 3$ (or equivalently, $6$), shown in blue solid lines. The reason is that $f_2$ is first reached by $X_3$, and because $X_3$ grows from $5$, one needs to figure out how $5$ is first reached: it is reached by $X_1$ that acts nontrivially on $8\in S_0$. Similarly, the irreducible path from $S_0$ to $f_3$ is the orange path $\varGamma_3=8\rightarrow X_1\rightarrow 9\rightarrow X_5\rightarrow 14 (13, 15)$. 
Merging these two paths, we find the ordered irreducible skeleton $\gamma=(X_1^\gamma,X_2^\gamma,X_3^\gamma)=(X_1,X_3,X_5)$ that $T$ belongs. The length is $\ell=3$ The corresponding non-ordered irreducible skeleton is thus $\varGamma=\{X_1,X_3,X_5\}$, which coincidentally, only has $\gamma$ as its ordered counterpart. The other factors $\{X_2,X_4,X_6\}$ in the casual tree $T$, are linked to the skeleton by causality, as shown by brown dashed lines. They are allowed factors in $\LL_1^\gamma, \LL_2^\gamma,\LL_3^\gamma$ respectively, defined in \eqref{eq:Lkgamma}, which are resummed to unitaries inserted between steps of the irreducible skeleton $\gamma$, in \eqref{eq:skeleton}. 
In Subsection \ref{sec:Y}, we make further definitions. $\varGamma$ has two bifurcation factors: $Y_0=S_0$, and $Y_1=X_1$ (pink triangle) where the two paths $\varGamma_2$ and $\varGamma_3$ merge. As for the irreducible paths, $\cP_1$ is trivial (length $l_1=0$) because $S_0$ intersects with $f_1$. $\cP_2$ and $\cP_3$, which start from $Y_1$, are $5\rightarrow X_3\rightarrow 3(6)$ and $9\rightarrow X_5\rightarrow 14(13,15)$ respectively. The paths $\cP'_1, \cP'_2$ connecting $Y_p$s are both trivial.
(b-d) Topological structure of irreducible skeletons. (b-c) are for $2d$, with (c) being the special case that $l_1=0$, as the example in (a). (d) A more complicated example for $5d$. }
\end{figure}

The next crucial step is to realize that the causal trees $T$ can be re-summed in an elegant way by sorting them into equivalence classes based on \emph{irreducible} subsequences of $T$. 
 Given $T$ and its circumscribed simplex $S(T)$, for any of the $d+1$ faces $f_1,f_2,\cdots,f_{d+1}$ of $S(T)$, there is a unique self-avoiding path $\varGamma_p\subseteq T$ that starts from the initial support $S_0$ and ends at the face $f_p$ \cite{chen2021operator}. For the casual tree example in Fig.~\ref{fig:skeleton}(a), the path $\varGamma_1$ is trivial containing only the root $S_0$, while $\varGamma_2=S_0\ni 8\rightarrow X_1 \rightarrow 5\rightarrow X_3 \rightarrow 3\in f_2$ (or equivalently, arriving at $6\in f_2$). 
The $d+1$ paths $\varGamma_p$ form a tree $\varGamma$, of the form
in Fig.~\ref{fig:skeleton}(b) for 2d, which we call the irreducible skeleton $\varGamma(T)$ for the causal tree $T$. Here ``irreducible'' means self-avoiding. Note that some of the branches of the irreducible skeleton may be absent, for example $\varGamma$ looks like Fig.~\ref{fig:skeleton}(c) in $2d$, if $S_0$ already shares one face of $S$. For the casual tree example in Fig.~\ref{fig:skeleton}(a), the irreducible skeleton is \begin{equation}
    S_0\ni 8 \text{ \textemdash}\, X_1 \left\langle \, \begin{aligned}
    &5 \text{ \textemdash}\, X_3 \text{ \textemdash}\, 3\in f_2 \\
    &9 \text{ \textemdash}\, X_5 \text{ \textemdash}\, 14\in f_3 \\
    \end{aligned}\right.
\end{equation}

$\varGamma(T)$ is a well-defined function of $T$, provided that
\begin{prop}\label{prop:tree}
For each causal tree $T$, there is a unique irreducible skeleton $\varGamma$.
\end{prop}
\begin{proof}
$\varGamma$ is unique because it is composed by the unique irreducible paths $\varGamma_p$ from $S_0$ to $f_p$. For the uniqueness of $\varGamma_p$, see Proposition 2 of \cite{chen2021operator}. Intuitively, for a given face $f_p$, one can first uniquely determine the first factor $\tilde{X}_1$ in $T$ that hits it. Next, there is a unique factor $\tilde{X}_2$ in $T$ that first hits the support of $\tilde{X}_1$. Further tracing back along the casual tree $T$, the unique $\varGamma_p$ is then $S_0\rightarrow\cdots \rightarrow \tilde{X}_2\rightarrow \tilde{X}_1$. 
\end{proof}
As a result, we can define $[\varGamma]\subseteq \mathcal{T}_{S_0}$ as the equivalence class of causal trees that have $\varGamma$ as its irreducible skeleton, and define $\mathbf{\Gamma}(S)$ as the set of all $\varGamma$ that have $S$ as their circumscribed equilateral triangle, with each face of $S$ touched only once by $\varGamma$, and only by its $d+1$ end points. 

Each irreducible skeleton $\varGamma$ corresponds to factors $\lbrace X_1^\varGamma,\cdots, X_\ell^\varGamma\rbrace$, where $\ell = \ell(\varGamma)$ is the ``length'' of $\varGamma$. However, their sequence acting on the operator can vary because different branches of $\varGamma$ can grow ``independently'': the only constraint is causality, that each factor in the irreducible path $\Gamma_p$ to a fixed face must appear in that same order in $\Gamma$.  Thus we define $\gamma=(X_1^\gamma,\cdots, X_\ell^\gamma)$ as the ordered irreducible skeleton, and use $\gamma \in \varGamma$ to mean $\gamma$ is one ordered skeleton of $\varGamma$. Each $\gamma$ thus corresponds to a finer equivalence class $[\gamma]\subseteq \mathcal{T}_{S_0}$ of casual trees $T$. Now we are prepared to give our re-summation lemma.

\begin{lem}The following identities hold:
\begin{align}\label{eq:skeleton}
    \lr{\ee^{t\LL}\OO}_S &= \sum_{\varGamma\in \mathbf{\Gamma}(S)}\lr{\ee^{t\LL}\OO}_\varGamma, \quad \mathrm{where} \nonumber\\ 
    \lr{\ee^{t\LL}\OO}_\varGamma &= \sum_{\gamma \in \varGamma} \int\limits_{\mathbb{T}^\ell(t)} \dd t_1\cdots \dd t_\ell \, \ee^{(t-t_\ell)\LL^\gamma_\ell}\LL_{X^\gamma_\ell}\cdots \ee^{(t_2-t_1)\LL^\gamma_1} \LL_{X^\gamma_1} \ee^{t_1\LL^\gamma_0}\OO,
\end{align}
with $\ell=\ell(\varGamma)=\ell(\gamma)$, and \begin{equation}
    \mathbb{T}^\ell(t) = \{(t_1,\cdots,t_\ell)\in [0,t]^\ell: t_1\le t_2\le \cdots \le t_\ell\}.
\end{equation}
The set of allowed factors between steps $k$ and $k+1$ are included in \begin{align}\label{eq:Lkgamma}
    \LL^\gamma_k = \LL|_{R^\gamma_k} - \sum_{Y\subseteq R^\gamma_k: Y\cap W^\gamma_k\neq \emptyset}\LL_Y,
\end{align}
where \begin{equation}\label{eq:Wk}
    W^\gamma_k = \bigcup_{m=k+2: X^\gamma_m\cap \gamma_k =\emptyset }^{\ell(\gamma)} X^\gamma_{m}, \quad \mathrm{where} \quad \gamma_k := S_0 \cup X^\gamma_1 \cup \cdots \cup X^\gamma_k,
\end{equation}
and \begin{equation}\label{eq:Rk}
    R^\gamma_k = \left\{ \begin{array}{ll}
        S, & k=\ell \\
        S-\bigcup_{p'=p}^{d+1} f_{p'}^\gamma, & \ell_{p-1} \le k=\ell_{p-1} < \ell_p, \quad (p=1,\cdots, d+1) \\
    \end{array} \right.
\end{equation}
where $\gamma$ first reaches the faces in order $(f_1^\gamma,f_2^\gamma,\cdots,f_{d+1}^\gamma)$, by factors $(X^\gamma_{\ell_1},X^\gamma_{\ell_2},\cdots,X^\gamma_{\ell_{d+1}})$ respectively. Here $\ell_0 = 0, \ell_{d+1}=\ell$, and $R^\gamma_k$ is always a simplex.
\end{lem}

\begin{proof}
Define \begin{equation}\label{eq:O_Gamma}
    \lr{\ee^{t\LL}\OO}_\varGamma := \sum_{T\in [\varGamma]} \frac{t^n}{n!} \LL_{X_n}\cdots\LL_{X_1} \OO = \sum_{\gamma\in \varGamma} \sum_{T\in [\gamma]} \frac{t^n}{n!} \LL_{X_n}\cdots\LL_{X_1} \OO.
\end{equation}
The first equality in \eqref{eq:skeleton} holds because each sequence $T$ is a causal tree and by Proposition \ref{prop:tree} each $T$ belongs to a unique equivalence class. Now we need to verify the second equality of (\ref{eq:skeleton}). To do so, we first prove \begin{equation}\label{eq:m0ml}
    \sum_{T\in [\gamma]} \frac{t^n}{n!} \LL_{X_n}\cdots\LL_{X_1} \OO = \sum_{m_0,\cdots,m_\ell=0}^\infty \frac{t^{\ell+m_0+\cdots+m_\ell} }{\lr{\ell+m_0+\cdots+m_\ell}!} \lr{\LL^\gamma_\ell}^{m_\ell}\LL_{X^\gamma_\ell}\cdots \lr{\LL^\gamma_1}^{m_1} \LL_{X^\gamma_1} \lr{\LL^\gamma_0}^{m_0}\OO,
\end{equation}
by identifying terms on the two sides, as a generalization of Lemma 4 in \cite{chen2021operator}. Observe that when expanding $\LL^\gamma_k$ into factors using \eqref{eq:Lkgamma}, each term on the right hand side forms a casual tree $T$ that has $\gamma$ as its ordered irreducible skeleton. Indeed, before acting with $\LL_{X^\gamma_{k+1}}$, all previous factors included in $\LL^\gamma_{0},\cdots,\LL^\gamma_{k}$ do not touch $W^\gamma_k$, 
the future skeleton factors that have not been touched by the previous skeleton factors $S_0,X_1^\gamma,\cdots,X_k^\gamma$. Thus when tracing back through $T$, the irreducible skeleton has to pass the skeleton factors $X^\gamma_k$, not any of the factors contained in $\LL^\gamma_k$. To summarize, each term on the right hand side, which is a casual tree $T\in [\gamma]$, appears on the left hand side with an identical prefactor due to $n=\ell+m_0+\cdots+m_\ell$, and appears only once. 

\eqref{eq:m0ml} then holds if the left hand side contains no more terms, i.e., each $T\in [\gamma]$ can be expressed as one term on the right hand side (when expanding $\LL_k^\gamma$ to factors). To prove this statement, consider all possible ``candidates'' that may have $\gamma$ as its ordered irreducible skeleton. They must belong to terms in \begin{equation}
    \lr{\LL}^{m_\ell}\LL_{X^\gamma_\ell}\cdots \lr{\LL}^{m_1} \LL_{X^\gamma_1} \lr{\LL}^{m_0}\OO,
\end{equation}
with $m_0,\cdots,m_\ell$ being non-negative integers. However, not all terms in each $\LL$ are allowed, because they may modify the irreducible skeleton to another one rather than $\gamma$. For each $(\LL)^{m_k}$ before acting with $\LL_{X^\gamma_{k+1}}$, there are three kinds of them. First, any factor not contained in $S$ will make the irreducible skeleton different from $\gamma$, because it will not have $S$ as its circumscribed simplex. Second, any factor not contained in $R^\gamma_k$ defined in \eqref{eq:Rk} will also modify the irreducible skeleton, because it touches one of the faces $f_p,\cdots,f_{d+1}$ of $S$ that ought to be touched at the first time by later skeleton factors $X^\gamma_{k+1},\cdots,X^\gamma_{\ell}$. Third, any factor $Y$ that touches $W^\gamma_k$ in \eqref{eq:Wk} is also not allowed. Suppose $Y$ touches $X_{k+3}^\gamma$ that does not overlap with $\gamma_k$: $X_{k+3}^\gamma\cap \gamma_k = \emptyset$. Then when tracing backwards the casual tree to find the irreducible skeleton, $X_{k+3}^\gamma$ should be traced back to $Y$ instead of some factor in $S_0,X_1^\gamma,\cdots,X_{k+2}^\gamma$, making the irreducible skeleton different than $\gamma$. Indeed, these are the only three cases that are avoided in \eqref{eq:Lkgamma}, so any $T\in [\gamma]$ must appear on the right hand side of \eqref{eq:m0ml}.

As a result, we have constructed an one-to-one correspondence between the terms on the two sides of \eqref{eq:Lkgamma}, so the two sides are equal given that the prefactors always agree. We then ``exponentiate'' the right hand side using Lemma 5 in \cite{chen2021operator}, to get \begin{equation}\label{eq:exponentiate}
    \sum_{T\in [\gamma]} \frac{t^n}{n!} \LL_{X_n}\cdots\LL_{X_1} \OO = \int_{\mathbb{T}^\ell(t)} \dd t_1\cdots \dd t_\ell \, \ee^{(t-t_\ell)\LL^\gamma_\ell}\LL_{X^\gamma_\ell}\cdots \ee^{(t_2-t_1)\LL^\gamma_1} \LL_{X^\gamma_1} \ee^{t_1\LL^\gamma_0}\OO.
\end{equation}
\eqref{eq:skeleton} then follows from \eqref{eq:O_Gamma} and \eqref{eq:exponentiate}.
\end{proof}

Taking the operator norm of \eqref{eq:skeleton} using $\norm{\LL_X \OO} \le 2\norm{H_X} \norm{\OO}$, we obtain:
\begin{cor}
The operator norm of \eqref{eq:skeleton} is bounded by \begin{equation}\label{eq:NGamma}
    \norm{\lr{\ee^{t\LL}\OO}_S} \le \norm{\OO} \sum_{\varGamma\in \mathbf{\Gamma}(S)} \frac{(2|t|)^\ell}{\ell!}N(\varGamma) \prod_{k=1}^\ell \norm{H_{X^\varGamma_k}},
\end{equation}
where $N(\varGamma)$ is the number of ordered $\gamma$ that corresponds to $\varGamma$. 
\end{cor}

\subsection{Proof of Proposition \ref{prop:LO-O} (restatement)}
\begin{proof}[Proof of Proposition \ref{prop:LO-O} (restatement)]
We assume $t\ge 0$ without loss of generality in the remainder of this section. To bound $\norm{\ee^{t\LL}\OO-\OO}_{\kappa'}$,
we choose the local decomposition for $\ee^{t\LL}\OO$ in the previous section, namely 
\begin{equation}\label{eq:S0S}
    \lr{\ee^{t\LL}\OO}_S = \sum_{S_0\subseteq S}\lr{\ee^{t\LL}\OO_{S_0}}_S,
\end{equation}
where each term can be further expanded to irreducible skeletons \eqref{eq:skeleton}. The decomposition $\OO=\sum_{S_0}\OO_{S_0}$ is chosen as the optimal one that realizes $\norm{\OO}_\kappa$. The bound for this decomposition would immediately lead to a bound on the infimum \eqref{eq:norm} over all possible decompositions of $\norm{\ee^{t\LL}\OO-\OO}_{\kappa'}$.
We divide the terms in \eqref{eq:S0S} into two cases: $S_0=S$ for case $1$ and $S_0\neq S$ for case 2. We will further divide case 2 to 2A and 2B. We prove this Proposition \ref{prop:LO-O} by showing that the contribution to $\norm{\ee^{t\LL}\OO - \OO}_{\kappa',i}$ from each of the three cases, is bounded by the form of the right hand side of \eqref{eq:LO-O} that does not depend on $i$. Recall that $\norm{\ee^{t\LL}\OO - \OO}_{\kappa',i}$ defined in \eqref{eq:kappai}, probes the $\kappa'$-norm of $\ee^{t\LL}\OO-\OO$ at site $i\in \Lambda$. The bound \eqref{eq:LO-O} for $\norm{\ee^{t\LL}\OO - \OO}_{\kappa'}$ then follows by a trivial maximization over $i$. We begin with the simplest case 1.

\textbf{Case 1:} If $S_0=S$, we have \begin{equation}\label{eq:LOSS}
    \lr{\ee^{t\LL}\OO_{S}}_S = \ee^{t\LL|_S}\OO_S,
\end{equation}
where \begin{equation}
    \LL|_S:=\sum_{X\subseteq S} \LL_X.
\end{equation}
The corresponding local terms $H_X$ are chosen as the optimal decomposition that realizes $\norm{H}_\kappa$.
\eqref{eq:LOSS} cancels $\OO_S$ at zeroth order of $t$: \begin{align}\label{eq:LO-O<HO}
    \norm{\lr{\ee^{t\LL}\OO_{S}}_S - \OO_S} = \norm{\ee^{t\LL|_S}\OO_S-\OO_S} \le 2t\norm{H|_S}\norm{\OO_S} \le 2t\norm{\OO_S} \sum_{j\in S}\sum_{X\ni j} \norm{H_X} \le 2t\norm{\OO_S}|S|\norm{H}_{\kappa},
\end{align}
which follows the treatment in \eqref{eq:AOA-O<}.

An operator $\lr{\ee^{t\LL}\OO_{S}}_S - \OO_S$ can contribute to $\norm{\ee^{t\LL}\OO - \OO}_{\kappa',i}$ only if $S\ni i$, so the total contribution of such ``local rotations'' (case 1 terms in \eqref{eq:S0S}) to $\norm{\ee^{t\LL}\OO - \OO}_{\kappa',i}$ is \begin{align}\label{eq:S0=S}
    \sum_{S\ni i}\norm{\lr{\ee^{t\LL}\OO_{S}}_S - \OO_S}\ee^{\kappa'(\mathrm{diam} S)^\alpha} &\le 2t\norm{H}_{\kappa}\sum_{S\ni i} |S|\norm{\OO_S}\ee^{\kappa'(\mathrm{diam} S)^\alpha} \nonumber\\
    &\le 2c_{\mathrm{vol}} t\norm{H}_{\kappa}\sum_{S\ni i}(\delta\kappa)^{-d/\alpha} \ee^{\delta \kappa(\mathrm{diam} S)^\alpha} \norm{\OO_S}\ee^{\kappa'(\mathrm{diam} S)^\alpha} \nonumber\\ &= 2c_{\mathrm{vol}} (\delta\kappa)^{-d/\alpha} t\norm{H}_{\kappa}\norm{\OO}_\kappa.
\end{align}
Here we have used \eqref{eq:LO-O<HO} and \eqref{eq:cd}: \begin{equation}\label{eq:vol<d}
    |S| \le c_d (1+\mathrm{diam} S)^d \le c_{\mathrm{vol}}(\delta\kappa)^{-d/\alpha} \ee^{\delta \kappa(\mathrm{diam} S)^\alpha},
\end{equation}
for some constant $c_{\mathrm{vol}}$ that only depends on $d$ and $\alpha$. We also used the optimality of decomposition $\OO_S$ in the final step. Therefore, the case 1 contribution \eqref{eq:S0=S} to $\norm{\ee^{t\LL}\OO - \OO}_{\kappa',i}$ is bounded by the right hand side of \eqref{eq:LO-O}, since $\delta \kappa$ is bounded from above by an O(1) constant, so we can always make the power of $\delta \kappa$ more negative: \begin{equation}
c_{\mathrm{vol}} (\delta \kappa)^{-d/\alpha} \le c_{\mathrm{vol}} \left(\frac{\kappa_1}{\delta \kappa}\right)^{\frac{d-1}{\alpha}}(\delta \kappa)^{-d/\alpha} = c^\prime_{\mathrm{vol}} (\delta \kappa)^{-\frac{2d-1}{\alpha}}.
\end{equation}

\textbf{Separate case 2 to 2A and 2B:}
For a given $\OO_{S_0}$, define its contribution $K_{\kappa',i}(\OO_{S_0})$ to $\norm{\ee^{t\LL}\OO - \OO}_{\kappa',i}$ that is beyond the local rotation above, as  
\begin{align}\label{eq:KS0}
    K_{\kappa',i}(\OO_{S_0}) :=& \sum_{S\ni i: S_0 \subsetneqq S} \norm{\lr{\ee^{t\LL}\OO_{S_0}}_S} \ee^{\kappa'(\mathrm{diam}S)^\alpha} = \sum_{S\ni i: S_0 \subsetneqq S} \norm{\lr{\ee^{t\LL}\OO_{S_0}}_S} \ee^{\lr{\kappa''-\frac{\delta\kappa}{2}}(\mathrm{diam}S)^\alpha} \nonumber\\
    \le& \ee^{-\frac{\delta\kappa}{2}\mathsf{d}^\alpha(i,S_0)} \sum_{S: S_0 \subsetneqq S} \norm{\lr{\ee^{t\LL}\OO_{S_0}}_S} \ee^{\kappa''(\mathrm{diam}S)^\alpha},
\end{align}
where we have defined \begin{equation}\label{eq:kappa''}
    \kappa''=\kappa'+\frac{\delta\kappa}{2} = \kappa-\frac{\delta\kappa}{2},
\end{equation}
and used Markov inequality like \eqref{eq:markov} with $\mathrm{diam}S \ge \mathsf{d}(i,S_0)$, since $\{i\}\cup S_0\subseteq S$. We have also relaxed the restriction that $S$ must contain $i$, since the prefactor $\ee^{-\frac{\delta\kappa}{2}\mathsf{d}^\alpha(i,S_0)}$ already suppresses contributions from $S_0$ that are faraway from $i$.
According to the irreducible skeleton expansion \eqref{eq:skeleton}, \eqref{eq:KS0} becomes \begin{align}\label{eq:K123}
    \ee^{\frac{\delta\kappa}{2}\mathsf{d}^\alpha(i,S_0)} K_{\kappa',i}(\OO_{S_0}) &\le \sum_{\ell=1}^\infty \sum_{\varGamma:\ell(\varGamma)=\ell} \norm{\lr{\ee^{t\LL}\OO_{S_0}}_\varGamma}\ee^{\kappa''(\mathrm{diam}S(\varGamma))^\alpha} \equiv \sum_{\ell=1}^\infty K_\ell(\OO_{S_0}),
\end{align}
where $K_\ell(\OO_{S_0})$ contains contributions from irreducible skeletons of length $\ell$. Note that $K_\ell(\OO_{S_0})$ no longer cares about whether $i$ is reached or not. We further separate \eqref{eq:K123} to the sum of $K_1(\OO_{S_0})$ and $\sum_{\ell=2}^\infty K_\ell(\OO_{S_0})$, whose contributions to $\norm{\ee^{t\LL}\OO - \OO}_{\kappa',i}$ are called case 2A and 2B respectively.

\textbf{Case 2A:}
We compute $K_1(\OO_{S_0})$ using \eqref{eq:NGamma}, where each irreducible skeleton $\varGamma$ corresponds to a single factor $X$ that grows $S_0$ to a strictly larger simplex $S(X)$ so that $N(\varGamma)=1$. \begin{align}\label{eq:K1<}
    K_1(\OO_{S_0}) &= \sum_{\varGamma:\ell(\varGamma)=1} \norm{\lr{\ee^{t\LL}\OO_{S_0}}_\varGamma}\ee^{\kappa''(\mathrm{diam}S(\varGamma))^\alpha} \nonumber\\
    &\le 2t\norm{\OO_{S_0}} \sum_{X: X\cap S_0\neq \emptyset, X\cap S^c_0\neq \emptyset} \norm{H_X}\ee^{\kappa''(\mathrm{diam}S(X))^\alpha} \nonumber\\
    &\le 2t\norm{\OO_{S_0}} \sum_{X: X\cap S_0\neq \emptyset, X\cap S^c_0\neq \emptyset} \norm{H_X}\ee^{\kappa''\mlr{(\mathrm{diam}S_0)^\alpha + (\mathrm{diam}X)^\alpha}} \nonumber\\
    &\le 2t\norm{\OO_{S_0}} \sum_{j\in \partial S_0} \sum_{X\ni j} \norm{H_X}\ee^{\kappa''\mlr{(\mathrm{diam}S_0)^\alpha + (\mathrm{diam}X)^\alpha}} \nonumber\\
    &\le 2t\norm{\OO_{S_0}} |\partial S_0| \norm{H}_{\kappa} \ee^{\kappa''(\mathrm{diam}S_0)^\alpha}.
\end{align}
Here in the third line, we have used \eqref{eq:a+b<} with $\mathrm{diam}S(X) \le \mathrm{diam}S_0+\mathrm{diam}X$. In the fourth line, we have used the fact that every $X$ that grows $S_0$ must contain at least one vertex $j$ in the boundary $\partial S_0$.

We then bound the contribution from all $K_1(\OO_{S_0})$ to $\norm{\ee^{t\LL}\OO - \OO}_{\kappa',i}$, following the same strategy of \eqref{eq:<K+K}. Namely, all $\OO_{S_0}$ are grouped by their distance $x$ to the site $i$, and the site $j\in S_0$ that has distance $x$ to $i$. Using \eqref{eq:K1<}, the contribution is \begin{align}\label{eq:2A}
    \sum_{\OO_{S_0}} \ee^{-\frac{\delta\kappa}{2}\mathsf{d}^\alpha(i,S_0)} K_1(\OO_{S_0}) &\le \sum_{x=0}^\infty \ee^{-\frac{\delta\kappa}{2}x^\alpha} \sum_{j: \mathsf{d}(i,j)=x} \sum_{\OO_{S_0}}^{j,x} 2t\norm{\OO_{S_0}}|\partial S_0| \norm{H}_{\kappa}  \ee^{\kappa''(\mathrm{diam}S_0)^\alpha} \nonumber\\
    &\le 2c_1 t \lr{\delta\kappa}^{-\frac{d-1}{\alpha}} \norm{H}_{\kappa} \sum_{x=0}^\infty \ee^{-\frac{\delta\kappa}{2}x^\alpha} \sum_{j: \mathsf{d}(i,j)=x} \sum_{S_0\ni j} \norm{\OO_{S_0}}\ee^{\frac{\delta\kappa}{2}(\mathrm{diam}S_0)^\alpha} \ee^{\kappa''(\mathrm{diam}S_0)^\alpha} \nonumber\\
    &\le 2c_1 t \lr{\delta\kappa}^{-\frac{d-1}{\alpha}} \norm{H}_{\kappa} \sum_{x=0}^\infty \ee^{-\frac{\delta\kappa}{2}x^\alpha} \sum_{j: \mathsf{d}(i,j)=x} \norm{\OO}_\kappa \nonumber\\
    &\le 2c_1 t \lr{\delta\kappa}^{-\frac{d-1}{\alpha}} \norm{H}_{\kappa}\norm{\OO}_\kappa \sum_{x=0}^\infty c_\mathsf{d}(2x+1)^{d-1}\ee^{-\frac{\delta\kappa}{2}x^\alpha} \nonumber\\
    &\le c_{\mathrm{2A}} t \lr{\delta\kappa}^{-\frac{2d-1}{\alpha}} \norm{H}_{\kappa}\norm{\OO}_\kappa,
\end{align}
bounded by the form of the right hand side of \eqref{eq:LO-O}.
Here in the second line, we have used the surface analog of \eqref{eq:vol<d} with numerical factors adjusted: \begin{equation}\label{eq:partial<}
    |\partial S_0| \le c_d (1+\mathrm{diam} S_0)^{d-1} \le c_1(\delta\kappa)^{-\frac{d-1}{\alpha}} \ee^{\frac{\delta \kappa}{2}(\mathrm{diam} S_0)^\alpha}.
\end{equation}
In the third line we have used \eqref{eq:kappa''} and the $\kappa$-norm of $\OO$. 
\eqref{eq:cLam} is used in the fourth line. The final sum over $x$ in \eqref{eq:2A} is bounded by transforming to an integral, and rescaling $x\rightarrow (\delta\kappa)^{-1/\alpha}x$, so that \begin{equation}\label{eq:sumx}
    \sum_{x=0}^\infty (2x+1)^{d-1}\ee^{-\frac{\delta\kappa}{2}x^\alpha} \le \frac{c_{\mathrm{2A}} }{2c_1 c_d} \lr{\delta\kappa}^{-d/\alpha},
\end{equation}
for a constant $c_{\mathrm{2A}}$ determined by $d,\alpha$ and $\kappa_1$.



\textbf{Case 2B:}
It is more tedious to find higher orders $K_\ell(\OO_{S_0})$, because each irreducible skeleton has many branches. Fortunately, the assumption \eqref{eq:tH<C} simplifies the analysis by cancelling powers of $\delta\kappa$ from geometrical factors like $|S|$ and $|\partial S|$. This is achieved in the following Lemma, whose proof is the goal of later subsections.
\begin{lem}\label{lem:K<K1}
If \eqref{eq:tH<C} holds, then
there exists a constant $c_K$ depending on $d,\alpha$ and $\kappa_1$, such that
\begin{equation}\label{eq:K2+<}
    \sum_{\ell=2}^\infty K_\ell(\OO_{S_0}):=\sum_{\ell=2}^\infty \sum_{\varGamma:\ell(\varGamma)=\ell} \norm{\lr{\ee^{t\LL}\OO_{S_0}}_\varGamma}\ee^{\kappa''(\mathrm{diam}S(\varGamma))^\alpha} \le c_K  t\norm{\OO_{S_0}} \norm{H}_{\kappa} \ee^{\kappa(\mathrm{diam}S_0)^\alpha}.
\end{equation}
\end{lem}
Comparing to \eqref{eq:K1<}, \eqref{eq:K2+<} no longer has the surface factor $|\partial S_0|$. As a result, the contribution to $\norm{\ee^{t\LL}\OO - \OO}_{\kappa',i}$ in this 2B case is bounded by the same procedure of \eqref{eq:2A}: \begin{equation}\label{eq:2B}
    \sum_{\OO_{S_0}} \ee^{-\frac{\delta\kappa}{2}\mathsf{d}^\alpha(i,S_0)} \sum_{\ell=2}^\infty K_\ell(\OO_{S_0}) \le c_{\mathrm{2B}} t \lr{\delta\kappa}^{-\frac{d}{\alpha}} \norm{H}_{\kappa}\norm{\OO}_\kappa,
\end{equation}
which is smaller than \eqref{eq:2A} by a factor $(\delta\kappa)^{\frac{d-1}{\alpha}}$. In fact, the exponent in \eqref{eq:tH<C} is deliberately chosen such that \eqref{eq:2B} is of the same order as \eqref{eq:S0=S} for case 1, so that case 2A dominates at small $\delta\kappa$.

Summarizing the three cases 1, 2A and 2B, \eqref{eq:S0=S}, \eqref{eq:2A} and \eqref{eq:2B} combine to produce \eqref{eq:LO-O}, with the constant $c_-$ depending only on $d,\alpha$ and $\kappa_1$.
\end{proof}

\subsection{Bounding the contributions of irreducible skeletons}\label{sec:Y}
The rest of the subsections serve to prove Lemma \ref{lem:K<K1}, where we need to sum over all irreducible skeletons $\varGamma$, starting from a fixed local operator $\OO=\OO_{S_0}$. In this subsection, we transform this sum into more tractable sums over factors and paths. Recall that each skeleton $\varGamma$ is topologically a tree that have branches to reach all faces of the final simplex $S=S(\varGamma)$. Thus, we can identify the unique set of bifurcation factors \begin{equation}\label{eq:Y=Y0Y1}
    Y=\{Y_0\equiv S_0, Y_1, \cdots, Y_{l(Y)}\} \subseteq \{S_0, X_1^\varGamma, \cdots, X_\ell^\varGamma\},
\end{equation}
as the internal vertices (including the root) of the tree $\varGamma$, while the faces $f_1,\cdots,f_{d+1}$ are terminal vertices. In Fig.~\ref{fig:skeleton} for example, $Y=\{Y_0, Y_1\}$ where $Y_1=X_1$. $l(Y)=|Y|-1\le d$ is the number of distinct factors in the sequence at which the path to another face bifurcates (see Figure \ref{fig:Y} for an illustration); if $l(Y)=0$ it means that the paths never share any common $Y_n$s.  We do not care about the labeling order in $Y_n$, but keep track of their connection. I.e., for each $Y_n$, define $\mathsf{p}(Y_n)$ as its parent vertex in the tree, which is some bifurcation factor $Y_{n'}$. Similarly define $\mathsf{p}(f_p)$ as the parent bifurcation factor of face $f_p$. In $\varGamma$, 
$f_p$ is then reached by an irreducible path $\cP_p$  that starts from $\mathsf{p}(f_p)$. Each such path $\cP_p$ is a single branch of the skeleton $\varGamma$. Each bifurcation factor $Y_n$ is also connected to its parent $\mathsf{p}(Y_n)$ by an irreducible path $\cP'_n$. See Fig.~\ref{fig:Y}(a) as a sketch in $2$d. We let $\mathbf{Y}(S)$ be all possible sets $Y$ of bifurcation factors that support a $\varGamma\in \mathbf{\Gamma}(S)$. 

Furthermore, we define the sum over paths that bounds operator growth between two given sets $Z, Z'\subseteq \Lambda$ (not necessarily simplices). There are irreducible paths $\cP$ in the factor graph $G$ that connects them, which we label by $\cP:Z\rightarrow Z'$. We define the sum over such irreducible paths \begin{equation}\label{eq:CZZ}
    C_l(Z,Z') := \sum_{\cP:Z\rightarrow Z', \ell(\cP)= l} \frac{(2t)^{\ell}}{\ell!} \prod_{k=1}^\ell \norm{H_{X^\cP_k} },
\end{equation}
as a bound for the relative amplitude that an operator in $Z$ can grow to $Z'$, where $(X_1^\cP,\cdots,X_\ell^\cP)$ are the factors in $\cP$. Note that we also use the parameter $l$ to control the length of a path that can contribute. Depending on whether $Z$ intersects with $Z'$, we have \begin{equation}\label{eq:C0}
    C_l(Z,Z') \left\{ \begin{aligned} 
    &=\delta_{l 0},  &Z\cap Z'\neq \emptyset \\
    & \propto (1-\delta_{l 0}), \quad &Z\cap Z'= \emptyset
    \end{aligned} \right.
\end{equation}

With the definitions above and the short hand notation \begin{equation}\label{eq:JS}
    J(S) := \sum_{\varGamma\in \mathbf{\Gamma}(S):\ell(\varGamma)>1, } \norm{\lr{\ee^{t\LL}\OO}_\varGamma},
\end{equation}
we prove the following Lemma. Namely, the sum over skeletons in \eqref{eq:skeleton} is bounded by first summing over the bifurcation factors, and then summing over the irreducible paths that connect them to the faces of $S$ and to themselves, with the constraint that the length $\ell$ of the skeleton should agree with the number of factors contained in the bifurcation factors and the paths.

\begin{figure}[t!]
\centering
\includegraphics[width=.45\textwidth]{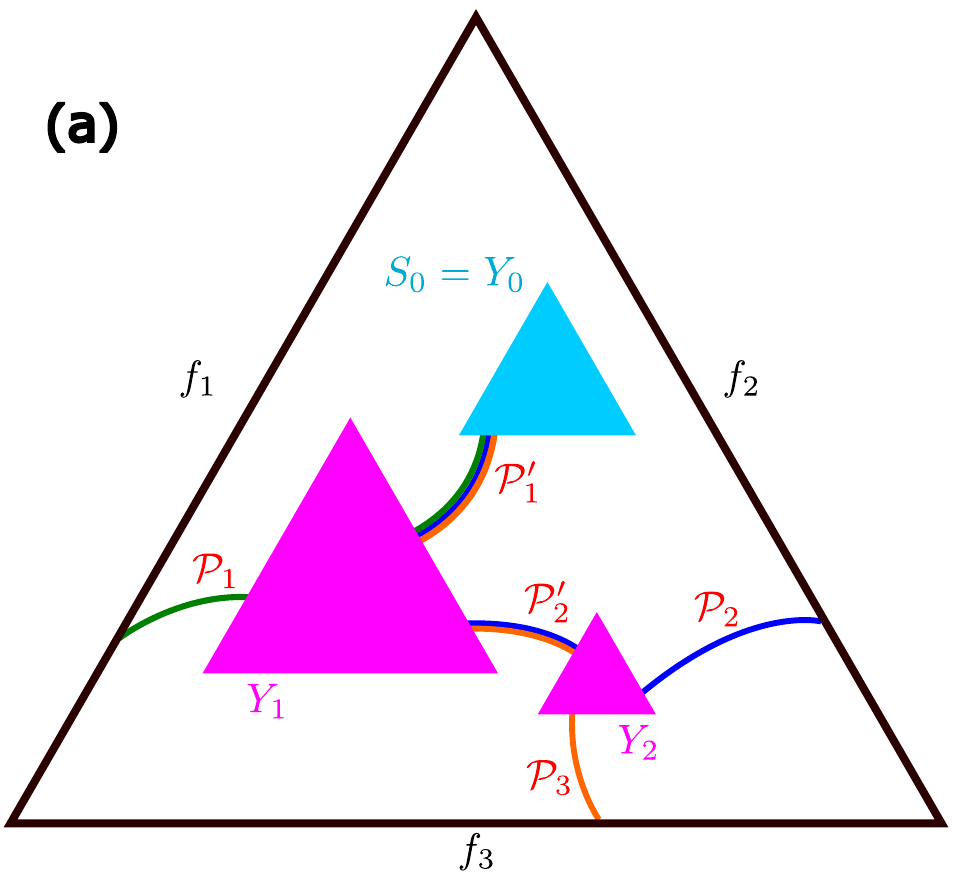}\qquad
\includegraphics[width=.45\textwidth]{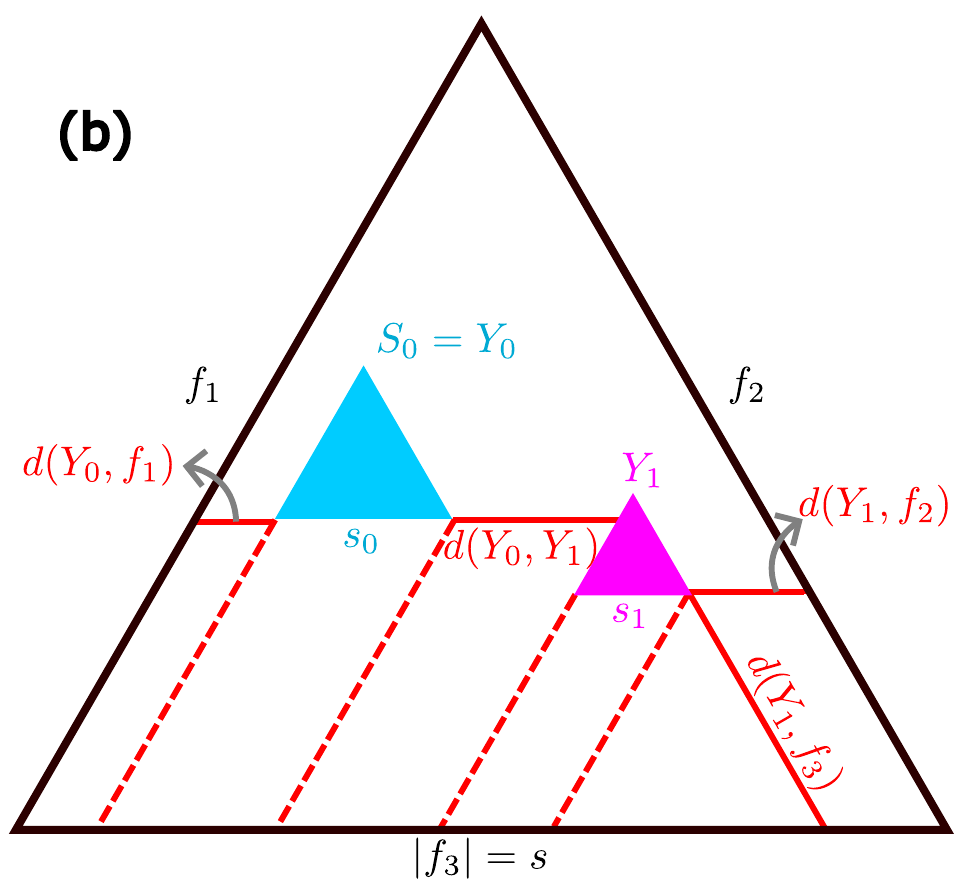}
\caption{\label{fig:Y}
(a) A $2$d sketch of bifurcation factors $Y_p$ and irreducible paths $\cP_p,\cP'_n$ that comprise an irreducible skeleton $\varGamma$. The irreducible paths $\varGamma_p$ from $S_0$ to the three faces are also denoted by lines of three different colors. (b) A $2$d sketch of why the inequality \eqref{eq:>s-sn} holds. In particular, here the $l(Y)=1$ case takes the equality, guided by the red dashed lines. }
\end{figure}

\begin{lem} The following identities hold: 
\begin{align}\label{eq:gamma<p}
    J(S) &\le (d+1)! \norm{\OO} \sum_{Y\in \mathbf{Y}(S)}(2t)^{l(Y)} \lr{ \prod_{n=1}^{l(Y)} \norm{H_{Y_n}} } \mathcal{G}(Y),\quad \mathrm{where} \nonumber\\
    \mathcal{G}(Y) &:= \sum_{\substack{l_1,\cdots,l_{d+1},l'_1,\cdots,l'_{l(Y)}\ge 0:\\ l(Y)+l_1+\cdots+l_{d+1}+l'_1+\cdots+l'_{l(Y)}>1}} \mlr{ \prod_{n=1}^{l(Y)} C_{l'_n}(Y_n, \mathsf{p}(Y_n))} \mlr{ \prod_{p=1}^{d+1} C_{l_p}(f_p, \mathsf{p}(f_p))}.
\end{align}
Here $\{f_1,\cdots,f_{d+1}\}$ are the faces of $S$, and $Y=\{Y_0,\cdots,Y_{l(Y)}\}$ is the set of bifurcation factors in \eqref{eq:Y=Y0Y1}.
\end{lem}

\begin{proof}
Suppose $\norm{\OO}=1$. According to \eqref{eq:JS} and \eqref{eq:NGamma}, \begin{align}\label{eq:gamma<I}
    J(S) &\le \sum_{\varGamma\in \mathbf{\Gamma}(S):\ell(\varGamma)\equiv \ell>1} N(\varGamma) \frac{(2t)^\ell}{\ell!} \prod_{k=1}^\ell \norm{H_{X^\varGamma_k}} \nonumber\\ 
    & \le (d+1)! \sum_{Y\in \mathbf{Y}(S)} \sum_{\cP_1:\mathsf{p}(f_1)\rightarrow f_1}\cdots \sum_{\cP_{d+1}:\mathsf{p}(f_{d+1})\rightarrow f_{d+1}} \sum_{\cP'_1:\mathsf{p}(Y_1)\rightarrow Y_1}\cdots \sum_{\cP'_{l(Y)}:\mathsf{p}(Y_{l(Y)})\rightarrow Y_{l(Y)}} N(\varGamma) \frac{(2t)^{\ell}}{\ell!}\prod_{k=1}^\ell \norm{H_{X^\varGamma_k}} \mathbb{I}(\ell>1) \nonumber\\
    &= (d+1)! \sum_{Y\in \mathbf{Y}(S)} (2t)^{l(Y)} \lr{ \prod_{n=1}^{l(Y)} \norm{H_{Y_n}} } \nonumber\\ &\quad \lr{\sum_{\cP_1:\mathsf{p}(f_1)\rightarrow f_1}(2t)^{l_1} \prod_{k=1}^{l_1} \norm{H_{X^{\cP_1}_k}} }\cdots \lr{\sum_{\cP_{d+1}:\mathsf{p}(f_{d+1})\rightarrow f_{d+1}}(2t)^{l_{d+1}} \prod_{k=1}^{l_{d+1}} \norm{H_{X^{\cP_{d+1}}_k}} } \nonumber\\ 
    &\quad \lr{\sum_{\cP'_1:\mathsf{p}(Y_1)\rightarrow Y_1}(2t)^{l'_1} \prod_{k=1}^{l'_1} \norm{H_{X^{\cP'_1}_k}} }\cdots \lr{\sum_{\cP'_{l(Y)}:\mathsf{p}(Y_{l(Y)})\rightarrow Y_{l(Y)}}(2t)^{l'_{l(Y)}} \prod_{k=1}^{l'_{l(Y)}} \norm{H_{X^{\cP'_{l(Y)}}_k}} } \frac{N(\varGamma)}{\ell!}  \mathbb{I}(\ell>1).
\end{align}
In the second line, we have disintegrated each skeleton $\varGamma$ as its bifurcation factors $Y_n$ and branch irreducible paths $\cP_p,\cP'_n$. $\mathbb{I}(\cdot)$ is the indicator function that returns $1$ if the input is True and $0$ if False, so that only $\varGamma$ with $\ell=\ell(\varGamma)>1$ is included in the sum. The prefactor $(d+1)!$ comes from the number of sequences of faces $\{f_1,\cdots,f_{d+1}\}$ that a skeleton reaches in order. Note that the inequality is loose in general, because certain pairs of irreducible paths like $\cP'_n$ and $\cP_p$ that starts from $Y_n$, are allowed to intersect with each other here, while they are not allowed to intersect as two consecutive branches of the irreducible skeleton $\varGamma$.  In the third line of \eqref{eq:gamma<I}, we have moved factors around (each factor $X^\varGamma_k$ belongs to either an irreducible path $\cP_p$ or $\cP'_n$, or the bifurcation factors $Y$ ), and used \begin{equation}
    \ell = l(Y)+\sum_{p=1}^{d+1}l_p+\sum_{n=1}^{l(Y)}l'_n,
\end{equation}
where $l_p$ ($l'_n$) is the length of $\cP_p$ ($\cP'_n$). Now compare \eqref{eq:gamma<I} with the goal \eqref{eq:gamma<p}. The indicator function $\mathbb{I}(\ell>1)$ matches exactly the sum rule $l(Y)+l_1+\cdots+l_{d+1}+l'_1+\cdots+l'_{l(Y)}>1$ in \eqref{eq:gamma<p}, so \eqref{eq:gamma<I} implies \eqref{eq:gamma<p} as long as \begin{equation} \label{eq:NGamma<}
    \frac{N(\varGamma)}{\ell!} \le \lr{\prod_{p=1}^{d+1} \frac{1}{l_p!} } \lr{\prod_{n=1}^{l(Y)} \frac{1}{l'_n!} },
\end{equation}
because then the $\frac{1}{l_p!}$ ($\frac{1}{l'_n!}$) factor can be moved into the sum over $\cP_p$ ($\cP'_n$), and each sum over paths independently yields a $C_l$ in \eqref{eq:gamma<p}.

In the remaining proof, we verify \eqref{eq:NGamma<} by the recursion structure of $\varGamma$ as a tree. Any tree $\tilde{\varGamma}$ can be viewed as the root connecting to some subtrees $\tilde{\varGamma}_m$. For example, the tree $\varGamma$ in Fig.~\ref{fig:skeleton}(d) has three subtrees: a subtree $\tilde{\varGamma}_1$ with root $Y_1$, a trivial tree $\tilde{\varGamma}_2=f_2$, and a subtree $\tilde{\varGamma}_3$ with root $Y_2$. They are all connected to $Y_0$, the root of $\varGamma$. $\tilde{\varGamma}_1$ has a further subtree $\tilde{\varGamma}_4$ with root $Y_3$. Each subtree $\tilde{\varGamma}_m$ also corresponds to an irreducible skeleton (although it does not necessarily start from $S_0$), so it has length $\ell(\tilde{\varGamma}_m)$, and $N(\tilde{\varGamma})$ as the number of corresponding ordered irreducible skeletons. We first prove that, if a tree $\tilde{\varGamma}_0$ with root $\tilde{Y}_0$ is composed by subtrees $\tilde{\varGamma}_m: m=1,\cdots,m_0$, whose roots $\tilde{Y}_m$ are directly connected to $\tilde{Y}_0$ by irreducible paths $\tilde{\cP}_m$ of length $\tilde{\ell}_m$, then \begin{equation}\label{eq:N0<}
    \frac{N(\tilde{\varGamma}_0)}{\ell(\tilde{\varGamma}_0)!} \le \frac{1}{\ell(\tilde{\varGamma}_0)!} \frac{\lr{\sum_{m=1}^{m_0}\tilde{\ell}_m+\ell(\tilde{\varGamma}_m)}!}{\prod_{m=1}^{m_0}\lr{\tilde{\ell}_m+\ell(\tilde{\varGamma}_m)}!} \prod_{m=1}^{m_0} N(\tilde{\varGamma}_m) \le \lr{ \prod_{m=1}^{m_0}\frac{1}{\tilde{\ell}_m!}} \prod_{m=1}^{m_0}\frac{N(\tilde{\varGamma}_m)}{\ell(\tilde{\varGamma}_m)!} .
\end{equation}
The second inequality comes from $\ell(\tilde{\varGamma}_0)\ge \sum_{m=1}^{m_0}\tilde{\ell}_m+\ell(\tilde{\varGamma}_m)$, and \begin{equation}
    (a+b)!\ge a!\ b!,\quad \forall a,b>0.
\end{equation}
The first inequality of \eqref{eq:N0<} comes from the following. Starting from $\tilde{Y}_0$ there are $m_0$ branches, where factors of different branches can stagger in an ordered skeleton. Thus assuming the order within each branch is fixed, there are at most \begin{equation}\label{eq:tl'}
    \frac{\lr{\sum_{m=1}^{m_0} \tilde{\ell}'_m}!}{\prod_{m=1}^{m_0} \tilde{\ell}'_m!},
\end{equation}
number of ways to stagger the $m_0$ branches, where $\tilde{\ell}'_m=\tilde{\ell}_m+\ell(\tilde{\varGamma}_m)$ is the total number of factors in the $m$-th branch. The multinomial coefficient can be understood as the number of orders to place $m_0$ groups of balls in a line, where each group labeled by $m$ has $\tilde{\ell}'_m$ identical balls, while balls in different groups are distinguishable. The number of inter-branch orders may be strictly smaller than \eqref{eq:tl'}, because some ordering may produce a $\gamma$ that is not a skeleton. In Fig.~\ref{fig:skeleton}(a) for example, $X_5\in \cP_3$ must act after $X_3\in\cP_2$, since otherwise $f_2$ is already reached by $X_5$ and $X_3$ is no longer needed for a skeleton. 

Now we let the relative order within each branch to vary. The intra-branch order for the $m$-th branch has exactly $N(\tilde{\varGamma}_m)$ possibilities, since the irreducible path $\tilde{\cP}_m$ must act in a fixed order, followed by factors in $\tilde{\varGamma}_m$ that have $N(\tilde{\varGamma}_m)$ number of different orders. Because the intra-branch order is independent of the inter-branch order counted by \eqref{eq:tl'}, the total number of ordering is the product of them, and \eqref{eq:N0<} is established.

Finally, we show \eqref{eq:N0<} yields \eqref{eq:NGamma<}. Set $\tilde{\varGamma}_0=\varGamma$, we have \begin{align}
    \frac{N(\varGamma)}{\ell(\varGamma)!} &\le \lr{ \prod_{m_1\in \mathsf{st}(\varGamma)}\frac{1}{\tilde{\ell}_{m_1}!}} \prod_{m_1\in \mathsf{st}(\varGamma)} \frac{N(\tilde{\varGamma}_{m_1})}{\ell(\tilde{\varGamma}_{m_1})!} \nonumber\\
    &\le \lr{ \prod_{m_1\in \mathsf{st}(\varGamma)}\frac{1}{\tilde{\ell}_{m_1}!}} \prod_{m_1\in \mathsf{st}(\varGamma)}\mlr{ \lr{ \prod_{m_2\in \mathsf{st}(\tilde{\varGamma}_{m_1})}\frac{1}{\tilde{\ell}_{m_2}!}} \prod_{m_2\in \mathsf{st}(\tilde{\varGamma}_{m_1})} \frac{N(\tilde{\varGamma}_{m_2})}{\ell(\tilde{\varGamma}_{m_2})!} }\le \cdots \nonumber\\
    &\le \lr{\prod_{p=1}^{d+1} \frac{1}{l_p!} } \lr{\prod_{n=1}^{l(Y)} \frac{1}{l'_n!} },
\end{align}
where $m\in \mathsf{st}(\tilde{\varGamma})$ represents all subtrees $\tilde{\varGamma}_{m}$ connecting to the root of $\tilde{\varGamma}$, and we have used \eqref{eq:N0<} recursively to smaller and smaller subtrees. The iteration terminates when all subtrees $\tilde{\varGamma}$ become trivial: $\frac{N(\tilde{\varGamma})}{\ell(\tilde{\varGamma})!}=1$, and what remain are exactly one factor of either $\frac{1}{l_p!}$ or $\frac{1}{l'_n!}$ for each of the edges of the tree $\varGamma$.
\end{proof}

\subsection{Sum over irreducible paths: Lieb-Robinson bound with sub-exponential tails}
Next, we prove a Lieb-Robinson bound for $C_l(Z,Z^\prime)$ with a Hamiltonian $H = \sum H_S$, where $\lVert H_S\rVert \sim \mathrm{e}^{-r^\alpha}$ has subexponential decay. Amusingly, this goal is rather elegant to achieve, as it turns out that the sub-exponential decay $\ee^{-\kappa r^\alpha}$ with $\alpha<1$ is reproducing (a technical property explained below),\footnote{This appears related to the fact that power-law decaying interactions are reproducing.} while the faster decay with $\alpha=1$ is not.  
\begin{prop}[Lieb-Robinson bound for interactions decaying with a stretched exponential]\label{prop:reproducing}
The commutator of two local operators, one time-evolved, is bounded by the irreducible path summation defined in \eqref{eq:CZZ}: \begin{equation}\label{eq:comm<C}
    \norm{[\OO_Z(t), \OO'_{Z'}]} \le 2 \norm{\OO_Z}\norm{\OO'_{Z'}} \sum_{l\ge 0} C_l(Z,Z').
\end{equation}
Suppose the interaction decays as a stretched exponential, which is guaranteed by \begin{equation}
    \norm{H}_{\alpha,\kappa} < \infty.
\end{equation}
There exists  \begin{equation}\label{eq:lambda}
    \lambda = c_\lambda \kappa^{-\frac{d+1}{\alpha}},
\end{equation}
with $c_\lambda$ determined by $d$ and $\alpha$, such that
\begin{equation}\label{eq:C1+}
    \sum_{l\ge 1} C_l(Z,Z') \le \min(|\partial Z|, |\partial Z'|) \ee^{-\kappa \mathsf{d}^\alpha(Z,Z')} \lambda^{-1}\lr{\ee^{2\lambda t \norm{H}_\kappa} -1},
\end{equation}
and \begin{equation}\label{eq:C2+}
    \sum_{l\ge 2} C_l(Z,Z') \le \min(|\partial Z|, |\partial Z'|) \ee^{-\kappa \mathsf{d}^\alpha(Z,Z')} \lambda^{-1}\lr{\ee^{2\lambda t \norm{H}_\kappa} -1-2\lambda t \norm{H}_\kappa }.
\end{equation}
\end{prop}

\begin{proof}
\eqref{eq:comm<C} is proven in Theorem 3 of \cite{chen2021operator}, so we focus on \eqref{eq:C1+} and \eqref{eq:C2+}.
\eqref{eq:CZZ} can be bounded explicitly as \begin{equation}\label{eq:CZZ<}
    C_l(Z,Z') \le \frac{(2t)^l}{l!} \sum_{Z_1:Z_1 \text{ grows } Z} \norm{H_{Z_1}} \sum_{Z_2:Z_2 \text{ grows } Z_1} \norm{H_{Z_2}} \cdots \sum_{Z_l:Z_l \text{ grows } Z_{l-1}, Z_l\cap Z'\neq 0} \norm{H_{Z_l}},
\end{equation}
where ``$X \text{ grows } Y$'' means the condition $X\cap Y\neq \emptyset$ and $X\cap Y^c\neq \emptyset$ with $Y^c$ being the complement of $Y$ as a subset of $\Lambda$. We require $X\cap Y^c\neq \emptyset$ because otherwise $X$ is a factor of local rotation that does not appear in irreducible paths. Nevertheless, the path $(Z_1,\cdots,Z_l)$ in \eqref{eq:CZZ<} is not necessarily irreducible because $Z_3$ may touch $Z_1$ for example, which makes \eqref{eq:CZZ<} a loose bound in general.
Now we follow closely \cite{hastings2010locality}, starting at their Eq. (27).  To simplify the sums in (\ref{eq:CZZ<}) we may replace 
\begin{equation}\label{eq:grow<partial}
    \sum_{X:X \text{ grows } Y} \le \sum_{i\in \partial Y}\sum_{X\ni i},
\end{equation} 
 because $X$ must intersect with $\partial Y$ in order to grow $Y$ nontrivially.\footnote{Recall that $X$ is mandated to be a simplex, which is convex.  Thus this intersection must be non-vanishing.} Now consider, the last term in (\ref{eq:CZZ<}), which ends in:
 \begin{align} \label{eq:ZlZl1}
     \sum_{i \in \partial Z_{l-2}} &\sum_{Z_{l-1} \ni i: Z_{l-1} \cap Z_l \ne \emptyset }  \sum_{j \in \partial Z_{l-1}} \sum_{Z_l \ni j: Z_l\cap Z^\prime \ne \emptyset } \lVert H_{Z_{l-1}} \rVert\lVert H_{Z_l} \rVert \notag \\
     &\le \sum_{i \in \partial Z_{l-2}} \sum_{Z_{l-1} \ni i: Z_{l-1} \cap Z_l \ne \emptyset }  \sum_{j \in \partial Z_{l-1}} \sum_{Z_l \ni j } \lVert H_{Z_{l-1}} \rVert\lVert H_{Z_l} \rVert \mathrm{e}^{ - \kappa \mathsf{d}^\alpha(j,Z^\prime) + \kappa \mathrm{diam}(Z_l)^\alpha}\notag \\
     &\le \sum_{i \in \partial Z_{l-2}} \sum_{Z_{l-1} \ni i : Z_{l-1} \cap Z_l \ne \emptyset}\sum_{j \in \partial Z_{l-1}} \lVert H_{Z_{l-1}}\rVert \mathrm{e}^{-\kappa \mathsf{d}^\alpha(i,j) + \kappa \mathrm{diam}(Z_{l-1})^\alpha}\times    \mathrm{e}^{-\kappa \mathsf{d}^\alpha(j,Z^\prime) } \lVert H\rVert_\kappa,
 \end{align}
 where in the second line we used the identity $\mathsf{d}(j,Z^\prime) \le \mathrm{diam}(Z_l)$, and subsequently relaxed the restriction that $Z_l\cap Z^\prime $ was non-empty.  In the third line, we summed over $Z_l$, and \emph{also} began the process anew by including a similar identity $\mathsf{d}(i,j) \le \mathrm{diam}(Z_{l-1})$, since both $i,j\in Z_{l-1}$.  At this point, we will need to use the \emph{reproducibility} ansatz (which we will prove a little later):\begin{equation}\label{eq:KK<K}
    \sum_{m\in \Lambda} \mathcal{K}(\mathsf{d}(i,m) )\mathcal{K}(\mathsf{d}(m,j) ) \le \lambda \mathcal{K}(\mathsf{d}(i,j) ), 
\end{equation}
for some constant $\lambda$.  Combining (\ref{eq:ZlZl1}) and (\ref{eq:KK<K}): \begin{align}
    \sum_{i \in \partial Z_{l-2}} &\sum_{Z_{l-1} \ni i: Z_{l-1} \cap Z_l \ne \emptyset }  \sum_{j \in \partial Z_{l-1}} \sum_{Z_l \ni j: Z_l\cap Z^\prime \ne \emptyset } \lVert H_{Z_{l-1}} \rVert\lVert H_{Z_l} \rVert \notag \\
     &\le \sum_{i \in \partial Z_{l-2}} \sum_{Z_{l-1} \ni i : Z_{l-1} \cap Z_l \ne \emptyset} \lambda \mathrm{e}^{-\kappa \mathsf{d}^\alpha(i,Z^\prime) + \kappa \mathrm{diam}(Z_{l-1})^\alpha} \lVert H_{Z_{l-1}}\rVert \lVert H\rVert_\kappa \notag \\
     &\le  \sum_{i \in \partial Z_{l-2}} \lambda \mathrm{e}^{-\kappa \mathsf{d}^\alpha(i,Z^\prime) } \lVert H\rVert_\kappa^2.
\end{align}
Hence the process iterates $l$ times and we obtain (\ref{eq:C1+}) by resumming a simple exponential and subtracting off the constant term.  The surface factor $\min(|\partial Z|, |\partial Z'|)$ comes from our improvement \eqref{eq:grow<partial}, and the fact that irreducible paths from $Z$ to $Z'$ are exactly those from $Z'$ to $Z$, so we can pick whichever as the starting point. Furthermore, \eqref{eq:C2+} also follows by deleting the linear in $t$-term as well.


To complete the proof, it only remains to verify the reproducing property \eqref{eq:KK<K}. For any pair $i,j\in \Lambda$, we use a $d$-dimensional ``prolate spheroidal coordinate'' suited for the lattice $\Lambda$. Namely, for any $m\in \Lambda$, define its distance to the two ``focal points'' $i$ and $j$ as $r_i$ and $r_j$, respectively. Further define \begin{subequations}\begin{align}
        \sigma &= r_i+r_j,  \\
        r&= \min(r_i, r_j) \le \frac{\sigma}{2}.
    \end{align}
\end{subequations} 
Suppose there are $M(\sigma, r)$ vertices $m$ corresponding to a given $(\sigma,r)$. Then the left hand side of \eqref{eq:KK<K} becomes \begin{align}\label{eq:KK=}
    \sum_{m\in\Lambda} \ee^{-\kappa (r_i^\alpha + r_j^\alpha)} &= \sum_{r=0}^{\infty} \sum_{\sigma=\max(r_0,2r)}^{2r+r_0} M(\sigma, r) \ee^{-\kappa \mlr{r^\alpha + (\sigma-r)^\alpha}},
\end{align}
where $r_0=\mathsf{d}(i,j)$, and we have used the triangle inequalities $\sigma = r_i+r_j\ge r_0$ and $\sigma-2r=|r_i-r_j|\le r_0$ to constrain the sum. To proceed, observe that \begin{equation}\label{eq:<r+sigma}
    r^\alpha + (\sigma-r)^\alpha \ge \sigma^\alpha + (2-2^\alpha) r^\alpha.
\end{equation}
This equation comes from bounding the function \begin{equation}
    g(\rho):= 1+ (\rho-1)^\alpha -\rho^\alpha \ge g(2)= 2-2^\alpha,\quad \rho\equiv \frac{\sigma}{r} \ge 2,
\end{equation}
as one can verify by taking the derivative to confirm $g(\rho)$ is monotonic increasing. 

Combining \eqref{eq:KK=} with \eqref{eq:<r+sigma},
\begin{align}
    \sum_{m\in\Lambda} \ee^{-\kappa (r_i^\alpha + r_j^\alpha)}
    &\le \sum_{r=0}^{\infty} \sum_{\sigma=\max(r_0,2r)}^{2r+r_0} M(\sigma, r) \ee^{-\kappa \mlr{\sigma^\alpha + (2-2^\alpha) r^\alpha}} \nonumber\\
    &\le \sum_{r=0}^{\infty} \sum_{\sigma=\max(r_0,2r)}^{2r+r_0} 2c_d (r+1)^{d-1} \ee^{-\kappa \mlr{\sigma^\alpha + (2-2^\alpha) r^\alpha}} \nonumber\\
    &\le 2c_d \sum_{r=0}^{\infty} (2r+1) (r+1)^{d-1} \ee^{-\kappa \mlr{r_0^\alpha + (2-2^\alpha) r^\alpha}} \nonumber\\
    &\le c_\lambda \kappa^{-\frac{d+1}{\alpha}} \ee^{-\kappa r_0^\alpha}.
\end{align}
Here in the second line we have used \eqref{eq:cLam} that there are at most $c_d (r+1)^{d-1}$ sites that are distance $r$ to a given site $i$ or $j$ (each of which must be counted leading to an additional multiplicative factor of 2). In the third line we have used $\sigma\ge r_0$ in the exponent, and that there are at most $2r+1$ values for $\sigma$ to choose. We arrive at the last line analogously to \eqref{eq:sumx}, where $c_\lambda$ depends only on $d$ and $\alpha$. 

To conclude, we have verified $\mathcal{K}(r)= \ee^{-\kappa r^\alpha }$ is reproducing, with $\lambda$ given by \eqref{eq:lambda}. \eqref{eq:C1+} and \eqref{eq:C2+} then follow from (\ref{eq:grow<partial}) and (\ref{eq:KK<K}).
\end{proof}

\subsection{Proof of Lemma \ref{lem:K<K1}}

\begin{proof}[Proof of Lemma \ref{lem:K<K1}]
First, observe that a sufficient condition to prove \eqref{eq:K2+<} is if, 
for any $\OO=\OO_{S_0}$ with $\norm{\OO}=1$, \begin{equation}\label{eq:JS<}
    J(S) \le c_J(\delta\kappa)^{-\frac{2d}{\alpha}} t^2 \norm{H}_\kappa^2 \ee^{\frac{\delta\kappa}{2}s_0^\alpha} \ee^{-\tilde{\kappa}(s- s_0)^\alpha },
\end{equation}
where $s=\mathrm{diam}S,s_0=\mathrm{diam}S_0$, \begin{equation}
    \tilde{\kappa} = \kappa - \frac{\delta\kappa}{4}=\kappa'' + \frac{\delta\kappa}{4},
\end{equation}
and $c_J$ is a constant determined by $d,\alpha$ and $\kappa_1$.
The reason is that, using \eqref{eq:tH<C}, \begin{align}
    \sum_{\ell\ge 2} K_\ell(\OO) = \sum_{S\supset S_0}\ee^{\kappa''s^\alpha} J(S) &\le c_J(\delta\kappa)^{-\frac{2d}{\alpha}} t^2 \norm{H}_\kappa^2 \ee^{\frac{\delta\kappa}{2}s_0^\alpha} \sum_{S\supset S_0}\ee^{\kappa''s^\alpha} \ee^{-\tilde{\kappa}(s- s_0)^\alpha } \nonumber\\
    &\le c_J(\delta\kappa)^{-\frac{2d}{\alpha}} t^2 \norm{H}_\kappa^2 \ee^{\frac{\delta\kappa}{2}s_0^\alpha} \sum_{s>s_0}\tilde{c}_\Lambda (s-s_0)^d\ee^{\kappa''[s_0^\alpha+(s-s_0)^\alpha]} \ee^{-\tilde{\kappa}(s- s_0)^\alpha } \nonumber\\
    &\le c_J\tilde{c}'_\Lambda(\delta\kappa)^{-\frac{2d}{\alpha}} t^2  \norm{H}_\kappa^2 \ee^{\kappa s_0^\alpha} (\delta\kappa)^{-\frac{d+1}{\alpha}} \le c'_J t \norm{H}_\kappa \ee^{\kappa s_0^\alpha},
\end{align}
which implies \eqref{eq:K2+<} for $\OO$ with a general normalization $\norm{\OO}$.
In the second line, we have used the fact that for a given $S_0$ and $s$, there are at most $\tilde{c}_\Lambda (s-s_0)^d$ simplices $S\supset S_0$ that has diameter $s$. \eqref{eq:a+b<} is also used. In the third line we have summed over $s-s_0$, and used \eqref{eq:tH<C} to cancel extra powers of $(\delta\kappa)^{-1}$: observe that the form of \eqref{eq:tH<C} was chosen precisely so that this identity might hold. The final prefactor $c'_J$ then only depends on $d,\alpha$ and $\kappa_1$.

Thus, it remains to prove the condition \eqref{eq:JS<}. According to \eqref{eq:gamma<p}, we separate the sum over $Y\in \mathbf{Y}(S)$ into two cases: $l(Y)=0$ for case 1, where the $d+1$ paths to faces originate independently from $S_0$, and $l(Y)\ge 1$ for case 2, where $l(Y)$ is the number of nontrivial bifurcation factors defined in \eqref{eq:Y=Y0Y1}. For each case, we prove below that the contribution to $J(S)$ is bounded by the right hand side of \eqref{eq:JS<}, albeit with a case-dependent constant prefactor. We start with the simpler $l(Y)=0$ case, where some technical results are also used for the later case.

\textbf{Case 1:} If $l(Y)=0$, or equivalently $Y=\{S_0\}$, it suffices to bound $\mathcal{G}(Y)$ in the second line of \eqref{eq:gamma<p}, which describes  the $d+1$ faces of $S_0$ growing independently to the corresponding parallel faces of $S$, via irreducible paths $\cP_p: S_0\rightarrow f_p$. Here $\mathsf{p}(f_p)=S_0$ because $S_0$ is the parent of all faces.

If $S_0$ shares a simplex vertex with $S$, such that only one face $f_{d+1}$ of $S_0$ grows nontrivially, then $\mathcal{G}(Y)$ is bounded by \eqref{eq:C2+}: \begin{align}\label{eq:G11}
    \mathcal{G}(Y) = \sum_{l_{d+1}>1} C_{l_{d+1}}(S_0, f_{d+1}) &\le |\partial S_0| \ee^{-\kappa (s-s_0)^\alpha} \lambda^{-1}\lr{\ee^{2\lambda t \norm{H}_\kappa} -1-2\lambda t \norm{H}_\kappa } \nonumber\\
    &\le c_1 (\delta\kappa)^{-\frac{d-1}{\alpha}} \ee^{\frac{\delta\kappa}{2}s_0^\alpha} \ee^{-\kappa (s-s_0)^\alpha}\lambda^{-1}\lr{\ee^{2\lambda t \norm{H}_\kappa} -1-2\lambda t \norm{H}_\kappa } \nonumber\\ 
    &\le c_1 (\delta\kappa)^{-\frac{d-1}{\alpha}} \ee^{\frac{\delta\kappa}{2}s_0^\alpha} \ee^{-\kappa (s-s_0)^\alpha} c''_\lambda c_\lambda \kappa^{-\frac{d+1}{\alpha}} t^2 \norm{H}_\kappa^2\nonumber\\ 
    &\le c_1 c''_\lambda c_\lambda (\delta\kappa)^{-\frac{2d}{\alpha}} t^2 \norm{H}_\kappa^2 \ee^{\frac{\delta\kappa}{2}s_0^\alpha} \ee^{-\kappa (s-s_0)^\alpha}.
\end{align}
Here \eqref{eq:partial<} is used in the second line. In the third line, we have used that the argument of the exponential function is bounded by $\mathrm{O}(1)$: \begin{equation}\label{eq:lambdat<}
    2\lambda t \norm{H}_\kappa \le 2 c_\lambda \kappa^{\frac{3d+1-(d+1)}{\alpha}} \le 2 c_\lambda \kappa_1^{\frac{2d}{\alpha}},
\end{equation}
which holds by combining \eqref{eq:lambda}, \eqref{eq:tH<C} and \begin{equation}\label{eq:kap<kap}
    \delta\kappa<\kappa\le\kappa_1.
\end{equation}
The constant\begin{equation}\label{eq:c''lambda}
    c''_\lambda = \max_{\xi \le 2 c_\lambda \kappa_1^{2d/\alpha}} \frac{\ee^{2\xi} - 1-2\xi}{\xi^2} = \left.\frac{\ee^{2\xi} - 1-2\xi}{\xi^2}\right|_{\xi = 2 c_\lambda \kappa_1^{2d/\alpha}}
\end{equation} is completely determined by $d,\alpha$ and $\kappa_1$. We have also plugged in the definition of $\lambda$ from \eqref{eq:lambda}. In the last line of \eqref{eq:G11}, we have used \eqref{eq:kap<kap}.  Note that (\ref{eq:G11}) actually has a stronger exponential decay exponent ($\kappa$) than the $\tilde\kappa$ advertised in (\ref{eq:JS<}).

If $S_0$ does not share a simplex vertex with $S$ instead, there are $1<q\le d+1$ faces that grow nontrivially. The $\sum_p l_p>1$ summation condition in \eqref{eq:JS<} can be relaxed, so that the sum over $l_p$ are independent for different $p$. For each of the $q$ faces labeled by $p \in \{d+2-q,\cdots,d+1\}$, we invoke \eqref{eq:C1+} to get \begin{equation}\label{eq:c'lambda}
    \sum_{l\ge 1} C_l(S_0, f_p) \le |\partial S_0| \ee^{-\kappa \mathsf{d}^\alpha(S_0, f_p)} \lambda^{-1}\lr{\ee^{2\lambda t \norm{H}_\kappa} -1} \le |\partial S_0| \ee^{-\kappa \mathsf{d}^\alpha(S_0, f_p) } c'_\lambda t \norm{H}_\kappa,
\end{equation}
where we again used \eqref{eq:lambdat<} with $c'_\lambda$ determined by $d,\alpha$ and $\kappa_1$ in a similar way to \eqref{eq:c''lambda}. Combining all the $p$ faces, \begin{align}\label{eq:G1q}
    \mathcal{G}(Y) &\le \lr{c'_\lambda |\partial S_0| t \norm{H}_\kappa}^q \exp\mlr{-\kappa \sum_{p=d+2-q}^{d+1} \mathsf{d}^\alpha(S_0, f_p) } \nonumber\\ 
    &\le c_q (\delta\kappa)^{-q\frac{d-1}{\alpha}} \ee^{\frac{\delta\kappa}{2}s_0^\alpha} \lr{c'_\lambda t \norm{H}_\kappa}^q \exp\lr{-\kappa (s-s_0)^\alpha} \nonumber\\
    &\le c_q \lr{c'_\lambda}^q (\delta\kappa)^{-2\frac{d-1}{\alpha}} t^2 \norm{H}_\kappa^2 \ee^{\frac{\delta\kappa}{2}s_0^\alpha}\ee^{-\kappa (s-s_0)^\alpha} (\delta\kappa)^{(q-2)\frac{2d+2}{\alpha}}.
\end{align}
Here in the second line we have used \begin{equation}\label{eq:c_q}
    |\partial S_0|^q \le c_q (\delta\kappa)^{-q\frac{d-1}{\alpha}} \ee^{\frac{\delta\kappa}{2}s_0^\alpha},
\end{equation}
similar to \eqref{eq:partial<}. We have also used \begin{equation}
    \sum_{p=d+2-q}^{d+1} \mathsf{d}^\alpha(S_0, f_p) \ge (s-s_0)^\alpha,
\end{equation}
which follows from Proposition \ref{prop:simp}. In the third line of \eqref{eq:G1q} we have used \eqref{eq:tH<C}.


Summarizing Case 1, either \eqref{eq:G11} for $q=1$, or \eqref{eq:G1q} for $q\ge 2$, is bounded by the form of \eqref{eq:JS<}, considering $\delta\kappa < \kappa\le \kappa_1$.

\textbf{Case 2:} If $l(Y)\ge 1$, recall that besides paths $\cP_p$ from bifurcation factor $\mathsf{p}(f_p)$ to $f_p$, we also have paths $\cP'_n$ connecting $\mathsf{p}(Y_n)$ to $Y_{n}$, where $\mathsf{p}(Y_n)$ is the parent bifurcation factor of $Y_n$. Among them, at least one path must be nontrivial according to the sum rule in \eqref{eq:JS<} ($l^\prime_p>0$ and/or $l_p>0$) i.e., the path connects two sets $Z$ and $Z'$ that do not overlap. For each nontrivial path from $Z\rightarrow Z'$, the bound \eqref{eq:C1+} becomes \begin{equation}
    \sum_{l\ge 1}C_l(Z,Z') \le |\partial S| \ee^{-\kappa \mathsf{d}^\alpha(Z,Z') } c'_\lambda t \norm{H}_\kappa,
\end{equation}
following \eqref{eq:c'lambda}. For a trivial path with $Z\cap Z'\neq \emptyset$ instead, we have factor $\sum_{l\ge 0}C_l(Z,Z')=1=\ee^{-\kappa \mathsf{d}^\alpha(Z,Z') }$ from \eqref{eq:C0}. Thus we only need to keep track of the number of paths $1\le q(Y)\le 2d+1$ that are nontrivial, and do not need to track of which individual path is trivial or not, as if $\mathsf{d}(\mathsf{p}(f),f)=0$ (e.g.) it just contributes trivially to \begin{align}\label{eq:G2}
    \mathcal{G}(Y) &\le \lr{|\partial S|c'_\lambda t \norm{H}_\kappa }^{q(Y)} \exp\glr{-\kappa \mlr{\sum_{p=1}^{d+1}\mathsf{d}^\alpha(\mathsf{p}(f_p), f_p) + \sum_{n=1}^{l(Y)} \mathsf{d}^\alpha(\mathsf{p}(Y_n), Y_n)} }.
\end{align}

To proceed, we bound the exponent of \eqref{eq:G2} by \begin{equation}\label{eq:>sa-sn}
    \sum_{p=1}^{d+1}\mathsf{d}^\alpha(\mathsf{p}(f_p), f_p) + \sum_{n=1}^{l(Y)} \mathsf{d}^\alpha(\mathsf{p}(Y_n), Y_n) \ge (s-s_0)^\alpha - \sum_{n=1}^{l(Y)} s_n^\alpha,
\end{equation}
where \begin{equation}
    s_n :=\mathrm{diam}Y_{n}.
\end{equation}
\eqref{eq:>sa-sn} comes from applying \eqref{eq:a+b<} repeatedly to \begin{equation}\label{eq:>s-sn}
    \sum_{p=1}^{d+1}\mathsf{d}(\mathsf{p}(f_p), f_p) + \sum_{n=1}^{l(Y)} \mathsf{d}(\mathsf{p}(Y_n), Y_n) \ge s - \sum_{n=0}^{l(Y)} s_n,
\end{equation}
which is derived as follows. First, we introduce some notations. Let $\mathsf{ch}(Y_n)$ be the set of all children faces of $Y_n$ in the tree $\varGamma$, \begin{equation}
    \mathsf{ch}(Y_n) = \{f_p: \mathsf{p}(f_p)=Y_n\},
\end{equation}
and let $\mathsf{de}(Y_n)$ be the set of all descendant faces of $Y_n$. Namely, any $f_p\in \mathsf{de}(Y_n)$ is reached by an irreducible path $\varGamma_p$ that passes $Y_n$. Similarly, define $\mathsf{ch}'(Y_n)$ to be the set of all children bifurcation factors of $Y_n$: \begin{equation}
    \mathsf{ch}'(Y_n) = \{Y_{n'}\in Y: \mathsf{p}(Y_{n'})=Y_n\}.
\end{equation}
Observe the following geometric fact \begin{equation}\label{eq:d+s+d>d}
    \mathsf{d}(\mathsf{p}(Y_n), Y_n) + s_n + \sum_{f_p \in \mathsf{de}(Y_n)}\mathsf{d}(Y_n, f_p) \ge \sum_{f_p \in \mathsf{de}(Y_n)}\mathsf{d}(\mathsf{p}(Y_n), f_p).
\end{equation}
This is trivial for $\mathsf{de}(Y_n)= \{f_1,\cdots,f_{d+1}\}$, because then according to Proposition \ref{prop:simp}, the right hand side is not larger than $s$, while the left hand side is not smaller than $s$, which equals the sum of the last two terms. Otherwise if $\mathsf{de}(Y_n)\neq \{f_1,\cdots,f_{d+1}\}$, \eqref{eq:d+s+d>d} comes from triangle inequality \begin{equation}
    \mathsf{d}(\mathsf{p}(Y_n), Y_n) + s_n + \mathsf{d}( Y_{n}, W_n) \ge \mathsf{d}( \mathsf{p}(Y_n), W_n),
\end{equation}
with $W_n:=\bigcap_{f_p \in \mathsf{de}(Y_n)}f_p\neq \emptyset$, and \begin{equation}
    \mathsf{d}( Y_{n'},W_n) = \sum_{f_p \in \mathsf{de}(Y_n)} \mathsf{d}( Y_{n'}, f_p),
\end{equation}
due to the Manhattan metric of the simplicial lattice. 

\eqref{eq:d+s+d>d} gives a relation on the distance from a parent $\mathsf{p}(Y_n)$ to one of its children factors $Y_n$. If $Y_n$ contains further children factors $Y_{n'}$ -- i.e.  $\mathsf{de}(Y_n)\neq \mathsf{ch}(Y_n)$ -- then \eqref{eq:d+s+d>d} also gives a relation that connects the distances from $Y_n$ to $f_p\notin \mathsf{ch}(Y_n)$, with the distances from $Y_{n'}$s to these faces. Thus \eqref{eq:d+s+d>d} can be used recursively to eliminate all ``indirect'' distances $\mathsf{d}(Y_n, f_p)$ where $f_p\in \mathsf{de}(Y_n)\setminus \mathsf{ch}(Y_n)$, with the ``direct'' ones $f_p\in \mathsf{ch}(Y_n)$. To be specific, we start from the root $Y_0=S_0$, \begin{align}
    s-s_0 &= \sum_{f_p\in \mathsf{de}(Y_0)}\mathsf{d}(Y_0, f_p) = \sum_{f_p\in \mathsf{ch}(Y_0)}\mathsf{d}(Y_0, f_p)+ \sum_{Y_{n_1}\in \mathsf{ch}'(Y_0)} \sum_{f_p\in \mathsf{de}(Y_{n_1})} \mathsf{d}(Y_0, f_p) \nonumber\\
    &\le \sum_{f_p\in \mathsf{ch}(Y_0)}\mathsf{d}(Y_0, f_p)+ \sum_{Y_{n_1}\in \mathsf{ch}'(Y_0)}\mlr{ \mathsf{d}(Y_0, Y_{n_1}) + s_{n_1} + \sum_{f_p \in \mathsf{de}(Y_{n_1})}\mathsf{d}(Y_{n_1}, f_p)  } \nonumber\\
    &= \sum_{f_p\in \mathsf{ch}(Y_0)}\mathsf{d}(\mathsf{p}(f_p), f_p)+ \sum_{Y_{n_1}\in \mathsf{ch}'(Y_0)}\mlr{ \mathsf{d}(\mathsf{p}(Y_{n_1}), Y_{n_1}) + s_{n_1}  } + \sum_{Y_{n_1}\in \mathsf{ch}'(Y_0)} \sum_{f_p \in \mathsf{de}(Y_{n_1})}\mathsf{d}(Y_{n_1}, f_p) \nonumber\\
    &\le \sum_{f_p\in \mathsf{ch}(Y_0)}\mathsf{d}(\mathsf{p}(f_p), f_p)+ \sum_{Y_{n_1}\in \mathsf{ch}'(Y_0)}\mlr{ \mathsf{d}(\mathsf{p}(Y_{n_1}), Y_{n_1}) + s_{n_1}  } \nonumber\\
    &\quad + \sum_{Y_{n_1}\in \mathsf{ch}'(Y_0)} \glr{\sum_{f_p\in \mathsf{ch}(Y_{n_1})}\mathsf{d}(Y_{n_1}, f_p)+ \sum_{Y_{n_2}\in \mathsf{ch}'(Y_{n_1})} \mlr{\mathsf{d}(Y_{n_1}, Y_{n_2}) + s_{n_2} + \sum_{f_p \in \mathsf{de}(Y_{n_2})}\mathsf{d}(Y_{n_2}, f_p) } }
    &\le \cdots \nonumber\\
    &\le \sum_{p=1}^{d+1}\mathsf{d}(\mathsf{p}(f_p), f_p) + \sum_{n=1}^{l(Y)} \mlr{ \mathsf{d}(\mathsf{p}(Y_n), Y_n)+s_n }.
\end{align}
Here in the first line, we have used \eqref{eq:S0fp} together with the fact that any $f_p$ belongs to either $\mathsf{ch}(Y_0)$, or one of $\mathsf{de}(Y_{n_1})$ with $Y_{n_1}$ being one child of $Y_0$. I.e., $\{Y_{n_1}\}$ are the bifurcation factors at the second layer of the tree $\varGamma$ (the first layer being the root $Y_0$). In the second line we invoked \eqref{eq:d+s+d>d}. In the third line we identified $Y_0$ with nothing but the parent of its children. In the fourth line, we repeated the treatment in the first line for each $Y_{n_1}$, and used \eqref{eq:d+s+d>d} for each of the third-layer factors $Y_{n_2}$. Traversing down the tree $\varGamma$ recursively, we cover the distances corresponding to all edges in the tree exactly once, together with all the $s_n$s. As a result, we arrived at the final line and established \eqref{eq:>s-sn}, which further yields \eqref{eq:>sa-sn}.

Plugging \eqref{eq:>sa-sn} into \eqref{eq:G2},
\begin{align}\label{eq:G2<}
    \mathcal{G}(Y) &\le |S|^{-l(Y)} |S|^{l(Y)} \lr{|\partial S|c'_\lambda t \norm{H}_\kappa }^{q(Y)} \exp\glr{-\kappa \mlr{(s-s_0)^\alpha - \sum_{n=1}^{l(Y)} s_n^\alpha } } \nonumber\\
    &\le |S|^{-l(Y)} c_{l(Y),q(Y)} (c'_\lambda)^{q(Y)} (\delta\kappa)^{-\frac{dl(Y)+(d-1)q(Y)}{\alpha}} t\norm{H}_\kappa (\delta\kappa)^{[q(Y)-1]\frac{3d-1}{\alpha}}\ee^{\frac{\delta\kappa}{2}s_0^\alpha} \ee^{-\tilde{\kappa} (s-s_0)^\alpha } \ee^{\kappa \sum_{n=1}^{l(Y)} s_n^\alpha} \nonumber\\
    & \le c_{max} |S|^{-l(Y)} (\delta\kappa)^{-\frac{d-1+dl(Y)}{\alpha}} t\norm{H}_\kappa \ee^{\frac{\delta\kappa}{2}s_0^\alpha} \ee^{-\tilde{\kappa} (s-s_0)^\alpha } \ee^{\kappa \sum_{n=1}^{l(Y)} s_n^\alpha},
\end{align}
In the second line we have used a generalization of \eqref{eq:c_q}: \begin{equation}\label{eq:SpS}
    |S|^{l} |\partial S_0|^q \le c_{l,q} (\delta\kappa)^{-\frac{ld+q(d-1)}{\alpha}} \ee^{\frac{\delta\kappa}{4}(s^\alpha+s_0^\alpha)} \le c_{l,q} (\delta\kappa)^{-\frac{ld+q(d-1)}{\alpha}}\ee^{\frac{\delta\kappa}{4}[(s-s_0)^\alpha+ 2 s_0^\alpha]}.
\end{equation} and \eqref{eq:tH<C}. In the third line we have maximized the prefactor over $1\le q(Y)\le 2d+1$, using $\delta\kappa <\kappa\le \kappa_1$.

For a fixed vertex number $l+1$, there are at most $l!$ different trees. Thus for a given $l(Y)=l\le d$, plugging \eqref{eq:G2<} into \eqref{eq:JS<} yields \begin{align}
    J(S) &\le (d+1)! \ l!\  \sum_{Y_{1},\cdots,Y_{l}\subseteq S} (2t)^l\prod_{n=1}^l \norm{H_{Y_{n}}} \mathcal{G}(Y) \nonumber\\
    &\le (d+1)!\ l!\  c_{max} (\delta\kappa)^{-\frac{d-1+dl}{\alpha}} (2t)^l t \norm{H}_\kappa  \ee^{\frac{\delta\kappa}{2}s_0^\alpha} \ee^{-\tilde{\kappa} (s-s_0)^\alpha } \sum_{i_1,\cdots,i_l\in S} |S|^{-l} \prod_{n=1}^l \lr{ \sum_{Y_{n}\ni i_n}\norm{H_{Y_{n}}} \ee^{\kappa s_n^\alpha}}  \nonumber\\
    &\le (d+1)!\ l!\  c_{max} (\delta\kappa)^{-\frac{2d-1+\mathsf{d}(l-1)}{\alpha}} (2t)^l t \norm{H}_\kappa  \ee^{\frac{\delta\kappa}{2}s_0^\alpha} \ee^{-\tilde{\kappa} (s-s_0)^\alpha } \norm{H}_\kappa^l \nonumber\\
    &\le c'_{max} (\delta\kappa)^{-\frac{2d-1}{\alpha}} t^2 \norm{H}_\kappa^2 \ee^{\frac{\delta\kappa}{2}s_0^\alpha} \ee^{-\tilde{\kappa} (s-s_0)^\alpha }.
\end{align}
Here in the first line, we sum over $Y\in \mathbf{Y}(S)$ by first summing over the topological tree structure that gives the prefactor $l!$, and then summing over each $Y_n\in Y$. In the second line, we overcount by replacing each $\sum_{Y_{n}\subseteq S}$ by $\sum_{i_n\in S} \sum_{Y_{n}\ni i_n}$.   While this counts the same $Y_n$, $|Y_n|$ times, we can use the $\kappa$-norm of $H$ to sum over $Y_{n}$ efficiently. Furthermore, the sum over sites is canceled by $|S|^{-l}$ from $\mathcal{G}(Y)$. In the last line, we have used $(2t\norm{H}_\kappa)^{l-1}$ to kill $(\delta\kappa)^{-\frac{\mathsf{d}(l-1)}{\alpha}}$ according to \eqref{eq:tH<C}, and finally arrived at the form of \eqref{eq:JS<}.
\end{proof}

\section{False vacuum decay in models with quasiparticles}
Finally, let us discuss some implications of Theorem \ref{thm1} for the low-energy dynamics of perturbed gapped phases with ``quasiparticle excitations". This is the only section in this Supplement that is not mathematically rigorous, but is still based on physical assumptions. 

We focus on the setting of false vacuum decay where $H_0$ has a discrete symmetry that is (in the thermodynamic limit) \emph{spontaneously} broken in its ground states, while the perturbation $V$ \emph{explicitly} breaks that symmetry. A $1$d spin example is given in the main text, which we reproduce here: \begin{subequations}\label{eq:H0V_supp}\begin{align}
    H_0 &= \sum_{i=1}^{N-1}\lr{ Z_iZ_{i+1}+J_x X_i X_{i+1}}+h\sum_{i=1}^N X_i , \label{eq:H0=Ising} \\
    V &= \epsilon Z= \epsilon \sum_{i=1}^N Z_i, \label{eq:V=Z_supp}
\end{align}\end{subequations}
where we choose $h=0.9, J_x=0.37$. If $J_x=0$, $H_0$ is the transverse-field Ising model with two ferromagnetic ground states, separated from the excited states by a gap $2(1-h)\approx 0.2$. The $J_x$ term is added to break the integrability of $H_0$, but keeps $H_0$ in the gapped ferromagnetic phase. One ground state of $H_0$,  which we call $|\psi_\uparrow\rangle $, has positive polarization $\langle Z_i\rangle >0$. When turning on $V$, $|\psi_\uparrow\rangle $ merges into the excitation spectrum by gaining an extensive amount of energy $\sim \epsilon N$, and is called the false vacuum. In the main text, we show numerically that polarization evolving from $|\psi_\uparrow\rangle$ by Hamiltonian $H$ remains large in a long-time window, which means the false vacuum takes a long time to decay to the true vacuum that prefers a negative polarization. Very similar dynamics is also observed for the product initial state $|\uparrow\cdots \uparrow\rangle$. 
In the following we explain this observed non-thermal dynamics with the rigorous Theorem \ref{thm1} and physical assumptions. 

First, with no assumptions, the third point of Corollary \ref{cor1} establishes that the initial state $U\ket{\psi_\uparrow}$ is a frozen state before $t_*$ (i.e., does not evolve) when probed by local observables. This state $U\ket{\psi_\uparrow}$ is then a ``prethermal scar state'' with rigorous bounds on thermalization time. Note that it is short-range entangled in the sense that it can be connected to a product state $\ket{\uparrow\cdots \uparrow}$ by a quasi-local unitary $UU_0$, where $U_0$ is the quasi-local unitary connecting $\ket{\psi_\uparrow}$ and $\ket{\uparrow\cdots \uparrow}$ since they are in the same phase. However, it is practically difficult to precisely prepare this state $U\ket{\psi_\uparrow}$, because the Schrieffer-Wolff unitary $U$ is defined order-by-order and has no closed-form expression, as we showed in our proof above.

Nevertheless, the dynamics of atypical large polarization is robustly observed for other initial states which we denote by $\ket{\psi_0}$ (e.g. $\ket{\psi_0}=\ket{\psi_\uparrow}$ or $\ket{\psi_0}=\ket{\uparrow\cdots\uparrow}$), where Corollary \ref{cor1} (third point) does not apply.  To explain the slow decay of polarization in this state, we use the fact that $D_*$ in \eqref{eq:UHU=} is $\Delta$-diagonal, beyond just preserving the ground subspace used in Corollary \ref{cor1}. We also need two physical assumptions. First, we assume that the Schrieffer-Wolff rotated initial state $U^\dagger \ket{\psi_0}$ is of low energy density measured by $H_0$, comparing to the ground state $\ket{\psi_\uparrow}$. This condition is easily satisfied by a large family of states. For example, it holds for $\ket{\psi_0}=\ket{\psi_\uparrow}$, because the Schrieffer-Wolff unitary $U$ only involves small $\sim \epsilon$ local ``rotations'', so the energy density of $U^\dagger \ket{\psi_0}$ is strictly bounded by $\epsilon$. This condition also holds conceptually for $\ket{\psi_0}=\ket{\uparrow\cdots \uparrow}$, because it is in the same phase of and thus close to the previous state $\ket{\psi_\uparrow}$. See \cite{falseVac_quasip} for an explicit verification for this product state in the mixed-field Ising chain.

The second assumption we will need is that $H_0$ hosts well-defined quasiparticles in its low energy spectrum. Then the rotated initial state $U^\dagger \ket{\psi_0}$, which has low energy density by assumption, contains a small density of quasiparticles near the bottom of the quasiparticle band. The gap $\Delta$ of $H_0$ is then just the ``mass'' of the quasiparticles. Before time $t_*$, the Hamiltonian in the rotated frame is effectively $H_0+D_*$, so the question is whether $D_*$ is able to create more and more quasiparticles from the initial state: If the quasiparticles proliferate then the state has decayed to the true vacuum. However, this is forbidden exactly by our Theorem \ref{thm1}, which guarantees $D_*$ is $\Delta$-diagonal! More precisely, since each local term in $D_*$ only acts on a finite number of  quasiparticles,  it cannot increase the quasiparticle number because there is not enough energy $\Delta$ to create an extra one.\footnote{This argument is not rigorous because we cannot rule out the possibility that $\ket{\psi_0}$ contains extremely energetic quasiparticles, whose kinetic energy $\gg \Delta$ might be large enough to spawn a new quasiparticle under dynamics from $D_*$ (by sacrificing kinetic energy $\Delta$ to create a new quasiparticle).  However, on physical grounds, this scenario seems unlikely.} Then we know the quasiparticles cannot proliferate, so that the state remains close to the false vacuum before $t_*$. This quasiparticle argument appears in \cite{falseVac_quasip} using second order perturbation theory, where the authors have to conjecture the emergent quasiparticle number conservation for higher orders. This conjecture is supported by our Theorem \ref{thm1}. Moreover, previous works exclusively deal with integrable $H_0$ (e.g. setting $J_x=0$ in \eqref{eq:H0=Ising}) in order to carry out perturbation theory explicitly at low order; here we argue that the integrability condition is not crucial, as long as the low-energy spectrum looks ``integrable'' in the sense of containing quasiparticles, which holds for many physical models.

\bibliography{preth}

\end{document}